\RequirePackage[hyphens]{url}
\documentclass[mnsc]{informs3h}
\OneAndAHalfSpacedXI

\usepackage[round]{natbib}
\usepackage{multibib}
\def\bibfont{\small}

\newcites{App}{References for Appendix}

\usepackage{secdot}
\usepackage{epsfig}
\usepackage{amssymb}
\usepackage{amsmath}
\usepackage{color}
\usepackage{graphicx}
\usepackage{multirow}
\usepackage{enumerate}
\usepackage{subfig}
\usepackage{natbib}

 \bibpunct[, ]{(}{)}{,}{a}{}{,}%
 \def\bibfont{\small}%
\usepackage[ruled,linesnumbered]{algorithm2e}
\usepackage[normalem]{ulem}
\usepackage{float}
\usepackage{enumitem}
\usepackage{amsmath}
\usepackage{epstopdf}
\usepackage{graphicx}
\usepackage{tikz}
\usepackage{pgfplots}
\usepackage{subfig}
\usepackage{tabularx}

\usepackage{array}
\pgfplotsset{compat=newest}
\usepackage{pgfplotstable}
\usetikzlibrary{arrows.meta,babel,calc,intersections,backgrounds,patterns,plotmarks,decorations.pathreplacing}
\usepackage{multirow}
\usepgfplotslibrary{fillbetween}
\tikzstyle{every picture} += [>=stealth]
\makeatother
\usepackage{lipsum}
\usepackage[justification=centering]{caption}
\usepackage{booktabs}
\usepackage[hyphens]{url}
\usepackage{mathtools}
\newcommand*{\QEDA}{\hfill\ensuremath{\blacksquare}}%

\TheoremsNumberedThrough    
\ECRepeatTheorems
\EquationsNumberedThrough  
\makeatletter
\newcommand\footnoteref[1]{\protected@xdef\@thefnmark{\ref{#1}}\@footnotemark}
\makeatother
\usepackage{apptools}
\usepackage{bbm}
\usepackage{chngcntr}
\usepackage{lipsum}
\usepackage{float}
\usepackage{hyperref}
\usepackage{cleveref}
\Crefname{line}{}{}
\usepackage{bm}
\hypersetup{
    hidelinks,
    colorlinks=true,
    linkcolor=black,
    citecolor=blue
} 

\usepackage{enotez}
\let\footnote=\endnote
\setenotez{backref = true}
\setenotez{list-name = {Endnotes}}
\usepackage[hang,flushmargin]{footmisc}
\makeatletter
\def\footnoterule{\relax
  \kern 1pt
  \hbox to \columnwidth{\vrule width 0.5\columnwidth height 0.5pt\hfill}
  \kern 1pt}
\makeatother

\tikzset{
 hatch distance/.store in=\hatchdistance,
 hatch distance=10pt,
 hatch thickness/.store in=\hatchthickness,
 hatch thickness=2pt
 }
 \makeatletter
 \pgfdeclarepatternformonly[\hatchdistance,\hatchthickness]{flexible hatch}
 {\pgfqpoint{0pt}{0pt}}
 {\pgfqpoint{\hatchdistance}{\hatchdistance}}
 {\pgfpoint{\hatchdistance-1pt}{\hatchdistance-1pt}}%
 {
 \pgfsetcolor{\tikz@pattern@color}
 \pgfsetlinewidth{\hatchthickness}
 \pgfpathmoveto{\pgfqpoint{0pt}{0pt}}
 \pgfpathlineto{\pgfqpoint{\hatchdistance}{\hatchdistance}}
 \pgfusepath{stroke}
 }
\makeatother
\usetikzlibrary{shapes.geometric, arrows, calc, shapes.misc, matrix, shapes, arrows, calc, positioning, shapes.geometric, shapes.symbols, shapes.misc, intersections, chains, scopes}
\tikzstyle{node1} = [circle, circle sides=6,minimum width=1.25cm, text badly centered, text width=2.35em,minimum height=1.15cm, draw=black, font=\small]
\tikzstyle{node2} = [rectangle,rounded corners, minimum width=1.25cm, minimum height=1cm, text-centered, draw=black, font=\small]
\tikzstyle{arrow} = [thick,->,>=stealth]

\newcommand{\E}{\mathbb{E}}
\newcommand{\Var}{\textbf{Var}}

\newcommand{\ALG}{\texttt{ALG}}
\newcommand{\Cov}{\operatorname{Cov}}
\newcommand{\tr}{\operatorname{tr}}
\usepackage[margin=1in]{geometry}
\usepackage{amsmath, amssymb,  mathtools}
\usepackage{bm}
\usepackage{bbm}        
\usepackage{hyperref}

\newenvironment{nassumption}[1][]{%
  \trivlist
  \item[\hspace*{1em}\hskip0.5em{\mdseries\scshape Assumption%
    \if\relax\detokenize{#1}\relax\else\ {\normalfont\bfseries(#1)}\fi.}]%
  \normalfont\it\ignorespaces
}{%
  \endtrivlist
}

\newcommand{\R}{\mathbb{R}}

\newcommand{\Z}{\mathbb{Z}}
\newcommand{\obs}{\mathrm{obs}}
\newcommand{\1}{\mathbbm{1}}
\newcommand{\PP}{\mathbb{P}}
\newcommand{\diag}{\mathrm{diag}}

\newcommand{\cF}{\mathcal{F}}
\newcommand{\cT}{\mathcal{T}}
\newcommand{\cX}{\mathcal{X}}
\newcommand{\cA}{\mathcal{A}}
\newcommand{\cY}{\mathcal{Y}}
\newcommand{\cM}{\mathcal{M}}
\newcommand{\cH}{\mathcal{H}}
\newcommand{\cE}{\mathcal{E}}
\newcommand{\cN}{\mathcal{N}}

\newcommand{\norm}[1]{\left\|#1\right\|}

\newcommand{\K}{K}
\newcommand{\A}{A}
\newcommand{\X}{X}
\newcommand{\Y}{Y}
\newcommand{\Oo}{O}

\newcommand{\e}{e}              
\newcommand{\ebar}{\bar e}      

\newcommand{\IF}{\varphi}       

\makeatletter
\def\theARTICLETOP{\vspace*{-57pt}}
\makeatother

\begin{document}

\RUNAUTHOR{Li et al.}
\RUNTITLE{Semiparametric Efficiency in Sequential Experiments: Characterization and Design via Average Propensity}
\TITLE{Semiparametric Efficiency in Sequential Experiments: Characterization and Design via Average Propensity}

\ARTICLEAUTHORS{
 \AUTHOR{Jiachun Li}
 \AFF{Laboratory for Information and Decision Systems, MIT, \url{jiach334@mit.edu}}
 \AUTHOR{David Simchi-Levi}
 \AFF{Laboratory for Information and Decision Systems, MIT, \url{dslevi@mit.edu}}
}

\ABSTRACT{%
\textbf{Abstract.}
Modern experiments, including evaluations of AI-enabled services and platform interventions, often depart from independent and identically distributed (i.i.d.) sampling because assignments may be adaptive, balanced across covariates, or subject to rollout constraints such as exposure, fairness, and budget limits. This paper studies the efficiency benchmark for estimating causal targets in such sequential experiments. We show that every non-anticipating design induces an average propensity score, and we establish a semiparametric lower bound: for regular locally unbiased estimators, attainable precision is bounded by the i.i.d. efficiency benchmark evaluated at this induced score. The average propensity score thereby serves as a common benchmark and design target, allowing sequential experimental design to be viewed as choosing or learning an efficient allocation rule, with operational constraints entering through the admissible set when present. We then develop implementable batched adaptive designs that approach this benchmark through two complementary mechanisms. The first uses regression adjustment based on efficient influence functions; for general smooth estimands it attains the benchmark under standard nuisance-rate conditions, while for linear functionals of outcome means it achieves a sharp second-order rate. The second uses adaptive covariate balancing to attain the same benchmark through the assignment mechanism, enabling simple moment-based estimation. Both routes require only a small number of policy updates, making them compatible with delayed feedback and easier to monitor in operational deployments. Numerical experiments and an empirical study of AI medical-assistant evaluation demonstrate the practical efficiency gains, including in multi-treatment settings. Overall, the paper provides a unified framework for characterizing and designing efficient sequential experiments.
}
\maketitle
\vspace{-2mm}

\section{Introduction}

In the era of large-scale AI deployment, platforms and service providers continuously introduce new AI-enabled services, agentic workflows, and intervention policies. Many of these systems remain difficult to evaluate through mechanistic reasoning alone: their real-world effects depend not only on complex model behavior but also on hard-to-predict human responses, including phenomena such as algorithm aversion \citep{dietvorst2015algorithm} and performance changes under sustained AI exposure \citep{budzyn2025deskilling}. As a result, randomized experiments have become an essential paradigm for causal evaluation. For example, recent large-scale efforts to assess the reliability of LLM-based medical assistants for the general public rely on randomized preregistered studies to establish safety and efficacy prior to broader deployment \citep{bean2026reliability}. More broadly, experimentation allows treatment effects to be identified with minimal modeling assumptions even when the underlying system is opaque \citep{gilotte2018offline}.

\begin{figure}[htbp]
    \centering
    \includegraphics[width=0.7\linewidth]{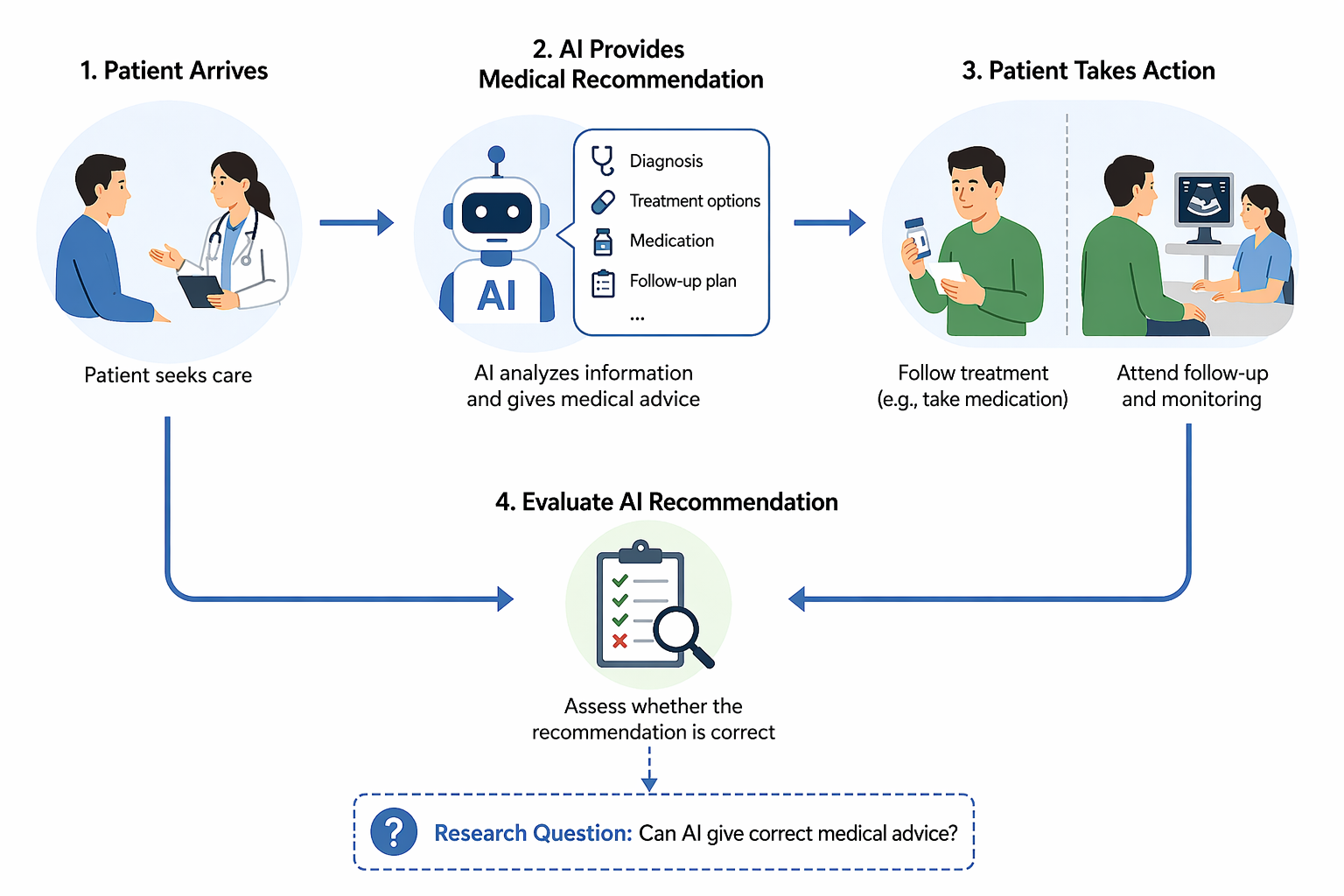}
    \caption{An illustration of the need to evaluate AI medical assistants prior to deployment.}
    \label{fig:illustration}
\end{figure}

In practice, however, such experiments rarely conform to classical independent and identically distributed (i.i.d.) sampling. Modern deployments often use \emph{adaptive assignment} to improve statistical precision, \emph{covariate balancing} to stabilize comparisons across heterogeneous populations, and rollout restrictions such as exposure limits, treatment budgets, or subgroup fairness requirements. These features create dependence across observations and make classical i.i.d.-based efficiency theory insufficient as a guide for experimental design.

Consequently, experimental design emerges as a primary statistical and operational decision. Experimental traffic, eligible users, and evaluation time are limited resources, and different treatment variants may provide different amounts of information about the target estimand. At the same time, highly reactive per-unit adaptation can be operationally costly, difficult to monitor, and statistically sensitive to noisy short-run outcomes. In platform experimentation and high-stakes settings such as clinical evaluation, outcome feedback may also arrive with delay, making continuous policy updates impractical. These considerations motivate \emph{batched adaptation}: the assignment policy is updated only a small number of times, allowing experimenters to maintain a rollout that is easier to implement and audit, more stable under noisy outcomes, and naturally compatible with delayed feedback. Batched designs are also analytically useful because observations remain conditionally i.i.d.\ within each batch.

These considerations lead to the central inquiry of this paper: \textit{what is the efficiency benchmark for estimating a target causal estimand in a sequential experiment, and how can an implementable design approach it in practice?}

Our response is organized around the \emph{induced average propensity score}. We show that every non-anticipating design induces such an object, and we establish a semiparametric lower bound showing that regular locally unbiased estimators are subject to the i.i.d.\ efficiency benchmark evaluated at this induced average propensity. The average propensity score therefore serves as a common benchmark and design target for sequential experiments. Operational restrictions such as exposure, budget, or fairness constraints can be encoded through the admissible set of average propensities when present, while the central theory focuses on identifying and approaching the corresponding benchmark. We then establish that, in standard and practice-relevant settings, this benchmark can be approached constructively by batched adaptive designs through two complementary routes: regression adjustment and covariate balancing. Together, these results provide a unified and practically relevant framework for efficient multi-treatment experimentation beyond classical i.i.d.\ randomization.
We summarize our contributions as follows.

\begin{figure}
\centering
\includegraphics[width=0.7\linewidth]{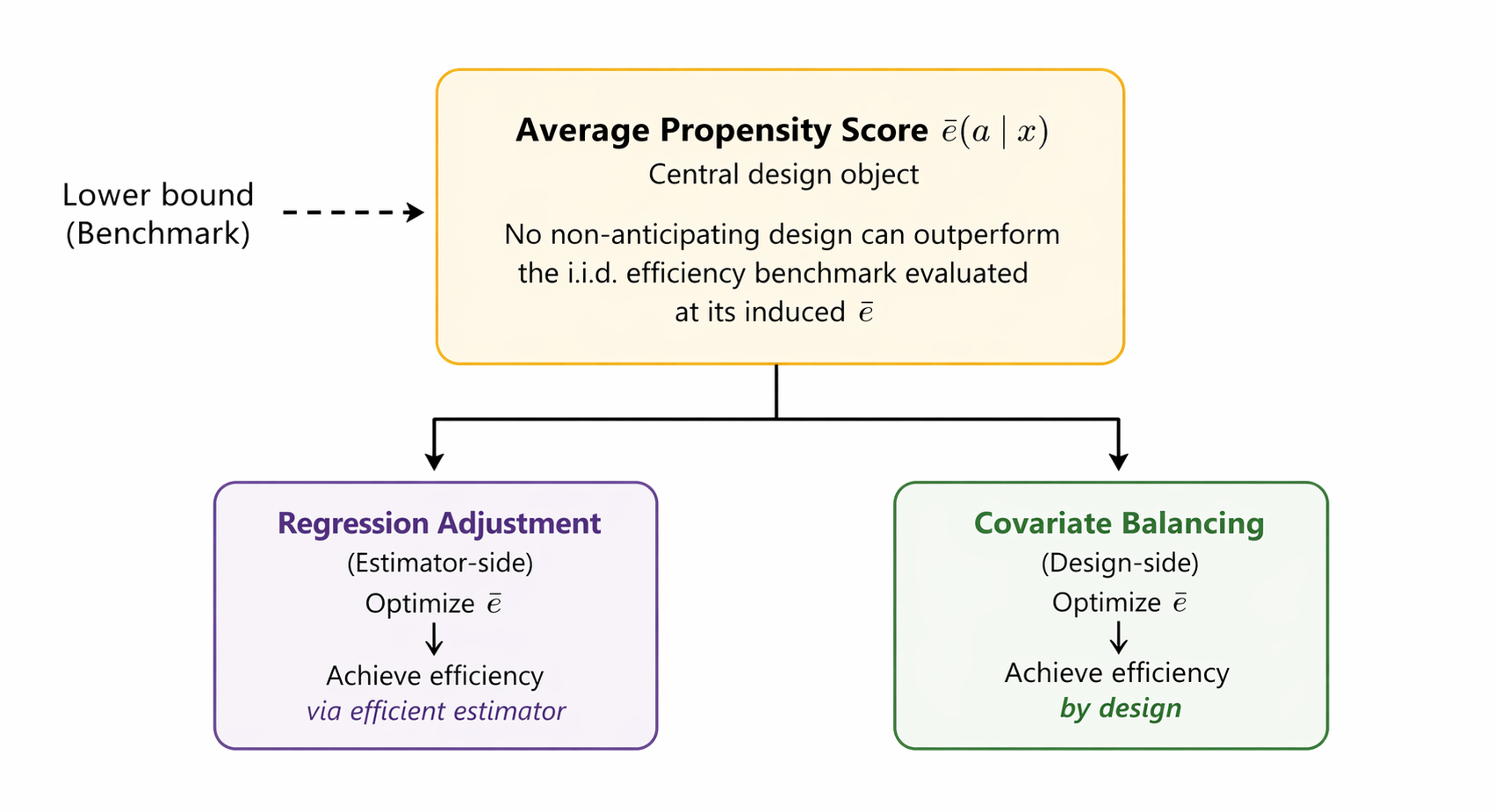}
\caption{A roadmap of the main results.}
\label{fig:contribution}
\end{figure}

\textbf{1. Characterizing the efficiency benchmark via average propensity.}
For a broad class of non-anticipating experimental designs, we establish that the induced average propensity score is the central object governing first-order efficiency. Specifically, we prove a semiparametric lower bound showing that regular locally unbiased estimators are subject to the efficiency benchmark of an i.i.d.\ experiment evaluated at the induced average propensity. This characterization places several previously separate settings including adaptive, finely stratified, and local-randomization designs under the same framework by showing that their statistical limits are governed by the same induced-propensity principle. It also recasts experimental design as the problem of choosing, or in adaptive settings learning, an efficient average propensity, with application-specific constraints incorporated through the admissible design class when present. For smooth estimands whose efficient influence functions are well defined, the corresponding optimality condition takes the form of a generalized Neyman allocation, recovering the classical Neyman allocation for the ATE as a special case. The framework also extends naturally from scalar targets to low-dimensional vector parameters: under $L_2$ risk, the optimal design is characterized by the conditional second-moment matrices of the efficient influence function, yielding a multivariate analogue of Neyman allocation.

\textbf{2. Approaching the benchmark via batched adaptive designs.}
Having characterized the benchmark, we show that it can be approached constructively in standard, practice-relevant settings using batched adaptive designs. Our procedures divide the experiment into a logarithmic number of batches and update the assignment policy only at batch boundaries. This batched structure has operational advantages: it requires few policy updates, is easier to implement and audit, is more stable under noisy outcomes, and accommodates delayed feedback, while simultaneously restoring a conditional i.i.d.\ structure within each batch. Building upon this shared principle, we develop two complementary routes to the same benchmark, each with different practical tradeoffs:

\begin{itemize}[leftmargin=2em]
\item \textbf{2.1 Regression adjustment.}
Our first route uses batched adaptive regression adjustment based on the efficient influence function (EIF). At a high level, the design assigns treatments according to an estimated optimal propensity within each batch and updates this target propensity only at batch boundaries using accumulated historical data. For general smooth estimands, we prove that under the classical nuisance-estimation rate of $n^{-1/4}$, this procedure approaches the oracle semiparametric benchmark. For the important special case of linear functionals of the outcome mean---including average treatment effects and policy evaluation---we obtain a sharper refinement: the $n^{-1/4}$ bottleneck vanishes, we explicitly characterize the second-order rate at which the efficiency benchmark is approached, and we establish a matching lower bound proving that this rate is minimax-sharp up to logarithmic factors.

\item \textbf{2.2 Covariate balancing.}
Our second route approaches the same benchmark \emph{by design}, through batched adaptive covariate balancing. This route is motivated by settings in which flexible nuisance estimation is difficult or undesirable, and it uses a simple, moment-based estimator in the final evaluation step. Compared with the existing covariate-balancing literature, our setting is more general: covariates arrive sequentially, treatment must be assigned immediately, the optimal target propensity is unknown and must be learned online across batches, and the resulting efficiency guarantees cover a broad class of moment-defined estimands, with the average treatment effect as a special case. By connecting balancing mechanisms to online discrepancy control, we construct a batched sequential design that drives the balancing remainder to zero within each batch while simultaneously learning the optimal average propensity across batches. As a result, the optimal semiparametric benchmark is attained through assignment-side control, bypassing the $n^{-1/4}$ nuisance-rate bottleneck that appears in the general regression-adjustment route.
\end{itemize}
The rest of the paper is organized as follows. Section~\ref{subsec:lit_review} reviews the related literature. Section~\ref{sec:lower-bound} introduces the induced average propensity score and establishes the main benchmark result of the paper: for regular locally unbiased estimators under any non-anticipating experiment, possibly with dependent assignments, the achievable precision is bounded by the semiparametric efficiency benchmark of an i.i.d.\ experiment evaluated at the corresponding average propensity score. Sections~\ref{sec:upper-bound} and~\ref{sec:cov-balancing} then develop two complementary constructive routes to approach this benchmark using batched adaptive designs. Section~\ref{sec:upper-bound} studies the estimator-side route based on batched adaptive regression adjustment, and shows that under suitable conditions this procedure learns the target propensity and approaches the oracle benchmark. Within this route, Section~\ref{subsec:linear-case} focuses on the important special case of linear functionals of the outcome mean, including weighted ATEs and policy value, and provides an explicit characterization of the rate at which the efficiency benchmark is approached, together with a matching lower bound. Section~\ref{sec:cov-balancing} studies the design-side route based on batched adaptive covariate balancing, showing that the same semiparametric benchmark can be attained by design using a simple moment-based estimator. Finally, Section~\ref{sec:numerical} reports simulation results and empirical evidence based on AI medical-assistant evaluation.

\subsection{Literature Review} \label{subsec:lit_review}

Our paper is related to three strands of literature: semiparametric efficiency beyond i.i.d.\ experiments, adaptive experimentation for efficient estimation, and covariate balancing as a route to efficiency by design. Across these strands, our central theme is that the induced average propensity score provides a common benchmark and design object for sequential experiments.

\noindent \textbf{Semiparametric efficiency beyond i.i.d.\ experiments.}
Semiparametric efficiency under independent and identically distributed (i.i.d.) sampling is classical in statistics and econometrics \citep{newey1994asymptotic,van2000asymptotic}. Beyond i.i.d.\ experiments, however, existing results are largely organized around specific non-i.i.d.\ design classes. For example, \cite{armstrong2022asymptotic} derives efficiency bounds for a broad class of unconstrained adaptive sequential experiments; \cite{bai2023efficiency} and \cite{rafi2023efficient} study finely stratified and covariate-adaptive designs under fixed propensity rules; and \cite{cytrynbaum2021optimal} analyzes optimal allocation in two-stage sampling-and-assignment designs. While these papers are closely related in spirit, they operate under different assumptions on the assignment mechanism and do not subsume one another. Our paper provides a unified lower-bound perspective encompassing these previously separate settings. By identifying the induced average propensity score as the central benchmark object, we develop a common efficiency characterization for non-anticipating designs, and in particular resolve the conjectured efficiency characterization in the two-stage design of \cite{cytrynbaum2021optimal}.

\noindent \textbf{Adaptive experimentation for efficient estimation.}
The literature on response-adaptive randomization studies how treatment allocation ratios can be learned and updated during the experiment \citep{hu2006theory}. More recent work has focused on adaptive experimentation for efficient causal estimation, often with the goal of approaching Neyman-type allocations that minimize the asymptotic variance of the average treatment effect (ATE) \citep{kato2020efficient,zhao2023adaptive,dai2023clip,li2024optimal,cook2024semiparametric}. Most of this literature, however, is centered on the ATE, binary treatments, or parameter-specific allocation rules. Our paper extends this line of work in three directions. First, we move from the ATE to general smooth estimands and show that the optimal design admits a generalized Neyman-allocation interpretation, with treatment allocation driven by arm-specific efficient-influence-function variance contributions. Second, we allow for multiple treatments rather than only two arms. Third, for linear functionals of the outcome mean, we provide a cross-fitted batched procedure together with an explicit second-order characterization of the rate at which the semiparametric benchmark is approached. To the best of our knowledge, this is the first minimax characterization of the rate of convergence to the semiparametric efficiency bound for a general class of experiments.

\noindent \textbf{Covariate balancing and by-design efficiency.}
Covariate balancing---including rerandomization, matched-pair designs, and finely stratified designs---provides an alternative route to improving efficiency, especially in small samples \citep{li2018asymptotic,bai2022optimality,cytrynbaum2024finely}. The central idea in this literature is that semiparametric efficiency can sometimes be attained ``by design,'' allowing simple estimators to perform well without heavy regression adjustment \citep{bai2023efficiency}. However, this literature is largely offline: it typically assumes that all covariates are observed before assignment, focuses primarily on the ATE, and treats the target propensity score as fixed or externally specified. In contrast, our covariate-balancing route is constructive, adaptive, and substantially more general. The design operates in an online batched setting in which covariates arrive sequentially and treatment must be assigned immediately. Rather than assuming that the optimal propensity is known in advance, our procedure jointly learns the target propensity and implements the balancing design within the experiment. Moreover, the resulting efficiency guarantees apply to a broad class of moment-defined estimands with multiple treatments, rather than only the binary-treatment ATE.

\section{Problem Formulation}\label{sec:formulation}

We study a class of causal experiments in which treatment assignment is constrained by a natural \emph{non-anticipation} requirement. We index units by an abstract \emph{assignment order} or \emph{information order}, rather than exclusively by an online arrival time. This broader viewpoint provides a single operational framework that unifies independent and identically distributed (i.i.d.) randomization, offline covariate-based balancing designs (such as stratified or matched-pair experiments), and sequential adaptive designs. This section establishes the basic setup and defines the ingredients needed for the benchmark analysis.

\subsection{Data-generating process}

Fix an integer $\K\ge 2$ and let the treatment space be
\[
\cA=\{1,\dots,\K\}.
\]
Let the covariate space be
$
\cX=[0,1]^d
$
for some fixed dimension $d\ge 1$. This normalization is imposed without loss of generality up to measurable rescaling, and is convenient for the nonparametric arguments used later in the paper.
For each unit $t=1,\dots,n$, let
\[
W_t:=\bigl(\X_t,\Y_t(1),\dots,\Y_t(\K)\bigr)
\]
denote the latent covariates and potential outcomes, where $\X_t\in\cX$ and $\Y_t(a)\in\cY$ for each $a\in\cA$. We assume throughout that the latent units are independent and identically distributed:
\[
W_1,\dots,W_n \stackrel{\mathrm{i.i.d.}}{\sim} P_0.
\]
Write $P_{\X}$ for the marginal law of $\X_t$, and for each $a\in\cA$ let $p_a(y\mid x)$ denote the conditional law of $\Y_t(a)$ given $\X_t=x$.
Under treatment assignment $\A_t\in\cA$, the observed outcome is
$
\Y_t:=\Y_t(\A_t).
$
Thus, the observed data for unit $i$ are given by
\[
\Oo_t=(\X_t,\A_t,\Y_t),\qquad t=1,\dots,n.
\]
Crucially, while the latent units $W_t$ are i.i.d., the observed sequence $\{\Oo_t\}_{t=1}^n$ need not be i.i.d., because the treatment assignment mechanism may depend on previously revealed information.

\subsection{Information order and non-anticipating designs}

The experiment proceeds according to an assignment order $t=1,\dots,n$. This order specifies the sequence in which treatment decisions are made, but need not coincide with online arrival in calendar time. The information available to the experimenter is described by a filtration
\[
\{\cF_t\}_{t=0}^n.
\]
Here $\cF_{t-1}$ denotes the information available immediately before treatment is assigned to unit $t$. The initial $\sigma$-algebra $\cF_0$ may be nontrivial; for instance, in offline designs, it may already contain the entire baseline covariate profile $(\X_1,\dots,\X_n)$.

\begin{definition}[Non-anticipating design]
\label{def:non-anticipating}
An experimental design $\ALG$ is \emph{non-anticipating} if for every $t=1,\dots,n$ there exists a measurable kernel
\(
\e_t(\cdot\mid x,f)\in\Delta(\cA)
\)
such that
\[
\PP(\A_t=a\mid \X_t,\cF_{t-1})=\e_t(a\mid \X_t,\cF_{t-1}),
\qquad a\in\cA.
\]
Equivalently, $\A_t$ is generated by a rule measurable with respect to $\sigma(\X_t,\cF_{t-1})$.
\end{definition}

The intuition behind this definition is straightforward: an assignment decision can depend on any information already revealed, but it cannot depend on unrevealed outcomes or covariates. Future covariates are admissible only if they are already measurable with respect to the current information set, as in offline designs where all covariates are known upfront.
We impose two standard regularity conditions throughout:
\begin{enumerate}[label=(\roman*),leftmargin=*]
\item \textbf{Outcome visibility.} Conditional on $(\X_t,\A_t)$, the law of $\Y_t$ does not depend on $\cF_{t-1}$.
\item \textbf{Positivity.} There exists $\varepsilon>0$ such that $\e_t(a\mid x,f)\ge \varepsilon$ for all $a\in\cA$, all $t$, and for $P_{\X}$-a.e.\ $x$ and $P^{\ALG}$-a.e.\ $f$.
\end{enumerate}
We write $\mathfrak{G}$ for a feasible \emph{design class}, representing any collection of non-anticipating designs satisfying operational restrictions of interest (e.g., exposure limits, treatment-share constraints, or balancing rules). The lower bound developed in Section~\ref{sec:lower-bound} applies broadly to any such design class. The constructive procedures developed later in the paper, however, will focus on approaching this benchmark within standard, practice-relevant subclasses.

\subsection{Examples of non-anticipating designs}

The non-anticipating formulation is designed to be inclusive. We provide four examples, each serving a distinct conceptual role. In each case we make explicit what the information order $\{\cF_t\}$ is, since this determines which designs are admissible.

\paragraph{Example 1: i.i.d.\ randomization.}
There is no pre-experimental information and the filtration is the natural data filtration:
\[
\cF_0=\{\emptyset,\Omega\},\qquad
\cF_t=\sigma(\X_1,\A_1,\Y_1,\dots,\X_t,\A_t,\Y_t),\quad t=1,\dots,n.
\]
The assignment rule does not actually use any of the accumulated history beyond the current covariate $\X_t$:
\[
\e_t(a\mid x,\cF_{t-1})=e(a\mid x),\qquad t=1,\dots,n,
\]
for a fixed propensity $e(a\mid x)$. This is the classical i.i.d.\ experiment: each $(\X_t,\A_t,\Y_t)$ is marginally distributed as in the product model with propensity $e$.

\paragraph{Example 2: complete randomization with fixed treatment shares.}
The experiment enforces exact treatment counts $n_a$ with $\sum_{a=1}^{\K} n_a=n$, assigning sequentially without replacement and without using covariates. The relevant information for the assignment rule is the vector of past treatment counts, so one may take
\[
\cF_0=\{\emptyset,\Omega\},\qquad
\cF_t=\sigma(\A_1,\dots,\A_t),\quad t=1,\dots,n,
\]
or equivalently the full history filtration (which contains the same treatment-count information). The assignment rule depends on the history only through $N_{t-1,a}:=\sum_{s=1}^{t-1}\1\{\A_s=a\}$:
\[
\e_t(a\mid x,\cF_{t-1})
=
\frac{n_a-N_{t-1,a}}{n-t+1}.
\]
This design yields dependent assignments across units---knowing past assignments changes the remaining allocation budget---but remains strictly non-anticipating because $N_{t-1,a}$ is $\cF_{t-1}$-measurable.

\paragraph{Example 3: offline stratified or matched-pair randomization.}
The distinctive feature of offline designs is that the entire covariate profile is observed before any treatment is assigned. Formally,
\[
\cF_0=\sigma(\X_1,\dots,\X_n),\qquad
\cF_t=\sigma\bigl(\X_1,\dots,\X_n,\,\A_{\pi(1)},\Y_{\pi(1)},\dots,\A_{\pi(t)},\Y_{\pi(t)}\bigr),
\quad t=1,\dots,n,
\]
where $\pi=\pi(\X_1,\dots,\X_n)$ is any measurable ordering of the units determined from the full covariate profile. Since all covariates are already observed at time $0$, fixing such an order is without loss of generality: it simply provides a sequential representation of the offline design.
Let $\{B_\ell\}_{\ell=1}^L$ be a partition of $\cX$, and let
\[
N_\ell:=\sum_{i=1}^n \1\{\X_i\in B_\ell\}
\]
be the realized size of stratum $\ell$. Suppose that within each stratum $\ell$, exactly $m_{\ell a}$ units are assigned to arm $a$, with $\sum_{a=1}^{\K}m_{\ell a}=N_\ell$. For the unit $\pi(t)$ with $\X_{\pi(t)}\in B_\ell$, let
\[
N_{\ell a,t-1}:=\sum_{s=1}^{t-1}\1\{\X_{\pi(s)}\in B_\ell,\A_{\pi(s)}=a\},
\qquad
R_{\ell,t-1}:=\sum_{s=1}^{t-1}\1\{\X_{\pi(s)}\in B_\ell\}.
\]
Then the sequential assignment rule is
\[
\e_t(a\mid \X_{\pi(t)},\cF_{t-1})
=
\frac{m_{\ell a}-N_{\ell a,t-1}}{N_\ell-R_{\ell,t-1}},
\qquad \text{when } \X_{\pi(t)}\in B_\ell.
\]
Thus, conditional on the full covariate profile, the design is simply without-replacement randomization within each stratum. The induced average propensity for any unit with $x\in B_\ell$ is therefore
\[
\bar e_a(x)=\frac{m_{\ell a}}{N_\ell},\qquad x\in B_\ell.
\]
Matched-pair randomization is the special case in which each stratum has size two and one unit is assigned to each arm.

\paragraph{Example 4: batched adaptive design.}
The assignment order is divided into consecutive batches $I_1,\dots,I_B$ of sizes $\gamma_b:=|I_b|$. The information order is the natural full-history filtration
\[
\cF_0=\{\emptyset,\Omega\},\qquad
\cF_t=\sigma(\X_1,\A_1,\Y_1,\dots,\X_t,\A_t,\Y_t),
\]
together with the coarser \emph{pre-batch} filtration $\cH_b:=\cF_{\min I_b-1}$ (with $\cH_1=\cF_0$), which contains all data collected before batch $b$ starts. The defining restriction of a batched design is that the within-batch propensity function $e^{(b)}(a\mid x)$ is $\cH_b$-measurable and does not update within the batch. For any $t\in I_b$,
\[
\e_t(a\mid x,\cF_{t-1})=e^{(b)}(a\mid x).
\]
Thus, within a batch, observations are conditionally i.i.d.\ given $\cH_b$; across batches, the propensity is updated based on accumulated history. The induced average propensity is
\[
\bar e_a(x)
=
\frac{1}{n}\sum_{b=1}^B \gamma_b\,\E^{\ALG}\!\left[e^{(b)}(a\mid x)\right].
\]
This batched structure is the primary operational form studied in the later constructive sections.

\subsection{Average propensity as the sufficient statistics of the design}

The central design object for the lower-bound analysis is the \emph{induced average propensity score}. It measures the average assignment mass allocated to each treatment at each covariate value over the course of the experiment.

\begin{definition}[Average propensity]
\label{def:avg-prop}
For any non-anticipating design $\ALG$, define
\begin{equation}
\label{eq:avg-prop}
\bar e_a(x)
:=
\frac{1}{n}\sum_{t=1}^n
\E^{\ALG}\!\left[\e_t(a\mid x,\cF_{t-1})\right],
\qquad a\in\cA.
\end{equation}
\end{definition}
Different designs may generate complex and highly diverse dependence structures. However, the lower bound in the next section will show that the relevant summary of the design for the statistical benchmark is captured by $\bar e$, which can be understood as the ``sufficient statistics" of a design.

\subsection{Target parameter and the efficiency benchmark under an i.i.d. assignment rule}

Before addressing dependent designs, we first ask a simpler benchmark question: if the assignment rule were fixed, what would the semiparametric efficiency benchmark be? 
Let $\Oo=(\X,\A,\Y)$ denote an observation from an i.i.d.\ model with density
\[
p(o)=p_{\X}(x)\,\e(a\mid x)\,p_a(y\mid x),
\qquad a\in\cA,
\]
where the propensity $\e(\cdot\mid x)$ is treated as a known, fixed rule, while $p_{\X}$ and $\{p_a(\cdot\mid x)\}_{a\in\cA}$ are unrestricted. Let $\theta:\cM(\e)\to\R$ be a pathwise differentiable target parameter, and let $\varphi_{\e}\in L_2^0(P_0)$ denote the efficient influence function (EIF) under propensity $\e$, which is the canonical score object whose variance determines the semiparametric efficiency benchmark in the i.i.d. design model.
The following lemma gives a standard structural representation of the EIF under known assignment.

\begin{lemma}[EIF structure under known assignment]
\label{lem:eif-structure}
Let $\theta:\cM(\e)\to\R$ be pathwise differentiable at $P_0$. Then there exist measurable functions $g:\cX\to\R$ and $v_a:\cX\times\cY\to\R$ for $a\in\cA$ such that
\begin{equation}
\label{eq:eif-structure}
\varphi_{\e}(\Oo)
=
\bigl(g(\X)-\theta(P_0)\bigr)
+
\sum_{a\in\cA}\frac{\1\{\A=a\}}{\e(a\mid \X)}\,v_a(\X,\Y),
\end{equation}
where
\[
\E\!\left[v_a(\X,\Y)\mid \X,\A=a\right]=0,\qquad a\in\cA.
\]
\end{lemma}
For a fixed propensity $\e$, define the corresponding semiparametric efficiency benchmark by
\begin{equation}
\label{eq:fixed-prop-benchmark}
V(\e):=\E\!\left[\varphi_{\e}(\Oo)^2\right].
\end{equation}
If the data were genuinely i.i.d.\ under the fixed rule $\e$, the semiparametric efficiency bound for estimating $\theta(P_0)$ would be $V(\e)/n$. 
For example, in the binary-treatment ATE case with propensity $e(x)=\PP(A=1\mid X=x)$ and outcome regressions $\mu_a(x)=\E[Y\mid X=x,A=a]$, the EIF is
\[
\varphi_{\mathrm{ATE}}(O)
=
\mu_1(X)-\mu_0(X)-\theta
+\frac{A}{e(X)}\bigl(Y-\mu_1(X)\bigr)
-\frac{1-A}{1-e(X)}\bigl(Y-\mu_0(X)\bigr).
\]
Accordingly, the semiparametric efficiency bound is $\E[\varphi_{\mathrm{ATE}}(O)^2]/n$, which is the smallest asymptotic variance achievable by regular estimators under the fixed assignment rule. The key benchmark question is whether, and how, this fixed-assignment bound continues to govern adaptive or balanced designs with \textbf{non-i.i.d.}\ observed data.

\section{The Statistical Limit of General Experiment Designs}\label{sec:lower-bound}

This section addresses the benchmark question raised in the introduction: within a feasible class of non-anticipating experiments, what is the best achievable statistical accuracy for estimating a target causal estimand? Our main result demonstrates that this benchmark is given by the fixed-assignment semiparametric efficiency bound evaluated at the induced average propensity score. In this sense, the statistical limit of a complex experiment is governed by a much simpler object: the average treatment exposure that the design allocates to each arm across covariate values.

\subsection{A scalar benchmark via average propensity}

We begin with a scalar target parameter $\theta(P_0)$. The key result of this subsection establishes that, although a non-anticipating design may induce substantial dependence across observations through adaptivity, balancing, or operational constraints, its semiparametric efficiency benchmark remains characterized by the fixed-assignment benchmark, evaluated at the induced average propensity score $\bar e$.
Intuitively, no matter how complicated the assignment rule is over the course of the experiment, the precision with which outcomes under a given treatment can ultimately be estimated depends on how much sample exposure the design allocates to that treatment across covariate values. The induced average propensity score $\bar e$ summarizes exactly this average exposure. The theorem shows that, for the purpose of the semiparametric efficiency benchmark, this is the object that matters. The formal proof is more subtle and is deferred to the appendix.

\begin{theorem}[Efficiency benchmark via average propensity]
\label{thm:avg-prop-lb}
Let $\ALG$ be any non-anticipating design, and let $\bar e$ be its induced average propensity. Let $T=T(\cF_n)$ be an estimator of $\theta(P_0)$. Suppose that $T$ is locally unbiased along the least-favorable submodel associated with the fixed-assignment model under propensity $\bar e$; that is, along this one-dimensional submodel $\{P_\tau:|\tau|<\delta\}$,
\[
\E_{P_0}^{\ALG}[T]=\theta(P_0),
\qquad
\left.\frac{d}{d\tau}\E_{P_\tau}^{\ALG}[T]\right|_{\tau=0}
=
\left.\frac{d}{d\tau}\theta(P_\tau)\right|_{\tau=0}.
\]
Then
\begin{equation}
\label{eq:avg-prop-lb}
\Var_{P_0}^{\ALG}(T)
\ \ge\
\frac{1}{n}\,V(\bar e)
\;=\;
\frac{1}{n}\E\!\left[\varphi_{\bar e}(\Oo)^2\right].
\end{equation}
\end{theorem}

\begin{remark}[Asymptotically unbiased estimators]
\label{rem:avg-prop-asymp}
The finite-sample statement above is the exact form used in the proof. The same argument yields the usual first-order version for regular estimators. For a sequence of non-anticipating designs $\ALG_n$ with induced average propensities $\bar e_n$, suppose that $T_n=T_n(\cF_n)$ has local bias $o(n^{-1/2})$ along the least-favorable submodel associated with the fixed-assignment model under $\bar e_n$. Then
\[
\Var_{P_0}^{\ALG_n}(T_n)
\ \ge\
\frac{1}{n}V(\bar e_n)+o(n^{-1}).
\]
Consequently, if $V(\bar e_n)\to V(\bar e)$, then
\[
\liminf_{n\to\infty} n\,\Var_{P_0}^{\ALG_n}(T_n)\ge V(\bar e).
\]
\end{remark}

Theorem~\ref{thm:avg-prop-lb} shows that, for locally unbiased estimation, the finite-sample efficiency benchmark of a non-anticipating experiment is governed by a single design object: its induced average propensity score. Remark~\ref{rem:avg-prop-asymp} gives the corresponding first-order formulation for regular estimators whose local bias is negligible at the $n^{-1/2}$ scale. Thus, the same benchmark applies both to the exact unbiased formulation used in the proof and to the asymptotic regime commonly used in semiparametric efficiency theory.
The implication is that adaptivity, balancing, batching, and other operational features affect first-order efficiency through the average assignment mass they deliver to each treatment arm across covariate values. Once the induced average propensity $\bar e$ is fixed, the corresponding i.i.d.\ semiparametric benchmark $V(\bar e)/n$ gives the relevant lower bound. In this sense, the average propensity score is a sufficient design summary for the efficiency benchmark and the natural object to optimize or learn in the constructive procedures developed below.

As a practical illustration, consider a platform evaluating a new AI-assisted feature under rollout constraints. The platform may utilize a highly adaptive rule to decide which users receive the feature, perhaps to respect traffic limits or fairness requirements. Theorem~\ref{thm:avg-prop-lb} shows that, for the purpose of the benchmark, the full complexity of that adaptive rule is encapsulated by the average treatment exposure it induces across user types. Thus, the relevant benchmark question is not how to optimize an arbitrary path-dependent assignment mechanism directly, but which average propensity scores are achievable and how favorable they are for the target estimand.

\subsection{Vector targets and generalized Neyman allocation}

Many modern experiments are inherently multi-treatment and therefore involve vector-valued targets. This is particularly relevant in applications where the goal is not a single scalar contrast, but the simultaneous evaluation of several treatment effects. For instance, a platform may compare multiple AI-assisted workflows against a common baseline, seeking to estimate the vector of treatment effects for all variants at once; similarly, one may wish to estimate several policy-value components or multiple weighted ATEs simultaneously. This motivates extending the benchmark from a scalar parameter to a vector target.

Let $\theta(P_0)\in\R^p$ be a pathwise differentiable vector parameter, and let $\varphi_{\bar e}(\Oo)\in\R^p$ denote the corresponding efficient influence function. For estimators that are locally unbiased along all one-dimensional least-favorable directions, applying Theorem~\ref{thm:avg-prop-lb} to each linear functional yields the matrix lower bound:
\begin{equation}
\label{eq:matrix-lb}
\Cov(\hat\theta)\ \succeq\ \frac{1}{n}\E\!\left[\varphi_{\bar e}(\Oo)\varphi_{\bar e}(\Oo)^\top\right].
\end{equation}
Taking the trace provides the corresponding lower bound on total squared error:
\begin{equation}
\label{eq:l2-lb}
\E\norm{\hat\theta-\theta(P_0)}^2
\ \ge\
\frac{1}{n}\E\!\left[\norm{\varphi_{\bar e}(\Oo)}^2\right].
\end{equation}

The matrix bound in \eqref{eq:matrix-lb} is crucial when one requires the full covariance structure of a vector estimator. However, in many experimental settings, a natural scalar summary of performance is the total squared error across coordinates, leading to the $L_2$-risk bound in \eqref{eq:l2-lb}. This criterion is especially intuitive for vector ATE problems, where the experimenter seeks to estimate several arm-specific treatment effects simultaneously and desires a design that minimizes aggregate estimation error across all coordinates.
Using the structural representation from Lemma~\ref{lem:eif-structure}, define:
\[
M_a(x):=
\E\!\left[v_a(\X,\Y)v_a(\X,\Y)^\top \mid \X=x,\A=a\right],
\qquad a\in\cA.
\]
These matrices capture the arm-specific conditional second-moment contribution of the efficient influence function. Substituting the EIF decomposition into the risk bound yields:
\[
\E\norm{\varphi_{\bar e}(\Oo)}^2
=
\E\norm{g(\X)-\theta(P_0)}^2
+
\E\!\left[\sum_{a\in\cA}\frac{\tr(M_a(\X))}{\bar e_a(\X)}\right].
\]
Absent additional design constraints, minimizing the integrand pointwise over all probability vectors $(e_a(x))_{a\in\cA}$ produces:
\begin{equation}
\label{eq:opt-prop}
e_a^\star(x)
=
\frac{\sqrt{\tr(M_a(x))}}{\sum_{j\in\cA}\sqrt{\tr(M_j(x))}},
\qquad a\in\cA.
\end{equation}

Equation~\eqref{eq:opt-prop} offers a direct design implication of the lower bound: treatment arms that are more difficult to evaluate precisely for the target metric should receive greater experimental exposure. In the special case of a binary average treatment effect, this reduces to the classical Neyman allocation rule. More generally, it yields an estimand-specific generalized Neyman allocation principle for multi-treatment experiments and, notably, for vector-valued ATE targets. When the target is the vector of treatment effects for several arms relative to a baseline, the optimal design is governed not by any single contrast in isolation, but by the joint second-moment structure of the efficient influence function across coordinates.

To interpret this practically, suppose a platform is comparing several AI service variants, and one variant exhibits substantially higher conditional variability in task quality or resolution rate than the others. That specific arm should receive more traffic if the goal is to minimize the overall estimation error of the target metric. In this manner, the lower bound does not merely characterize a limit; it points directly to a concrete allocation principle for efficient experimentation.

\subsection{Discussion and connection to the rest of the paper}

This section has addressed the benchmark question for a broad class of non-anticipating experiment designs. By isolating the induced average propensity score as the central benchmark object, Theorem~\ref{thm:avg-prop-lb} provides a common language for evaluating complex experimental structures.

From a theoretical perspective, this benchmark unifies several efficiency bounds that have recently been investigated in distinct areas of the experimental-design literature. Existing results for sequential adaptive experiments \citep{armstrong2022asymptotic} and highly stratified designs \citep{bai2023efficiency,cai2024performance} operate under fundamentally different assumptions and do not subsume one another; yet, both fall naturally within the present framework of non-anticipating designs. The average-propensity benchmark developed here therefore provides a unified foundation across these previously disparate settings. To further illustrate its generality, we apply this framework to resolve an existing conjecture regarding the semiparametric efficiency bound in two-stage experimental designs \citep{cytrynbaum2021optimal}. In such settings, an experimenter facing budget constraints must first select a subset of eligible units before determining their treatment assignments; our benchmark systematically accommodates this sequential selection-and-assignment mechanism.

This framework also extends naturally beyond scalar targets to vector-valued estimands. This is particularly relevant for the multi-treatment settings that motivate our applications, where the object of interest is often a vector of treatment effects rather than a single scalar contrast. In such scenarios, the benchmark characterizes the achievable covariance structure of the estimator and provides a joint allocation principle through the generalized Neyman rule in \eqref{eq:opt-prop}.

This average-propensity perspective also provides a simple way to incorporate application-specific operational constraints. In many experiments, the assignment rule must satisfy exposure limits, treatment-budget constraints, fairness requirements across subgroups, or other rollout restrictions. Such restrictions can be represented through an admissible class of average propensities, say $\mathcal E_{\mathrm{feas}}$. The relevant oracle benchmark is then obtained by minimizing the same semiparametric objective over $\mathcal E_{\mathrm{feas}}$ rather than over the full simplex. Thus, constraints affect the benchmark through the set of achievable average propensities, while the lower-bound principle itself remains unchanged: once a design induces an average propensity $\bar e$, its first-order efficiency limit is governed by the i.i.d. benchmark evaluated at $\bar e$. The constructive sections below focus on standard implementable classes of batched designs, while this viewpoint leaves room for application-specific constraints to be imposed through the propensity-update step.

Most importantly for the remainder of the paper, the induced average propensity score serves two roles. It is the sufficient summary of the design for the information-theoretic lower bound, and it also acts as the explicit target object for the constructive design problem. Guided by this benchmark, the subsequent sections develop two complementary batched adaptive routes to approach the same semiparametric efficiency limit: Section~\ref{sec:upper-bound} achieves this through estimator-side correction via regression adjustment, while Section~\ref{sec:cov-balancing} does so through design-side control via covariate balancing.


\section{Achieving Semiparametric Efficiency via Regression Adjustment}\label{sec:regression-adjustment}
\label{sec:upper-bound}

Having characterized the semiparametric efficiency benchmark through the induced average propensity score in Section~\ref{sec:lower-bound}, we now turn to the first constructive route for approaching that benchmark. This section studies estimator-side correction through batched adaptive regression adjustment. We show that the oracle benchmark $V^\star/n$ can be approached by a practical batched design, but that the required conditions depend fundamentally on the structure of the estimand. In particular, we first establish a general achievability result, and then show that for linear functionals of the outcome mean the regression-adjustment route becomes substantially sharper because the plug-in bias cancels exactly.

\subsection{Setup, Oracle Objective, and Batched Cross-Fitted Design}
\label{subsec:setup-upper}

We observe $n$ rounds of data $\Oo_t=(\X_t,\A_t,\Y_t)$ with treatments $\A_t\in\cA=\{1,\dots,\K\}$. The non-anticipating experiment is conducted in $B$ batches with index sets $\{\mathcal I_b\}_{b=1}^B$ and batch sizes $\gamma_b:=|\mathcal I_b|$. Within each batch $b$, the assignment rule $\e^{(b)}(\cdot)$ is fixed and measurable with respect to the pre-batch history $\cH_b$. Consequently, conditional on $\cH_b$, the observations within batch $b$ are i.i.d.

Recall from Lemma~\ref{lem:eif-structure} that the efficient influence function (EIF) for a target parameter $\theta(P_0)$ under propensity $\e$ takes the form
\[
\IF_\e(\Oo)
=
\bigl(g(\X)-\theta\bigr)
+
\sum_{a\in\cA}\frac{\1\{\A=a\}}{\e_a(\X)}\,v_a(\X,\Y).
\]
To isolate the design problem, we define the \emph{conditional second-moment function}
\begin{equation}
\label{eq:second-moment-def}
m_a(x) := \E\!\left[v_a(\X,\Y)^2\mid \X=x,\A=a\right],\qquad a\in\cA.
\end{equation}
The oracle design objective over a feasible class $\cE$ is then
\begin{equation}
\label{eq:design-objective}
L(\e) := \E\!\left[\sum_{a\in\cA}\frac{m_a(\X)}{\e_a(\X)}\right],
\qquad
\e^\star \in \arg\min_{\e\in\cE} L(\e),
\end{equation}
and the corresponding oracle semiparametric benchmark variance is
\[
V^\star := \Var\!\bigl(g(\X)-\theta\bigr) + L(\e^\star).
\]

Our regression-adjustment route uses a batched cross-fitted design with three ingredients. First, the sample is split into two folds to separate design learning from final EIF evaluation. Second, at the end of each batch, past data from one fold are used to estimate the second-moment proxy and update the target propensity for the opposite fold by plug-in minimization of the oracle design objective. Third, after all batches are completed, the final estimator is formed by cross-fitted EIF evaluation using the realized logging propensities. This design requires only a logarithmic number of policy updates, is naturally compatible with delayed feedback, and preserves conditional i.i.d.\ structure within each batch. Full implementation details are deferred to Appendix~\ref{app:alg-details}.
\begin{figure}
    \centering
    \includegraphics[width=0.8\linewidth]{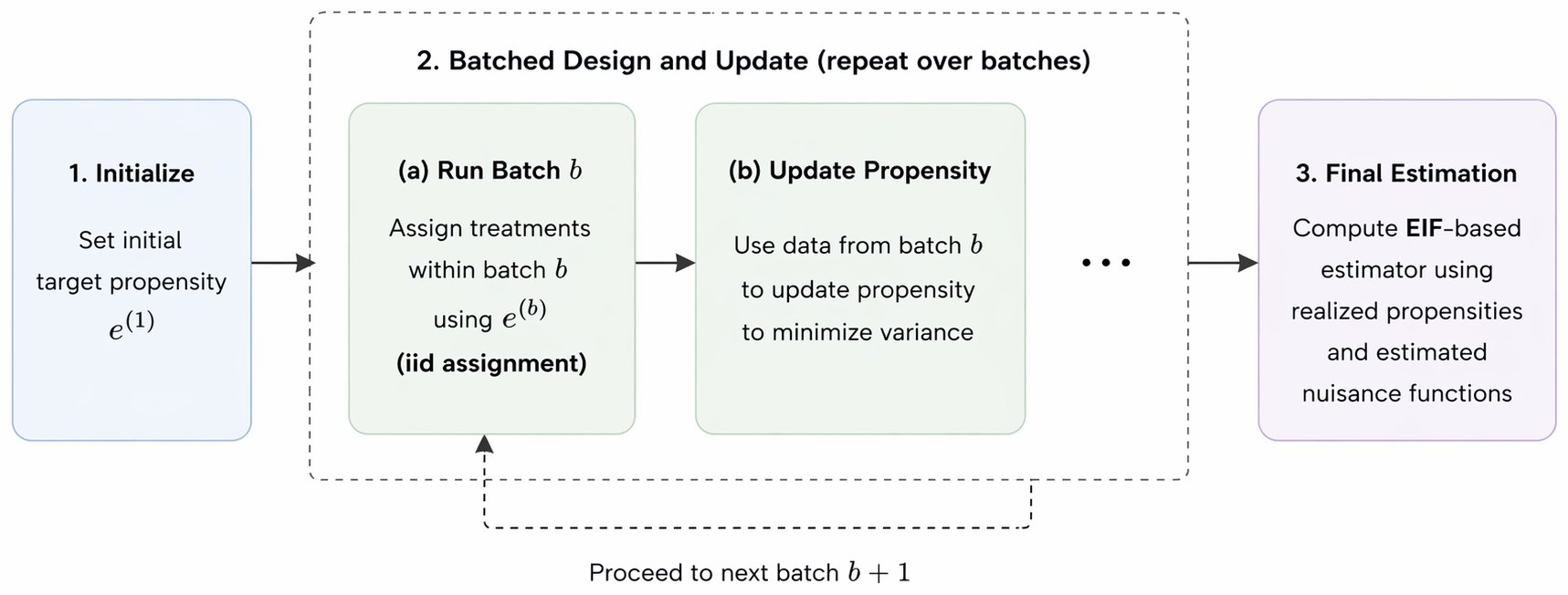}
    \caption{Illustration of the workflow of regression adjustment.}
    \label{fig:ra_illu}
\end{figure}

The remainder of this section shows that approaching the oracle rate $V^\star/n$ reduces to controlling two sources of error: nuisance estimation error (model estimation) and the regret from learning the assignment rule (design optimization). The next subsection introduces an excess-risk decomposition that organizes the analysis of the general and linear-functional regimes.
\subsection{The Excess-Risk Decomposition}
\label{subsec:decomposition}

This subsection provides the organizing device for the rest of the analysis. We introduce two infeasible oracle estimators: $\tilde\theta$, which uses the true nuisance parameters but the realized sequence of designs, and $\theta_n^\star$, which uses the true nuisance parameters together with the oracle design $\e^\star$. To isolate the sources of suboptimality, we expand the mean squared error of the feasible estimator $\hat\theta$.

\begin{proposition}[Excess MSE decomposition]
\label{prop:excess-mse-decomp}
The excess risk of the feasible estimator $\hat\theta$ decomposes exactly into three components:
\begin{equation}
\label{eq:excess-mse-decomp}
\E\!\left[(\hat\theta-\theta)^2\right]-\E\!\left[(\theta^\star_n-\theta)^2\right]
=
\underbrace{\E\!\left[(\hat\theta-\tilde\theta)^2\right]}_{\textnormal{model estimation}}
+
\underbrace{\Big(\E\!\left[(\tilde\theta-\theta)^2\right]-\E\!\left[(\theta^\star_n-\theta)^2\right]\Big)}_{\textnormal{design optimization}}
+
\underbrace{2\,\E\!\left[(\hat\theta-\tilde\theta)(\tilde\theta-\theta)\right]}_{\textnormal{interaction term}}.
\end{equation}
\end{proposition}
The \textbf{model estimation} term measures the cost of plugging estimated nuisance parameters into the EIF. The \textbf{design optimization} term measures the suboptimality of the learned propensities relative to the oracle design. The \textbf{interaction term} captures the coupling between nuisance error and design regret. The behavior of this interaction term---specifically, whether plug-in bias persists or cancels exactly---is what distinguishes the general EIF case from the linear-functional case.

\subsection{General EIF Scores: A Baseline and the $n^{-1/4}$ Bottleneck}
\label{subsec:general-case}

We first derive a baseline achievability result for general smooth estimands whose efficient
influence functions admit the representation in Lemma~\ref{lem:eif-structure}. The purpose of
this subsection is deliberately modest. It shows that the average-propensity benchmark can be
approached by batched regression adjustment under standard orthogonality and second-order
remainder conditions, but it also identifies the familiar nuisance-estimation bottleneck that
prevents sharper conclusions in full generality.

Let $\eta$ denote the collection of nuisance components entering the EIF score, and let
$d(\cdot,\cdot)$ be a pseudo-metric measuring nuisance error. The key structural condition is
\emph{Neyman orthogonality}: the pathwise derivative of the expected EIF score vanishes at the
true nuisance $\eta_0$. Together with common regularity conditions  collected in Assumption~\ref{assump:general-reg} in
Appendix~\ref{app:regularity}, this condition removes the first-order effect of nuisance
estimation. What remains, however, is generally a second-order plug-in bias. As a result, the
interaction term in \eqref{eq:excess-mse-decomp} need not vanish exactly and must be controlled
indirectly.
This general framework covers nonlinear smooth targets. For example, consider the quantile
treatment effect
\[
    \theta = q_1(\tau)-q_0(\tau),
    \qquad q_a(\tau):=F_{Y(a)}^{-1}(\tau).
\]
Let
\[
    G_a(y\mid x):=\Pr(\Y\le y\mid \X=x,\A=a)
\]
and let $f_a(q_a)$ denote the marginal density of $Y(a)$ at $q_a(\tau)$. The arm-specific EIF for
$q_a(\tau)$ is
\[
\phi_{q_a}(\Oo)
=
-f_a(q_a)^{-1}
\left[
G_a(q_a\mid \X)-\tau
+
\frac{\1\{\A=a\}}{\e_a(\X)}
\left\{\1\{\Y\le q_a\}-G_a(q_a\mid \X)\right\}
\right].
\]
Thus $\phi_{\mathrm{QTE}}=\phi_{q_1}-\phi_{q_0}$, and the arm-specific residual entering the EIF
representation is
\[
    v_a(\X,\Y)
    =
    \frac{\1\{\Y\le q_a\}-G_a(q_a\mid \X)}{f_a(q_a)}.
\]
Consequently, the design-relevant second moment is
\[
    m_a(x)
    =
    \frac{G_a(q_a\mid x)\{1-G_a(q_a\mid x)\}}{f_a(q_a)^2}.
\]
The nuisance vector
\(
    \eta=(q_a,f_a(q_a),G_a(q_a\mid\cdot))_{a\in\{0,1\}}
\)
enters through the natural arm-wise pseudo-metric, and Assumption~\ref{assump:general-reg}
holds under standard smoothness at the target quantiles. Linear targets such as ATEs and policy
values are treated separately in Section~\ref{subsec:linear-case}, because they enjoy a stronger
algebraic cancellation that is not available for general EIF scores.

\begin{theorem}[Model estimation bound]
\label{thm:model-estimation-general}
Under Assumption~\ref{assump:general-reg}, with $B$ batches, the model-estimation error satisfies
\begin{equation}
\label{eq:model-estimation-bound-general}
\E[(\hat\theta-\tilde\theta)^2]
\lesssim
\underbrace{
\frac{B}{n^2}
\sum_{b=1}^B
\gamma_b\,
\E\!\left[d(\hat\eta_b,\eta_0)^2\right]
}_{\text{variance}}
+
\underbrace{
\frac{B}{n^2}
\sum_{b=1}^B
\gamma_b^2\,
\E\!\left[d(\hat\eta_b,\eta_0)^4\right]
}_{\text{squared bias}} .
\end{equation}
\end{theorem}

Theorem~\ref{thm:model-estimation-general} makes the bottleneck explicit. If
$d(\hat\eta_b,\eta_0)=O_P(n^{-\alpha})$, then, up to logarithmic factors from the number of
batches, the variance term is of order $n^{-1-2\alpha}$ and is therefore $o(1/n)$ for any
$\alpha>0$. By contrast, the squared-bias term is of order $n^{-4\alpha}$ and is $o(1/n)$ only
when $\alpha>1/4$. Thus, for general EIF scores, batched regression adjustment reaches the
oracle benchmark only under the classical nuisance-rate condition $\alpha>1/4$.

This bottleneck is not caused by adaptivity itself. It arises because Neyman orthogonality removes
only the local first-order effect of nuisance estimation, leaving a second-order plug-in bias. The
batched structure is still useful: conditional on the pre-batch history, the assignment rule is fixed
within a batch, so nuisance estimation can rely on standard i.i.d. methods batch by batch. The
limitation is that, for general EIF scores, this i.i.d. nuisance estimation must be sufficiently fast.
We next control the design-optimization term by relating regret in \eqref{eq:design-objective} to
the accuracy of the second-moment proxies.

\begin{theorem}[Design learning bound]
\label{thm:design-opt-general}
If the assignment rules $\e^{(b)}$ are updated by approximately minimizing
\eqref{eq:design-objective} using second-moment estimates $\hat m^{(b)}$, then the cumulative
design regret satisfies
\begin{equation}
\label{eq:design-learning-general}
\E\!\left[(\tilde\theta-\theta)^2\right]
-
\E\!\left[(\theta^\star_n-\theta)^2\right]
\lesssim
\frac{1}{n^2}
\sum_{b=2}^B
\gamma_b\,
\E\!\left[
\sum_{a\in\cA}
\bigl(\hat m_a^{(b-1)}(\X)-m_a(\X)\bigr)^2
\right].
\end{equation}
\end{theorem}
Thus, consistency of the second-moment proxies is enough to make the design-learning term
$o(1/n)$. Combining Theorems~\ref{thm:model-estimation-general} and
\ref{thm:design-opt-general} gives the following baseline achievability result.

\begin{corollary}[Baseline achievability for general EIF scores]
\label{cor:main-general}
Under Assumption~\ref{assump:general-reg}, if the second-moment proxies are consistent and the
full nuisance parameter is estimated at a rate strictly faster than $n^{-1/4}$, then
Algorithm~\ref{alg:batched-cf} achieves the semiparametric efficiency benchmark:
\[
    \E[(\hat\theta-\theta)^2]
    =
    \frac{V^\star}{n}
    +
    o(1/n).
\]
\end{corollary}
The next subsection shows that this nuisance-rate bottleneck disappears for linear functionals of
the outcome mean. There, the regression-adjusted IPW score is globally unbiased for every
candidate regression function, so the second-order plug-in bias in
\eqref{eq:model-estimation-bound-general} vanishes identically rather than merely becoming small.

\subsection{Linear Functionals: Exact Cancellation and Sharp Rates}
\label{subsec:linear-case}

We now turn to the a more interesting regime of the regression-adjustment route. Many experimental targets of
primary practical interest are linear functionals of the arm-specific conditional means:
\[
\theta
=
\E\!\left[\sum_{a\in\cA}\omega_a(\X)\mu_a(\X)\right],
\]
where the weights $\{\omega_a\}_{a\in\cA}$ are known. This class includes average treatment
effects, weighted and subgroup average treatment effects, and policy values under a fixed target
policy. For this class, the general $n^{-1/4}$ nuisance-rate bottleneck disappears entirely.

The reason is structural. For general EIF scores, Neyman orthogonality removes only the
\emph{local first-order} effect of nuisance estimation, leaving a second-order plug-in bias of order
$d(\hat\eta,\eta_0)^2$. Linear functionals enjoy a stronger property. Because the target is linear
in $\mu_a$ and the EIF correction uses inverse-probability weighting, the plug-in error in any
candidate regression function $\tilde\mu_a$ cancels algebraically, arm by arm. Thus the score is
globally unbiased for every candidate regression function, not merely locally orthogonal at the
truth.
For a linear functional, the regression-adjusted IPW score is
\begin{equation}
\label{eq:linear-mu-eif}
\psi(\Oo;\mu,\e)
=
\sum_{a\in\cA}\omega_a(\X)\mu_a(\X)
+
\sum_{a\in\cA}
\frac{\1\{\A=a\}}{\e_a(\X)}
\,\omega_a(\X)\bigl(\Y-\mu_a(\X)\bigr)
-\theta .
\end{equation}
The next proposition records the exact cancellation property.

\begin{proposition}[Exact global unbiasedness and cancellation]
\label{prop:linear-unbiased}
For any candidate regression function $\tilde\mu=\{\tilde\mu_a\}_{a\in\cA}$,
\begin{equation}
\label{eq:linear-global-unbiased}
\E\!\left[\psi(\Oo;\tilde\mu,\e)\mid \X\right]
=
\sum_{a\in\cA}\omega_a(\X)\mu_a(\X)-\theta,
\qquad
\E\!\left[\psi(\Oo;\tilde\mu,\e)\right]=0 .
\end{equation}
Consequently, under cross-fitting, the interaction term in the excess-MSE decomposition vanishes
exactly:
\[
\E\!\left[(\hat\theta-\tilde\theta)(\tilde\theta-\theta)\right]=0 .
\]
\end{proposition}
The proof is immediate from conditional expectation. For each arm $a$,
\[
\E\!\left[
\frac{\1\{\A=a\}}{\e_a(\X)}
\,\omega_a(\X)\bigl(\Y-\tilde\mu_a(\X)\bigr)
\,\middle|\,\X
\right]
=
\omega_a(\X)\bigl(\mu_a(\X)-\tilde\mu_a(\X)\bigr),
\]
which exactly cancels the plug-in term $\omega_a(\X)\tilde\mu_a(\X)$ in
\eqref{eq:linear-mu-eif}. Hence the plug-in bias is not merely second-order small; it is
identically zero.

This exact cancellation changes the rate analysis. In the general EIF case, the model-estimation
bound contains a squared-bias term, which forces the nuisance-rate condition $\alpha>1/4$. In the
linear-functional case, that term vanishes. The remaining model-estimation cost is purely a
variance contribution from learning the regression functions. The design-learning cost is governed
separately by the accuracy with which the design-relevant second moments are estimated.
To state the rate result, define
\[
m_a(x)
:=
\omega_a(x)^2\sigma_a^2(x),
\qquad
\sigma_a^2(x):=\Var(\Y\mid \X=x,\A=a).
\]
Then the oracle design objective for the linear-functional problem is
\[
L(e)
=
\E\!\left[
\sum_{a\in\cA}\frac{m_a(\X)}{\e_a(\X)}
\right].
\]
Thus, two nuisance objects matter for the rate analysis: the outcome regressions
$\{\mu_a\}_{a\in\cA}$, which enter the final regression-adjusted estimator, and the second-moment
functions $\{m_a\}_{a\in\cA}$, which determine the optimal design.

Because the experiment is batched, the assignment rule is fixed within each batch conditional on
the pre-batch history. Hence the nuisance learners can be analyzed using standard conditional
i.i.d.\ arguments within each batch. We impose the following batchwise rate condition.

\begin{assumption}[Nuisance-rate conditions for linear functionals]
\label{assump:linear-rates}
There exist constants $\alpha,\beta>0$ and $C<\infty$ and estimators $\hat\mu_a$ and $\hat m_a$ such that given $n$ i.i.d. samples,
\begin{equation}
\label{eq:linear-mu-rate-assump}
\E\!\left[
\sum_{a\in\cA}
\|\hat\mu_a-\mu_a\|_{L_2(P_\X)}^2
\right]
\le
C n^{-2\alpha}, \, \text{and } \,
\E\!\left[
\sum_{a\in\cA}
\|\hat m_a-m_a\|_{L_2(P_\X)}^2
\right]
\le
C n^{-2\beta}.
\end{equation}
\end{assumption}
Assumption~\ref{assump:linear-rates} separates the two statistical tasks. The exponent $\alpha$
measures the difficulty of learning the outcome regressions used in the final regression-adjusted
estimator. The exponent $\beta$ measures the difficulty of learning the second-moment functions
that drive the design update. Importantly, no threshold condition such as $\alpha>1/4$ is imposed.
Any positive rates $\alpha,\beta>0$ are sufficient for first-order efficiency; their values determine
only the second-order speed at which the feasible procedure approaches the oracle benchmark.

\begin{theorem}[Convergence-rate upper bound for linear functionals]
\label{thm:linear-bounds}
Under Assumption~\ref{assump:linear-rates}, with a logarithmic number of batches and the
batched cross-fitted design in Algorithm~\ref{alg:batched-cf}, the model-estimation error satisfies
\begin{equation}
\label{eq:linear-model-estimation}
\E[(\hat\theta-\tilde\theta)^2]
\;\lesssim\;
\widetilde O\!\left(n^{-1-2\alpha}\right).
\end{equation}
Furthermore, the design-optimization suboptimality satisfies
\begin{equation}
\label{eq:linear-upper-bound}
\E[(\tilde\theta-\theta)^2]
-
\E[(\theta^\star_n-\theta)^2]
=
\widetilde O\!\left(n^{-1-2\beta}\right).
\end{equation}
Consequently,
\begin{equation}
\label{eq:linear-upper-combined}
\E[(\hat\theta-\theta)^2]
-
\frac{V^\star}{n}
=
\widetilde O\!\left(
n^{-1-2\alpha}
\vee
n^{-1-2\beta}
\right).
\end{equation}
\end{theorem}
Theorem~\ref{thm:linear-bounds} gives more than asymptotic semiparametric efficiency. It
characterizes the second-order gap to the oracle benchmark $V^\star/n$. The two terms have
distinct interpretations: $n^{-1-2\alpha}$ is the cost of learning the outcome regressions used in
the final regression-adjusted estimator, while $n^{-1-2\beta}$ is the cost of learning the
design-relevant second moments used to update the assignment rule.

It remains to ask whether these rates are intrinsic or merely artifacts of the analysis. The next
result shows that the two rates are minimax-sharp, up to logarithmic factors.

\begin{theorem}[Matching convergence lower bound for linear functionals]
\label{thm:linear-rate-lb}
Fix $\alpha,\beta>0$ and consider a class of linear-functional environments for which the
i.i.d.\ rate conditions in Assumption~\ref{assump:linear-rates} are tight: the conditional means and
the design-relevant second moments cannot be estimated faster than $n^{-\alpha}$ and
$n^{-\beta}$, respectively. Then there exist subfamilies of environments such that, for any
non-anticipating experiment and unbiased estimator pair $(\ALG,T)$,
\begin{equation}
\label{eq:linear-lb-main}
\sup_{P}
\left\{
\Var_P(T)
-
\frac{V^\star(P)}{n}
\right\}
\;\gtrsim\;
n^{-1-2\alpha}
\vee
n^{-1-2\beta}
\qquad
(\textnormal{up to logarithmic factors}).
\end{equation}
Equivalently, neither the model-estimation scale $n^{-1-2\alpha}$ nor the design-learning scale
$n^{-1-2\beta}$ above the oracle benchmark $V^\star/n$ can be improved in general.
\end{theorem}

The proof uses two complementary Assouad-type constructions. The first construction fixes the
arm variances and varies one arm's conditional mean over a Gaussian-location hypercube, isolating
the model-estimation cost. The second construction fixes the conditional means and varies one
arm's conditional variance, so that the oracle assignment rule itself changes across environments,
isolating the design-learning cost. Full proofs are deferred to
Appendix~\ref{app:proof-linear-rate-lb}.
Combining Theorems~\ref{thm:linear-bounds} and~\ref{thm:linear-rate-lb}, Algorithm~\ref{alg:batched-cf}
approaches the semiparametric benchmark at the minimax-optimal second-order rate:
\[
\E[(\hat\theta-\theta)^2]
=
\frac{V^\star}{n}
+
\widetilde O\!\left(
n^{-1-2\alpha}
\vee
n^{-1-2\beta}
\right),
\]
up to logarithmic factors.
\subsection{Comparison and Transition}
\label{subsec:comparison}

The preceding analysis reveals a sharp dichotomy in the achievability of semiparametric efficiency under adaptive experimentation. The fundamental distinction is whether the plug-in bias remains only second-order small, or instead vanishes entirely. Table~\ref{tab:comparison} summarizes this structural difference.
\begin{table}[h]
\centering
\caption{Achieving semiparametric efficiency via regression adjustment.}
\label{tab:comparison}
\renewcommand{\arraystretch}{1.1}
\begin{tabularx}{\linewidth}{>{\raggedright\arraybackslash}p{0.30\linewidth} >{\raggedright\arraybackslash}X >{\raggedright\arraybackslash}X}
\toprule
& \textbf{General estimands} & \textbf{Linear functionals} \\
\midrule
Plug-in bias $\E[\Delta_{b,1}\mid\cH_b]$
& $O\!\bigl(d(\hat\eta,\eta_0)^2\bigr)\neq 0$
& Exact zero \\
\midrule
Required nuisance estimation rate
& $\alpha>1/4$
& Any $\alpha>0$ \\
\midrule
Design-learning requirement
& Consistency of $\hat m_a$
& Any $\beta>0$ for variance estimation \\
\midrule
Interaction term
& Controlled by Cauchy--Schwarz
& Exact zero \\
\midrule
Second-order convergence rate
& Asymptotic efficiency under $\alpha>1/4$, but no sharp second-order rate result
& Sharp optimal rate with matching lower bound \\
\bottomrule
\end{tabularx}
\end{table}
The regression-adjustment route approaches the oracle benchmark via an estimator-side correction for dependent sequential data. However, as the general case demonstrates, this approach can place substantial statistical complexity on the final estimation step. When the plug-in bias is only second-order small, attaining the benchmark relies critically on the stringent $n^{-1/4}$ convergence threshold for nuisance estimation.

This operational reality motivates an alternative route to the same benchmark. Rather than concentrating statistical complexity in the estimator, the next section shifts this complexity directly into the assignment mechanism itself. There we show that the oracle rate $V^\star/n$ can be attained \emph{by design}, through adaptive covariate balancing, while utilizing a simple, transparent moment-based estimator and completely bypassing the nuisance-rate bottleneck in the final estimation step.


\section{Achieving Semiparametric Efficiency via Covariate Balancing}\label{sec:covariate-balancing}
\label{sec:cov-balancing}

We now turn to the second constructive route for approaching the same benchmark $V^\star/n$. Whereas Section~\ref{sec:upper-bound} approaches the benchmark through estimator-side correction, this section approaches it \emph{by design}. Crucially, this route does not require nonparametric nuisance estimation in the estimation step. The only nuisance quantity enters on the design side, through the second-moment object used to update target propensities across batches.

\subsection{Moment Condition and What Must Be Balanced}
\label{subsec:cb-moment-condition}

We consider parameters $\theta_0\in\Theta\subseteq\R^q$ defined implicitly through a general moment condition
\begin{equation}
\label{eq:moment-condition}
\E\!\left[m(\X,\A,\Y,\theta_0)\right]=0,
\end{equation}
where $m:\cX\times\cA\times\cY\times\Theta\to\R^q$ is a known function that does not depend on unknown nuisance parameters. Let
\[
M := \E\!\left[\frac{\partial m}{\partial\theta^\top}(\X,\A,\Y,\theta_0)\right]
\]
denote the nonsingular Jacobian. The naive method-of-moments estimator $\hat\theta_n$ solves the sample analog
\begin{equation}
\label{eq:naive-estimator}
\frac{1}{n}\sum_{i=1}^n m(\X_i,\A_i,\Y_i,\hat\theta_n)=0.
\end{equation}
The key advantage of this route is that the estimator itself is simple: it uses only the known moment function $m$ and does not require nonparametric nuisance estimation in the estimation step.
To analyze \eqref{eq:naive-estimator} under sequential experimental designs, we identify the covariate-dependent quantities that must be balanced. Define the arm-specific \emph{balance functions}
\begin{equation}
\label{eq:omega-def}
\omega_a(\X) \;:=\; -M^{-1}\,\E\!\left[m(\X,a,\Y,\theta_0)\mid \X\right].
\end{equation}
These functions capture the directions that must be balanced for the target estimand. In this sense,
covariate balancing is not merely balancing raw covariates; it is balancing the covariate-indexed
components of the moment condition that affect first-order estimation error.
For example, in the binary ATE case
\(
    \theta_0=\E[\Y(1)-\Y(0)],
\)
one may use the IPW moment
\[
m_{\mathrm{ATE}}(\X,\A,\Y,\theta)
=
\frac{\1\{\A=1\}}{\e_1(\X)}\Y
-
\frac{\1\{\A=0\}}{\e_0(\X)}\Y
-
\theta .
\]
The corresponding estimand-relevant balance directions are, up to the propensity normalization
embedded in the moment,
\[
    \omega_a(\X)=(-1)^{1-a}\mu_a(\X),
    \qquad
    \mu_a(\X)=\E[\Y\mid \X,\A=a].
\]
Thus, for the ATE, the design should balance the arm-specific outcome-regression functions rather
than arbitrary covariate features.
As a nonlinear example, consider the QTE
\[
    \theta_0^{\mathrm{QTE}}=q_1(\tau)-q_0(\tau),
    \qquad
    q_a(\tau):=F_{\Y(a)}^{-1}(\tau).
\]
Using the arm-specific quantile moment
\[
m_{\mathrm{Q},a}(\X,\A,\Y,\vartheta)
=
\frac{\1\{\A=a\}}{\e_a(\X)}
\bigl(\1\{\Y\le \vartheta_a\}-\tau\bigr),
\]
the corresponding balance direction is
\[
    \omega_a(\X)
    =
    \frac{F_{\Y(a)\mid \X}(q_a(\tau))-\tau}{f_a(q_a(\tau))},
\]
under standard smoothness at the target quantile. Additional examples, including policy value,
are collected in Appendix~\ref{app:cb-examples}.
 The following decomposition is the structural core of this section; its full derivation is deferred to Appendix~\ref{app:proof-three-term}.

\begin{proposition}[Structural decomposition]
\label{prop:three-term-decomp}
Under standard regularity conditions for moment-based asymptotic linearity, if unit $i$ is assigned treatment in a batch with target propensity $\e^{(b(i))}$, then
\begin{equation}
\label{eq:three-term-decomp}
\begin{aligned}
\sqrt{n}\,(\hat\theta_n-\theta_0)
&=
\underbrace{\frac{1}{\sqrt{n}}\sum_{i=1}^n \IF_{\e^\star}(\Oo_i)}_{\textnormal{(I) Oracle EIF}}
+
\underbrace{\frac{1}{\sqrt{n}}\sum_{i=1}^n \bigl[\IF_{\e^{(b(i))}}(\Oo_i)-\IF_{\e^\star}(\Oo_i)\bigr]}_{\textnormal{(II) Design Gap}} \\
&\qquad+
\underbrace{\frac{1}{\sqrt{n}}\sum_{i=1}^n \sum_{a\in\cA}
\bigl(\1\{\A_i=a\}-\e_a^{(b(i))}(\X_i)\bigr)\,\omega_a(\X_i)}_{\textnormal{(III) Balancing Remainder}}
+o_P(1).
\end{aligned}
\end{equation}
\end{proposition}
Term~(I) is the oracle EIF contribution and yields the target variance $V^\star$. Term~(II) is the cost of learning the optimal design. Term~(III) is the balancing remainder. The design problem is therefore clear: make Term~(III) small within batches and Term~(II) small across batches.

\subsection{Two-Layer Batched Balancing Design}
\label{subsec:cb-procedure}

This design mirrors Proposition~\ref{prop:three-term-decomp}: the within-batch layer targets the balancing remainder, while the across-batch layer targets the design gap. We divide the experiment into $B$ batches and proceed in two layers:
\begin{itemize}[leftmargin=2em]
\item \textbf{Within-batch balancing.} The target propensity $\e^{(b)}$ is held fixed throughout batch~$b$. As units arrive, the design enforces local balance relative to $\e^{(b)}$ by combining spatial binning of the covariate space with grouped permutations within each bin.

\item \textbf{Across-batch learning.} At the end of batch~$b$, the data from that batch are used to estimate the second-moment proxy $\hat m^{(b)}$ by standard i.i.d.\ regression. The next-batch target propensity $\e^{(b+1)}$ is then obtained by plug-in minimization of the oracle design objective.
\end{itemize}
\begin{figure}
    \centering
    \includegraphics[width=0.8\linewidth]{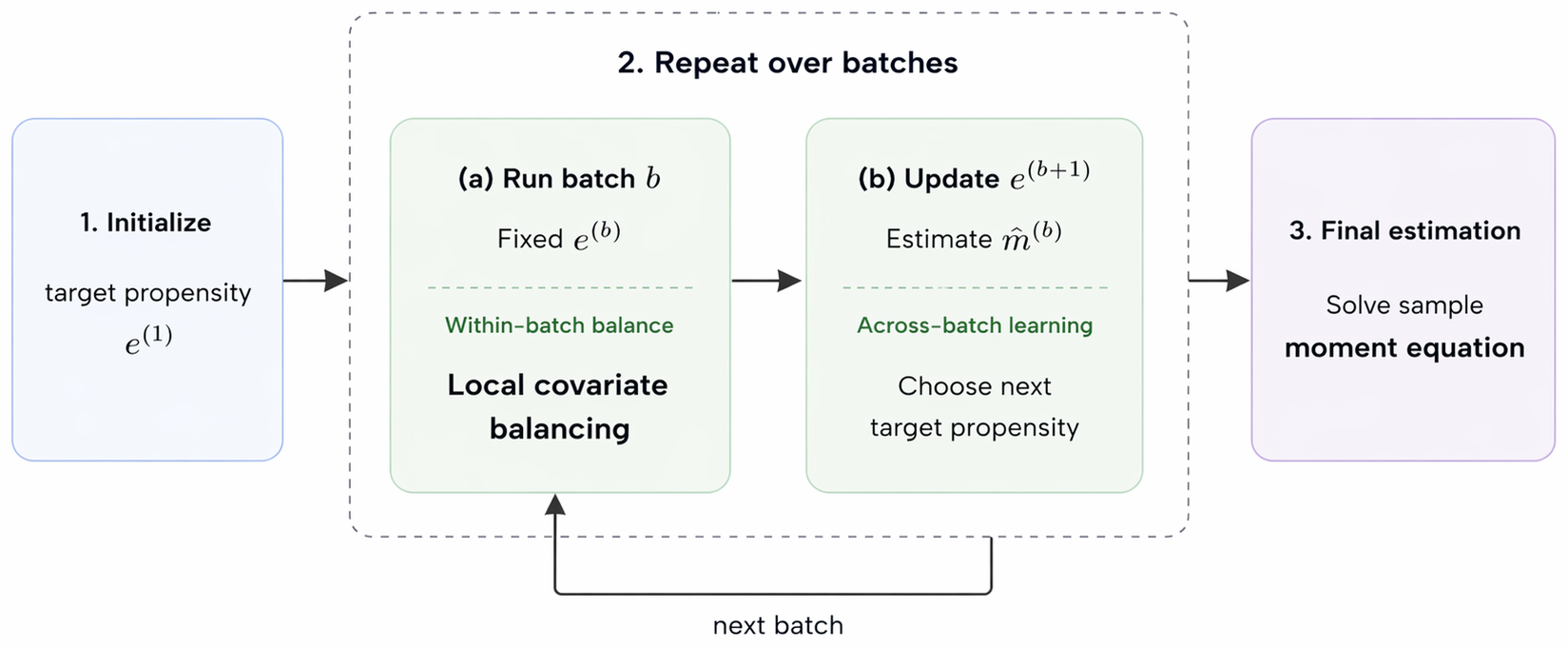}
    \caption{Illustration of the workflow of the covariate-balancing route.}
    \label{fig:cb_illu}
\end{figure}
Full implementation details---including explicit rules for randomized rounding, open-tuple boundary handling, and bin-refresh conventions---are deferred to Appendix~\ref{app:cb-alg-details}.

\subsection{Within-Batch Balancing Guarantee}
\label{subsec:cb-variance}

We first isolate the within-batch balancing layer under a fixed target propensity $\e^{(b)}$. The batch-level contribution to the balancing remainder is governed by the normalized imbalance
\begin{equation}
\label{eq:moment-imbalance}
S_b
:=
\frac{1}{\sqrt{\gamma_b}}
\sum_{t\in\mathcal I_b}
\sum_{a\in\cA}
\bigl(\1\{\A_t=a\}-\e_a^{(b)}(\X_t)\bigr)\,\omega_a(\X_t)
\;\in\; \R^q.
\end{equation}
By enforcing local balance through spatial binning and within-bin grouped permutation, the design simultaneously controls this imbalance for all Lipschitz balance functions without requiring knowledge of the unknown $\omega_a$'s themselves. Detailed variance derivations and martingale arguments are deferred to Appendix~\ref{app:proof-global-balance}.

\begin{theorem}[Within-batch balancing remainder]
\label{thm:global-balance}
Assume the balance functions $\omega_a$ are Lipschitz continuous. Under suitable geometric tuning of the bin diameter and tuple capacity for a batch of size $m$, the normalized imbalance satisfies
\begin{equation}
\label{eq:within-batch-balance}
\|S_b\| = O_P\!\left(m^{-\frac{1}{d+4}}\right)=o_P(1).
\end{equation}
\end{theorem}
Thus the within-batch balancing remainder is negligible. The rate in \eqref{eq:within-batch-balance} is dimension dependent, which is the main statistical cost of the design-side route.

\subsection{Main Efficiency Theorem}
\label{subsec:cb-batched}

Operationally, each batch functions simultaneously as an online experiment for balancing and an offline i.i.d.\ regression sample for design learning. Combining the within-batch balancing guarantee with the across-batch design updates yields the main result of this section.

\begin{theorem}[Semiparametric efficiency via covariate balancing]
\label{thm:cb-efficiency}
Assume the balance functions $\omega_a$ are Lipschitz, the covariate space is compact, and the second-moment proxies satisfy
\begin{equation}
\label{eq:cb-second-moment-rate}
\E\!\left[\sum_{a\in\cA}\bigl(\hat m_a^{(b)}(\X)-m_a(\X)\bigr)^2\right]
\lesssim \gamma_b^{-2\beta}
\end{equation}
for some $\beta>0$. Then, under the two-layer batched adaptive balancing design,
\begin{equation}
\label{eq:cb-efficiency-bound}
\E\!\left[(\hat\theta_n-\theta_0)^2\right]
=
\frac{V^\star}{n}
+
o\!\left(\frac{1}{n}\right).
\end{equation}
\end{theorem}
The result follows by controlling each term in Proposition~\ref{prop:three-term-decomp}. First, Theorem~\ref{thm:global-balance} makes the within-batch balancing remainder negligible. Second, the plug-in updates based on \eqref{eq:cb-second-moment-rate} make the cumulative design gap negligible. Third, the oracle EIF term contributes the target variance $V^\star$. Detailed rate calculations are deferred to Appendix~\ref{app:proof-cb-efficiency}.

\subsection{Comparison and Implications}
\label{subsec:cb-comparison}

This second constructive route provides a sharp contrast to the regression-adjustment strategy of Section~\ref{sec:upper-bound}. Both attain the same benchmark, but Section~\ref{sec:upper-bound} operates through estimator-side correction, whereas the present section achieves the benchmark by design. By placing statistical complexity in the assignment mechanism, the balancing route avoids the familiar $n^{-1/4}$ estimator-side nuisance bottleneck. The final estimator remains simple and robust. The tradeoff is geometric: the balancing remainder decays at the dimension-dependent rate in \eqref{eq:within-batch-balance}. Consequently, the design-side route is especially attractive in structured, lower-dimensional settings, whereas estimator-side correction is often preferable when reliable high-dimensional nonparametric learners are available.

A central operational takeaway is that both constructive routes require design updates only at batch boundaries. This aligns directly with the lower-bound theory of Section~\ref{sec:lower-bound}, which shows that the relevant sufficient summary for statistical efficiency is the induced average propensity score. By separating the timescales into within-batch balancing and across-batch learning, the design yields low switching costs and natural compatibility with delayed feedback. Both routes target the same average-propensity benchmark, but place complexity in different locations.


\section{Numerical Study}\label{sec:numerical}

We complement the theoretical analysis with a numerical study covering synthetic simulations and a real-data application. The synthetic simulations examine how the two constructive routes---batched cross-fitted regression adjustment (RA, Algorithm~\ref{alg:batched-cf}) and batched adaptive covariate balancing (CB, Algorithm~\ref{alg:batched-cb})---approach the semiparametric benchmark $V^\star$ in controlled environments where heteroskedasticity, cross-arm variance contrast, and covariate dimension can be varied. The real-data application then applies the same design principles to an AI medical-assistant evaluation, where the target is naturally multi-armed: we estimate the vector treatment effect of several assistants relative to a common baseline, using covariates that are operationally interpretable.

Across all experiments, we compare the adaptive designs with three reference points. The first is uniform i.i.d.\ assignment with the cross-fitted EIF estimator, which serves as the no-design baseline. The second is the oracle Neyman i.i.d.\ design with the same estimator, denoted oracle\_iid, which represents the i.i.d.\ design benchmark when the optimal propensity is known. The third is the oracle balanced design, denoted \texttt{oracle\_cb}, which uses the true generalized Neyman propensity inside Algorithm~\ref{alg:batched-cb} and represents the design-side benchmark for the balancing route. The two feasible adaptive algorithms are denoted \texttt{alg1\_ra} for RA and \texttt{alg2\_cb} for CB. We report normalized mean squared error $n\cdot\mathrm{MSE}$, paired differences against \texttt{uniform} with $\pm 2\,\mathrm{SE}$ intervals, and a verdict that marks a method as \textbf{BEATS} when the paired improvement exceeds two standard errors, \emph{ties} when the interval covers zero, and \emph{LOSES} when the paired loss exceeds two standard errors.

\subsection{Synthetic experiments}\label{subsec:numerical-synth}

We first construct a continuous-Gaussian DGP indexed by an ambient dimension $d\in\{2,3,8\}$. Covariates are drawn as $X\sim\mathcal{N}(0,I_d)$ and there are $K=4$ arms. The conditional mean and variance of the outcome are
\[
\mu_a(X)=\alpha_a+\beta_a^\top X,\qquad
\log\sigma_a^2(X)=c_a+\gamma_a^\top X_{1:2},\qquad
Y\mid X,A=a\sim\mathcal{N}\!\bigl(\mu_a(X),\sigma_a^2(X)\bigr).
\]
The first two covariates drive the signal: the coefficients $\beta_a$ load primarily on dimensions~$1$ and $2$, and the log-variance slopes $\gamma_a$ are nonzero only along these two dimensions. Cross-arm heteroskedasticity is induced through the intercept schedule $c=(-2.0,-0.7,0.7,2.0)$, which creates substantial differences in arm-specific noise levels. The remaining $d-2$ covariates are weakly informative or near-noise, yet they are still passed to every nuisance learner. Thus, the high-dimensional regime tests whether the adaptive designs can exploit a low effective dimension in the presence of irrelevant covariates.

We instantiate $n_{\mathrm{pop}}=5\times 10^4$ population units and draw $n=10^4$ sample units per replication, with $3000$ replications per cell. Replications share random seeds across methods, so the paired differences below are computed under common random numbers. We evaluate three estimands. The scalar ATE compares arms $2$ and $4$,
\[
\theta_{\mathrm{ATE}}=\mathbb E[\mu_4(X)-\mu_2(X)].
\]
The marginal median quantile treatment effect compares the same arm pair at $\tau=0.5$,
\[
\theta_{\mathrm{QTE}}=F_{Y(4)}^{-1}(0.5)-F_{Y(2)}^{-1}(0.5).
\]
The vector ATE compares all three active arms against the common reference arm $a_r=3$,
\[
\theta_{\mathrm{vATE}}=(\mathbb E[\mu_a(X)-\mu_3(X)])_{a\ne 3},
\]
and is evaluated under trace MSE. For each estimand, the RA estimator is the cross-fitted AIPW form developed in Section~\ref{sec:upper-bound}; for the QTE, it uses an EIF-derived AIPW estimator of the marginal CDF followed by grid inversion. The CB estimator is the moment-based estimator of Section~\ref{sec:cov-balancing}: an inverse-propensity-weighted average for linear estimands and a weighted-quantile estimator for the QTE. Nuisance regressions use gradient-boosted trees (XGBoost) with $50$ trees and depth $3$.

Table~\ref{tab:synth-main} reports the nine main cells. The oracle benchmarks confirm sizable design value: relative to uniform assignment, the oracle i.i.d.\ design \texttt{oracle\_iid} improves normalized MSE in every cell, and the oracle balanced design \texttt{oracle\_cb} also improves performance throughout. The feasible RA algorithm \texttt{alg1\_ra} is consistently beneficial across all estimands and dimensions, tracking the oracle i.i.d.\ benchmark especially closely for scalar estimands. The feasible CB algorithm \texttt{alg2\_cb} is beneficial in most cells, with performance depending more strongly on the geometry of the balancing step: it ties uniform for the ATE at $d=2$ and loses for the QTE at $d=8$ under the default all-feature binning rule. This pattern is consistent with the theory in Section~\ref{sec:cov-balancing}, where the balancing remainder depends on the effective dimension of the covariate partition.

\begin{table}[htbp]
\centering
\caption{Synthetic experiments at $n=10^4$, $3000$ replications. Each block reports normalized MSE ($n\cdot\mathrm{MSE}$) and the paired improvement over uniform $\pm 2\,\mathrm{SE}$. Bold entries beat uniform with statistical significance; italics lose; un-emphasized entries are ties.}
\label{tab:synth-main}
\scriptsize
\setlength{\tabcolsep}{3.2pt}
\renewcommand{\arraystretch}{0.92}
\resizebox{\textwidth}{!}{
\begin{tabular}{l|rrr|rrr|rrr}
\toprule
& \multicolumn{3}{c|}{ATE ($\mu_4-\mu_2$)}
& \multicolumn{3}{c|}{QTE ($\tau=0.5$)}
& \multicolumn{3}{c}{Vector ATE (trace MSE)} \\
Method
& $d=2$ & $d=3$ & $d=8$
& $d=2$ & $d=3$ & $d=8$
& $d=2$ & $d=3$ & $d=8$ \\
\midrule
$V^\star$
& 14.7 & 15.2 & 29.9
& 21.9 & 20.8 & 20.9
& 52.0 & 63.0 & 72.8 \\
uniform
& 20.3 & 21.7 & 38.8
& 28.7 & 26.4 & 25.0
& 76.4 & 99.6 & 118.4 \\
\midrule
\texttt{oracle\_iid}
& \textbf{14.4} & \textbf{15.5} & \textbf{30.3}
& \textbf{21.6} & \textbf{21.4} & \textbf{19.6}
& \textbf{52.1} & \textbf{67.5} & \textbf{78.0} \\
\quad paired $\pm 2$SE
& +5.9$\pm$1.3 & +6.2$\pm$1.4 & +8.5$\pm$2.5
& +7.1$\pm$1.8 & +5.1$\pm$1.7 & +5.4$\pm$1.6
& +24.4$\pm$3.6 & +32.0$\pm$4.9 & +40.4$\pm$5.3 \\
\texttt{oracle\_cb}
& \textbf{17.2} & \textbf{19.2} & \textbf{31.3}
& \textbf{25.3} & \textbf{23.2} & \textbf{22.8}
& \textbf{63.0} & \textbf{81.0} & \textbf{83.1} \\
\quad paired $\pm 2$SE
& +3.1$\pm$1.3 & +2.4$\pm$1.5 & +7.5$\pm$2.5
& +3.4$\pm$2.0 & +3.3$\pm$1.7 & +2.2$\pm$1.7
& +13.4$\pm$3.7 & +18.6$\pm$5.0 & +35.2$\pm$5.4 \\
\midrule
\texttt{alg1\_cf}
& \textbf{15.7} & \textbf{16.5} & \textbf{32.7}
& \textbf{23.6} & \textbf{23.2} & \textbf{22.3}
& \textbf{62.6} & \textbf{73.7} & \textbf{91.2} \\
\quad paired $\pm 2$SE
& +4.6$\pm$1.3 & +5.2$\pm$1.4 & +6.1$\pm$2.6
& +5.1$\pm$1.9 & +3.3$\pm$1.8 & +2.7$\pm$1.7
& +13.8$\pm$3.8 & +25.8$\pm$5.0 & +27.2$\pm$5.6 \\
\texttt{alg2\_cb}
& 19.1 & \textbf{19.2} & \textbf{35.1}
& \textbf{25.9} & \textbf{24.6} & \emph{28.9}
& \textbf{67.0} & \textbf{89.5} & \textbf{105.3} \\
\quad paired $\pm 2$SE
& +1.3$\pm$1.5 & +2.5$\pm$1.5 & +3.8$\pm$2.6
& +2.8$\pm$2.0 & +1.9$\pm$1.8 & $-3.9\pm$1.9
& +9.5$\pm$3.9 & +10.1$\pm$5.3 & +13.1$\pm$5.8 \\
\bottomrule
\end{tabular}
}
\end{table}

The table highlights three empirical lessons. First, \texttt{oracle\_iid} delivers the largest and most stable improvements, showing that the average-propensity benchmark is attainable when the optimal propensity is known. Second, \texttt{alg1\_ra} captures a large share of this gain using only batched updates, suggesting that the plug-in design based on learned nuisance functions is effective in this DGP. Third, \texttt{alg2\_cb} illustrates the main tradeoff of the design-side route. When the balancing partition is aligned with the low-dimensional structure of the DGP, CB can approach its oracle balancing benchmark; when the partition is formed over all ambient covariates, irrelevant dimensions can reduce the within-bin sample size and inflate the balancing remainder.

To isolate this high-dimensional effect, we re-run the $d=8$ cells with two binning schemes. The \emph{full} scheme bins on all eight covariates, matching the default CB implementation in Table~\ref{tab:synth-main}. The \emph{top-2} scheme bins only on dimensions~$1$--$2$, while still allowing the nuisance regressions to use all eight covariates. This top-2 scheme corresponds to the effective-dimension principle suggested by the rate in Theorem~\ref{thm:cb-efficiency}: the geometric part of the balancing design should be built on the covariates that drive treatment-effect or variance heterogeneity.

\begin{table}[H]
\centering
\caption{$d=8$, $n=10^4$: full versus top-2 design across estimands. Diff $=$ nMSE(full) $-$ nMSE(top-2); a positive diff means top-2 binning is more accurate.}
\label{tab:synth-highd}
\small
\begin{tabular}{lrrr}
\toprule
Estimand & RA \texttt{alg1\_ra} (full vs.\ top-2) & CB \texttt{alg2\_cb} (full vs.\ top-2) & Interpretation \\
\midrule
ATE        & $32.7$ vs $32.4$~~~(+0.3$\pm$2.6)   & $35.1$ vs $33.5$~~~(+1.6$\pm$2.6)             & CB modestly improves \\
QTE        & $22.3$ vs $21.8$~~~(+0.5$\pm$1.6)   & $28.9$ vs $23.6$~~~(\textbf{+4.6$\pm$1.9})    & CB requires top-2 \\
vector ATE & $91.2$ vs $90.1$~~~(+1.1$\pm$4.7)   & $105.3$ vs $93.5$~~~(\textbf{+8.9$\pm$5.1})   & CB requires top-2 \\
\bottomrule
\end{tabular}
\end{table}

Table~\ref{tab:synth-highd} confirms this interpretation. RA is nearly insensitive to the binning choice because the assignment rule is learned through flexible regressions that can screen out the six near-noise dimensions. CB is more sensitive. For the QTE and vector ATE, top-2 binning substantially improves normalized MSE relative to full binning, closing much of the gap created by the all-feature partition. This finding gives an implementation lesson for balancing-based sequential designs: in moderate or high ambient dimension, the balancing partition should be constructed using a low-dimensional summary of the covariates, obtained from prior knowledge, pilot data, or a preliminary heterogeneity screen. In contrast, the RA route can rely more directly on the implicit regularization of the nuisance learner.

\subsection{Application: HELPMed AI medical-assistant evaluation}\label{subsec:numerical-helpmed}

We next apply the same methodology to an evaluation experiment for AI medical assistants from the HELPMed study \citep{bean2026reliability}. The dataset contains $n_{\mathrm{pop}}=2400$ interactions between $1298$ participants and one of four AI assistants ($K=4$) across $10$ clinically grounded scenarios. The four arms correspond to a baseline assistant and three competitor variants. The outcome $Y$ is binary and indicates whether the participant's chosen disposition matches the gold-standard answer for the scenario. The original experiment used uniform $1/K$ randomization, yielding exactly $600$ interactions per arm; this provides a natural empirical reference point for assessing adaptive designs. Since some participants contribute multiple interactions, the estimand in this exercise is interpreted at the interaction level, matching the unit at which treatment assignment and outcomes are recorded.

We use a $d=6$ covariate vector: \texttt{scenario\_id} (10 clinical vignettes, such as chest pain, head injury triage, and pediatric fever), \texttt{age\_band} (six ordinal bands), \texttt{education} (five ordinal levels), \texttt{online\_health} (frequency of online health-information seeking), \texttt{llms\_health\_info} (frequency of LLM use for health questions), and \texttt{sex\_male}. To construct a low-dimensional design variant, we run a causal-forest variable-importance screen at the end of each batch using the accumulated data, and use the two highest-ranked covariates to update the next-batch design. In most replications and batches, this procedure selects \texttt{scenario\_id} and \texttt{age\_band}, suggesting that these two variables carry most of the treatment-effect heterogeneity signal. This creates a real-data analogue of the effective-dimension issue in Section~\ref{subsec:numerical-synth}: the balancing design may either operate on all six covariates or use an adaptively selected low-dimensional partition.

The target is the vector ATE of the three competitor assistants relative to the baseline assistant,
\(
\theta=(\mathbb E[\mu_a(X)-\mu_3(X)])_{a\in\{1,2,4\}}\in\mathbb{R}^3,
\)
evaluated under trace MSE. This estimand captures the joint comparison needed when the goal is to rank several assistant variants against a common deployment baseline. We use the HELPMed sample as an empirical finite population, draw resamples of size $n\in\{1200,2400,4800\}$ in each replication, and evaluate the two adaptive algorithms together with the two oracle ceilings over $5000$ replications per cell. The nuisance regressions use the same gradient-boosted-tree configuration as in the synthetic simulations. For both adaptive routes, we compare a full-feature design with a batchwise top-2 design. In the full-feature design, the propensity update uses all six covariates and the CB route forms bins over the full covariate vector. In the top-2 design, the next-batch design is based on the two covariates selected by the causal-forest screen at the previous batch boundary; for CB, these two covariates define the spatial partition used for balancing.

\begin{table}[H]
\centering
\caption{HELPMed vector ATE under normalized trace MSE. The target benchmark is $V^\star=4.515$. The top-2 design uses batchwise causal-forest screening to select two covariates for the next-batch design; in most batches, the selected variables are \texttt{scenario\_id} and \texttt{age\_band}. Bold entries beat uniform with statistical significance under paired $\pm 2\,\mathrm{SE}$ intervals; un-emphasized entries are ties.}
\label{tab:helpmed}
\small
\setlength{\tabcolsep}{4.5pt}
\renewcommand{\arraystretch}{0.96}
\begin{tabular}{l|ccc|ccc}
\toprule
& \multicolumn{3}{c|}{Full-feature design ($d=6$)}
& \multicolumn{3}{c}{Batchwise top-2 design} \\
$n$ & $1200$ & $2400$ & $4800$ & $1200$ & $2400$ & $4800$ \\
\midrule
uniform
& 5.43 & 5.27 & 5.10
& 5.43 & 5.27 & 5.10 \\
\midrule
\texttt{oracle\_iid}
& \textbf{5.08} & \textbf{5.00} & \textbf{4.81}
& --- & --- & --- \\
\quad paired $\pm 2$SE
& +0.35$\pm$0.20 & +0.27$\pm$0.19 & +0.29$\pm$0.19
& & & \\
\texttt{oracle\_cb}
& \textbf{4.73} & \textbf{4.77} & \textbf{4.84}
& --- & --- & --- \\
\quad paired $\pm 2$SE
& +0.70$\pm$0.20 & +0.49$\pm$0.19 & +0.26$\pm$0.20
& & & \\
\texttt{alg1\_ra}
& 5.28 & \textbf{5.02} & 5.00
& \textbf{5.20} & \textbf{4.99} & \textbf{4.82} \\
\quad paired $\pm 2$SE
& +0.15$\pm$0.20 & +0.24$\pm$0.20 & +0.10$\pm$0.20
& +0.23$\pm$0.21 & +0.28$\pm$0.20 & +0.28$\pm$0.20 \\
\texttt{alg2\_cb}
& \textbf{5.03} & \textbf{4.84} & 4.97
& \textbf{5.03} & \textbf{4.87} & \textbf{4.81} \\
\quad paired $\pm 2$SE
& +0.39$\pm$0.21 & +0.43$\pm$0.19 & +0.13$\pm$0.20
& +0.40$\pm$0.21 & +0.40$\pm$0.19 & +0.29$\pm$0.20 \\
\bottomrule
\end{tabular}
\end{table}

Table~\ref{tab:helpmed} shows that the average-propensity design principle continues to produce meaningful gains in a realistic multi-arm evaluation. The two oracle designs improve over uniform randomization in all three sample-size regimes, with \texttt{oracle\_cb} giving the largest improvement at $n=1200$ and \texttt{oracle\_iid} remaining stable across sample sizes. Among feasible adaptive designs, \texttt{alg2\_cb} delivers the strongest gains at $n=1200$ and $n=2400$, reducing normalized trace MSE from $5.43$ to about $5.03$ at $n=1200$ and from $5.27$ to about $4.84$--$4.87$ at $n=2400$. At $n=4800$, the batchwise top-2 versions of both adaptive routes remain statistically beneficial, while the full-feature versions are closer to uniform within Monte Carlo error. This pattern is consistent with the synthetic high-dimensional experiment: when the balancing partition spreads a few thousand observations over a large Cartesian grid, the finite-sample balancing remainder can offset part of the design gain.

The feature-screening comparison is also informative. For \texttt{alg1\_ra}, using the batchwise top-2 screen in the design update slightly improves performance in all three sample-size regimes, although the differences are modest. For \texttt{alg2\_cb}, batchwise top-2 binning performs similarly to full-feature binning at $n=1200$ and $n=2400$, and improves clearly at $n=4800$, where the full $10\times 6\times 5\times 6\times 6\times 2=10{,}800$-cell grid is especially sparse relative to the realized sample size. The application therefore reinforces the implementation lesson from Table~\ref{tab:synth-highd}: a balancing-based design should use a low-dimensional, data-adaptive partition, while flexible regression learners can continue to use the richer covariate set.

The managerial implications are straightforward. In this application, the oracle benchmarks suggest that efficient allocation could reduce normalized trace MSE by roughly $5$--$13\%$ relative to the original uniform design, and the feasible adaptive algorithms deliver statistically significant reductions of about $5$--$8\%$ in their strongest cells. Since normalized MSE plays the role of the variance constant in the $1/n$ rate, these gains translate approximately into the same percentage reduction in the number of interactions needed to achieve a fixed precision target. The gains are especially relevant for multi-assistant evaluation, where decision-makers care about the joint ranking of several variants rather than a single binary contrast. The empirical pattern also suggests a deployable workflow: after each batch, run a lightweight heterogeneity screen to select a small set of design covariates, use those covariates in the next-batch propensity update and CB partition, and update assignment probabilities only at batch boundaries.

\section{Conclusion}\label{sec:conclusion}

This paper develops a unified framework for characterizing and designing efficient  experiments. The central object is the induced average propensity score: every non-anticipating design induces such a score, and the corresponding semiparametric efficiency benchmark is given by the i.i.d. efficiency bound evaluated at this induced propensity. This result shows that the statistical effect of complex adaptive, balanced, or otherwise dependent assignment mechanisms can be summarized through a single design object. It also turns experimental design into the problem of choosing or learning an efficient average propensity, yielding generalized Neyman allocation principles for scalar and vector-valued causal estimands.

We further show that this benchmark can be approached by practical batched adaptive designs through two complementary routes. The regression-adjustment route learns the target propensity and uses efficient-influence-function-based estimation, while the covariate-balancing route achieves the same benchmark through assignment-side balance and simple moment-based estimation. Both approaches require only a small number of policy updates, making them compatible with delayed feedback and operational deployment. Taken together, the results suggest that average propensity provides a useful language for understanding the statistical limits of sequential experimentation and for designing efficient multi-treatment experiments in modern applications such as AI-enabled service evaluation.


\printendnotes

\bibliographystyle{chicago}
\bibliography{citation}

\setcounter{section}{0}
\renewcommand{\thesection}{\Alph{section}}
\renewcommand{\thesubsection}{\Alph{section}.\arabic{subsection}}
\renewcommand{\thesubsubsection}{\Alph{section}.\arabic{subsection}.\arabic{subsubsection}}

\section{Algorithms}
\label{app:algorithms}

This appendix gives the formal algorithmic specification of the two routes studied in the main text: the batched cross-fitted regression-adjustment estimator analyzed in Section~\ref{sec:regression-adjustment} and the batched adaptive covariate-balancing procedure analyzed in Section~\ref{sec:covariate-balancing}. Both algorithms share a common batch schedule and cross-fitting convention, which we state once for reference.

\paragraph{Common batch schedule.}
We index batches by $b=1,\dots,B$ with $B\asymp\log n$ and geometric batch sizes $\gamma_b\asymp n^{b/B}$. This makes the final batch contain a constant fraction of all units while keeping the number of design updates logarithmic in the sample size. All population expectations inside design updates are replaced by empirical averages over the pooled data from batches $\le b$.

\paragraph{Cross-fitting convention.}
Within each batch, the data are split into $S$ folds (typically $S=2$). For each fold $s$, nuisance estimates are fit on the remaining $S-1$ folds and then applied to fold $s$. The final estimator averages the $S$ fold-specific estimates to cancel first-order plug-in bias.

\subsection{Batched Cross-Fitted Regression Adjustment}
\label{app:alg-reg-adj}
\label{app:alg-details}

Algorithm~\ref{alg:batched-cf} implements the regression-adjustment route. The procedure alternates between (i)~updating the logging propensity $\e^{(b)}$ at the start of each batch, (ii)~running the batch under this propensity, and (iii)~refitting nuisances on the pooled data at the end. The final estimator is a cross-fitted EIF plug-in.

\begin{algorithm}[H]
\caption{Batched Cross-Fitted Regression Adjustment}
\label{alg:batched-cf}
\KwIn{Sample size $n$; batch count $B\asymp\log n$ and sizes $\{\gamma_b\}_{b=1}^B$ with $\sum_b\gamma_b=n$; cross-fit folds $S$; propensity floor $\varepsilon>0$; feasible class $\cE$; moment function $m(\X,\A,\Y,\theta)$.}
\KwOut{Estimator $\hat\theta_n$ of $\theta(P_0)$.}
Initialize $\e^{(1)}\in\cE$ with $\e_a^{(1)}(x)\ge\varepsilon$ for all $(a,x)$\;
\For{$b=1,\dots,B$}{
  \tcp{Run batch $b$ under logging propensity $\e^{(b)}$}
  \For{each incoming unit $i$ in batch $b$}{
    Observe $\X_i$, draw $\A_i\sim\e^{(b)}(\cdot\mid\X_i)$, record $\Y_i$\;
  }
  \tcp{End of batch: refit nuisances with cross-fitting}
  Split batch-$b$ data into $S$ folds $\{I_1,\dots,I_S\}$\;
  \ForEach{fold $s\in\{1,\dots,S\}$}{
    Fit $\hat g^{(b,-s)}$ and $\{\hat v_a^{(b,-s)}\}_{a\in\cA}$ on pooled data $\bigcup_{b'\le b}\bigcup_{s'\neq s}I_{s'}^{(b')}$ using a nonparametric regressor satisfying Assumption~\ref{assump:general-reg} or Assumption~\ref{assump:linear-rates}\;
    Fit $\hat m_a^{(b,-s)}(x)$ (second-moment proxy) on the same pooled data\;
  }
  \tcp{Compute next-batch design by plug-in optimization}
  $\displaystyle \e^{(b+1)} \leftarrow \argmin_{\e\in\cE,\ \e_a(x)\ge\varepsilon}\ \E_n\!\left[\sum_{a\in\cA}\frac{\hat m_a^{(b,-s)}(\X)}{\e_a(\X)}\right]$\;
  \tcp*{Scalar-target closed form: $\e_a^{(b+1)}(x)\propto\max\{\varepsilon,\sqrt{\hat m_a^{(b,-s)}(x)}\}$, renormalized}
}
\For{$s=1,\dots,S$}{
  Compute $\hat\theta^{(s)}=\E_n\!\left[\hat g^{(-s)}(\X)+\sum_a\frac{\1\{\A=a\}}{\e^{(b(i),-s)}(\A\mid\X)}\,\hat v_a^{(-s)}(\X,\Y)\right]$ over fold $s$, using the realized logging propensity\;
}
\Return $\hat\theta_n=\frac{1}{S}\sum_{s=1}^S\hat\theta^{(s)}$\;
\end{algorithm}

\paragraph{Design update.}
In the design-update step, we minimize the plug-in design objective over the feasible class $\cE$ subject to the floor $\e_a(x)\ge\varepsilon$. For the scalar-target case this reduces to the closed-form generalized Neyman allocation
\[
\e_a^{(b,-s)}(x) \;\propto\; \max\!\left\{\varepsilon,\ \sqrt{\hat m_a^{(b,-s)}(x)}\right\},
\]
renormalized so that $\sum_a \e_a^{(b,-s)}(x)=1$. For vector targets the optimization is convex and solvable by standard interior-point or projected-gradient methods.

\paragraph{Nuisance estimation.}
The fits $\hat g^{(b,-s)}$ and $\hat v_a^{(b,-s)}$ use the analyst's preferred nonparametric regression method (e.g., cross-validated series, random forest, or neural network). The only requirement is the rate condition in Assumption~\ref{assump:general-reg} or Assumption~\ref{assump:linear-rates}, as appropriate.

\paragraph{Final estimator.}
The realized logging propensity $\e^{(b(i),-s)}(\X_i)$ is known by construction, so no propensity estimation is required. The final estimator is the simple average of the $S$ fold-specific cross-fitted EIF estimators.

\subsection{Batched Adaptive Covariate Balancing}
\label{app:alg-cb}
\label{app:cb-alg-details}

Algorithm~\ref{alg:batched-cb} implements the covariate-balancing route. Compared to Algorithm~\ref{alg:batched-cf}, assignment within each batch is no longer i.i.d.\ draws from $\e^{(b)}$; instead, a two-layer design (spatial binning plus within-bin grouped permutation) enforces near-exact local balance at the target propensity. The between-batch design updates are identical in form to Algorithm~\ref{alg:batched-cf}, but the inner estimator is a method-of-moments estimator based on the balance condition rather than a plug-in EIF estimator.

\begin{algorithm}[H]
\caption{Batched Adaptive Covariate Balancing}
\label{alg:batched-cb}
\KwIn{Sample size $n$; batch count $B\asymp\log n$ and sizes $\{\gamma_b\}_{b=1}^B$; spatial bin diameter $h=h_n$; tuple capacity $G=G_n$; propensity floor $\varepsilon>0$; feasible class $\cE$; moment function $m(\X,\A,\Y,\theta)$.}
\KwOut{Estimator $\hat\theta_n$ solving $\frac{1}{n}\sum_i m(\X_i,\A_i,\Y_i,\hat\theta_n)=0$.}
Partition the covariate space $\cX$ into bins $\{\mathcal B_\ell\}_{\ell=1}^{M_n}$ of maximum diameter $h$, with $M_n\asymp h^{-d}$\;
For each bin $\mathcal B_\ell$, fix a reference point $x_\ell\in\mathcal B_\ell$\;
Initialize $\e^{(1)}\in\cE$ with $\e_a^{(1)}(x)\ge\varepsilon$\;
\For{$b=1,\dots,B$}{
  \ForEach{bin $\mathcal B_\ell$}{
    Open an empty tuple buffer $\cT_\ell$ of capacity $G$\;
  }
  \For{each incoming unit $i$ in batch $b$}{
    Locate its bin $\ell(i)$ such that $\X_i\in\mathcal B_{\ell(i)}$\;
    Append $i$ to $\cT_{\ell(i)}$\;
    \If{$|\cT_{\ell(i)}|=G$}{
      \tcp{Randomized rounding: choose a count vector}
      Let $\eta^{(b)}_\ell:=\e^{(b)}(x_{\ell(i)})\in\Delta(\cA)$\;
      Draw $L\in\Z_+^{\K}$ with $\sum_a L_a=G$ and $\E[L]=G\eta^{(b)}_\ell$ via $L_a=\lfloor G\eta^{(b)}_{\ell,a}\rfloor+\1\{a\in S\}$, where $|S|=G-\sum_a\lfloor G\eta^{(b)}_{\ell,a}\rfloor$ and $\PP(a\in S)$ equals the fractional part of $G\eta^{(b)}_{\ell,a}$\;
      \tcp{Uniform permutation within the tuple}
      Form the multiset containing $L_a$ copies of arm $a$; assign the $G$ positions of $\cT_{\ell(i)}$ by a uniform random permutation of this multiset\;
      Record outcomes $\{\Y_j\}_{j\in\cT_{\ell(i)}}$; close tuple\;
    }
  }
  \tcp{End of batch: handle tail tuples with the same rounding+permutation scheme}
  \ForEach{non-empty tail tuple}{Apply randomized rounding to $\eta^{(b)}_\ell$ scaled to the realized tuple length $\ell<G$ and assign by uniform permutation\;}
  \tcp{Update design}
  Fit second-moment proxy $\hat m_a^{(b)}$ on pooled data from batches $\le b$\;
  $\displaystyle \e^{(b+1)} \leftarrow \argmin_{\e\in\cE,\ \e_a(x)\ge\varepsilon}\ \E_n\!\left[\sum_{a\in\cA}\frac{\hat m_a^{(b)}(\X)}{\e_a(\X)}\right]$\;
}
\Return $\hat\theta_n$ solving $\frac{1}{n}\sum_{i=1}^n m(\X_i,\A_i,\Y_i,\hat\theta_n)=0$\;
\end{algorithm}

\paragraph{Batch schedule.}
As in Algorithm~\ref{alg:batched-cf}, we use geometric batches with $B\asymp\log n$ and $\gamma_b\asymp n^{b/B}$.

\paragraph{Spatial bins.}
The bin partition $\{\mathcal B_\ell\}$ is fixed once and reused across all batches; only the within-bin target propensity $\eta^{(b)}_\ell=\e^{(b)}(x_\ell)$ changes with $b$. Treating every unit falling in $\mathcal B_\ell$ as having the same target propensity $\e^{(b)}(x_\ell)$ introduces an $O(h)$ Lipschitz approximation error that is controlled in Appendix~\ref{app:proof-global-balance}.

\paragraph{Grouped permutation and rounding.}
Because $G\cdot\e^{(b)}_a(x_\ell)$ is generally non-integer, the rounding step draws a count vector $L\in\mathbb{Z}_+^{\K}$ with $\sum_a L_a=G$ and $\E[L\mid\cF_{r-1}]=G\cdot\e^{(b)}(x_\ell)$. The rounding variance is absorbed into the term~(II) of the tuple-variance decomposition (Proposition~\ref{prop:cb-tuple-variance}).

\paragraph{Handling tail tuples.}
At the end of each batch, each bin may have an incomplete (``tail'') tuple of length $\ell<G$. Tail-tuple assignments also use uniform permutation with count vector drawn from the rounding scheme; their contribution to the balancing remainder is bounded by the $M_n G/n$ term in Theorem~\ref{thm:global-balance}.

\paragraph{Design update.}
The update of $\e^{(b)}$ is identical in form to that of Algorithm~\ref{alg:batched-cf}: minimize the plug-in $\sum_a \hat m^{(b)}_a(x)/\e_a(x)$ objective subject to the propensity floor. In the scalar-target case this reduces to the closed-form generalized Neyman allocation $\e_a^{(b+1)}(x)\propto\max\{\varepsilon,\sqrt{\hat m^{(b)}_a(x)}\}$ after renormalization.

\paragraph{Tuning.}
Theorem~\ref{thm:global-balance} suggests the tuning $h\asymp n^{-1/(d+4)}$ and $G\asymp n^{2/(d+4)}$, which balances the three error sources in \eqref{eq:global-balance-bound} and is used throughout the analysis. Within a batch of size $\gamma_b$, the same tuning is applied with $\gamma_b$ in place of $n$.

\paragraph{Final estimator.}
Since the assignment mechanism is known by construction, no propensity estimation is required for the final estimator. The method-of-moments estimator $\hat\theta_n$ solving $\tfrac{1}{n}\sum_i m(\X_i,\A_i,\Y_i,\hat\theta_n)=0$ does not use any nuisance fit; nuisance estimation enters only through the design-update step.

\section{Efficient Influence Functions for the Canonical Examples}
\label{app:eif-examples}

This appendix derives, from first principles, the efficient influence function (EIF) in the fixed i.i.d.\ benchmark model for each of the three canonical estimands considered in the main text: the average treatment effect (ATE), the quantile treatment effect (QTE), and the policy value. Each derivation follows the same template: (i) state the identifying moment condition, (ii) compute the pathwise derivative of $\theta$ along a regular one-dimensional submodel, (iii) identify the unique Riesz representer in the observed-data tangent space $\cT(\e)$ characterized in Appendix~\ref{app:proof-eif-structure}, and (iv) specialize the resulting semiparametric efficiency bound $V(\bar e)/n$ to the example. Together these three derivations show that the general EIF structure theorem reduces, in each concrete case, to a closed-form influence function whose second moment is the oracle variance used throughout Section~\ref{sec:regression-adjustment} and Section~\ref{sec:covariate-balancing}.

\subsection{Setup and general EIF template}
\label{app:eif-setup}

Recall from Section~\ref{sec:formulation} that the observed data are $\Oo=(\X,\A,\Y)$ with $\A\mid\X\sim\e(\cdot\mid\X)$ and $\Y\mid(\X,\A=a)\sim P_{Y\mid X,a}$. The tangent space decomposition established in Appendix~\ref{app:proof-eif-structure} gives
\[
\cT(\e)=\cT_\X \ \oplus\ \bigoplus_{a\in\cA}\cT_a,
\]
and any pathwise-differentiable scalar parameter $\theta(P)$ admits a unique EIF of the form
\begin{equation}
\label{eq:eif-template-app-b}
\varphi_\e(\Oo)
= \bigl(g(\X)-\theta(P_0)\bigr)
+ \sum_{a\in\cA}\frac{\1\{\A=a\}}{\e(a\mid\X)}\,v_a(\X,\Y),
\qquad \E[v_a\mid\X,\A=a]=0,
\end{equation}
with efficiency bound $V(\e)=\E[\varphi_\e(\Oo)^2]=\Var(g(\X))+\sum_a\E\!\left[\sigma_a^2(\X)/\e(a\mid\X)\right]$ when $v_a(\X,\Y)=\Y-\E[\Y\mid\X,\A=a]$.

The four subsections that follow specialize \eqref{eq:eif-template-app-b} by identifying the scalar $g(\cdot)$ and the arm-specific residuals $v_a(\cdot,\cdot)$ for each example.

\subsection{Average Treatment Effect (ATE)}
\label{app:eif-ate}

\noindent\emph{Estimand.}
With $\mu_a(\X):=\E[\Y\mid\X,\A=a]$, the ATE is the linear functional
\[
\theta(P)=\E_P[\mu_1(\X)-\mu_0(\X)].
\]

\noindent\emph{Pathwise derivative.}
Fix a regular one-dimensional submodel $P_t$ with score $S(\Oo)=S_X(\X)+\sum_{a}\1\{\A=a\}S_a(\X,\Y)$. Differentiating $\theta(P_t)=\int[\mu_{1,t}(x)-\mu_{0,t}(x)]p_{X,t}(x)\,dx$ at $t=0$ and applying the product rule,
\begin{equation}
\label{eq:eif-ate-step1}
\dot\theta_{P_0}(S)
=\int[\mu_1(x)-\mu_0(x)]\dot p_X(x)\,dx
+\sum_{a}\int \dot\mu_a(x)\cdot(-1)^{1-a}\,p_X(x)\,dx.
\end{equation}
For the first term, $\dot p_X/p_X=S_X$ and $\E[S_X(\X)]=0$ allow us to subtract $\theta$:
\[
\int[\mu_1-\mu_0]\dot p_X\,dx
=\E\!\left[\bigl(\mu_1(\X)-\mu_0(\X)-\theta\bigr)S_X(\X)\right].
\]
For each arm $a$, $\dot\mu_a(x)=\E[(\Y-\mu_a(x))S_a(x,\Y)\mid\X=x,\A=a]$ since $\E[S_a\mid\X,\A=a]=0$. Converting back to an observed-data expectation by inserting the propensity-reweighting kernel $\1\{\A=a\}/\e(a\mid\X)$,
\begin{equation}
\label{eq:eif-ate-step2}
\int\dot\mu_a(x)p_X(x)\,dx
=\E\!\left[\frac{\1\{\A=a\}}{\e(a\mid\X)}\bigl(\Y-\mu_a(\X)\bigr)S_a(\X,\Y)\right].
\end{equation}

\noindent\emph{EIF.}
Substituting \eqref{eq:eif-ate-step2} into \eqref{eq:eif-ate-step1} and matching $\dot\theta_{P_0}(S)=\E[\varphi_\e(\Oo)S(\Oo)]$ termwise yields
\begin{equation}
\label{eq:eif-ate}
\boxed{\;
\varphi_\e(\Oo)
=\bigl(\mu_1(\X)-\mu_0(\X)-\theta\bigr)
+\frac{\1\{\A=1\}}{\e(1\mid\X)}\bigl(\Y-\mu_1(\X)\bigr)
-\frac{\1\{\A=0\}}{\e(0\mid\X)}\bigl(\Y-\mu_0(\X)\bigr).\;}
\end{equation}
This matches the template \eqref{eq:eif-template-app-b} with $g(\X)=\mu_1(\X)-\mu_0(\X)$, $v_1(\X,\Y)=\Y-\mu_1(\X)$, and $v_0(\X,\Y)=-(\Y-\mu_0(\X))$, both satisfying the orthogonality $\E[v_a\mid\X,\A=a]=0$.

\noindent\emph{Efficiency bound.}
Since the cross terms between the $\X$-piece and the IPW pieces vanish (by $\E[v_a\mid\X,\A=a]=0$) and the two arm-IPW pieces are mutually orthogonal ($\1\{\A=1\}\1\{\A=0\}\equiv 0$),
\begin{equation}
\label{eq:V-ate}
V_{\mathrm{ATE}}(\e)
=\Var\!\bigl(\mu_1(\X)-\mu_0(\X)\bigr)
+\E\!\left[\frac{\sigma_1^2(\X)}{\e(1\mid\X)}\right]
+\E\!\left[\frac{\sigma_0^2(\X)}{\e(0\mid\X)}\right],
\end{equation}
which is the standard semiparametric efficiency bound for the ATE under known propensities $\e(\cdot\mid\X)$.

\subsection{Quantile Treatment Effect (QTE)}
\label{app:eif-qte}

\noindent\emph{Estimand.}
Fix a quantile level $\tau\in(0,1)$. For $a\in\{0,1\}$, let $F_{\Y(a)}$ denote the marginal CDF of $\Y(a)$ and $q_a:=q_a(\tau)$ the $\tau$-quantile, defined implicitly by $F_{\Y(a)}(q_a)=\tau$. The QTE is
\[
\theta(P)=q_1(\tau)-q_0(\tau).
\]
Throughout this subsection we assume that $\Y(a)$ admits a continuous marginal density $f_{\Y(a)}$ with $f_{\Y(a)}(q_a)>0$, the standard regularity for pathwise differentiability of quantile functionals.

\noindent\emph{Pathwise derivative via implicit differentiation.}
Differentiating the identity $F_{\Y(a),t}(q_{a,t})=\tau$ at $t=0$ gives
\begin{equation}
\label{eq:eif-qte-implicit}
\dot q_a(S)=-\frac{\dot F_{\Y(a)}(q_a)}{f_{\Y(a)}(q_a)},
\end{equation}
so it suffices to compute $\dot F_{\Y(a)}(q_a)$. Under no unmeasured confounding, $F_{\Y(a)\mid\X}(y\mid x)=F_{\Y\mid\X,\A=a}(y\mid x)$, and
\[
F_{\Y(a)}(q_a)=\E\!\left[F_{\Y\mid\X,\A=a}(q_a\mid\X)\right].
\]
Differentiating along the submodel and applying the same product-rule decomposition as in B.2,
\begin{align}
\dot F_{\Y(a)}(q_a)
&=\E\!\left[\bigl(F_{\Y\mid\X,\A=a}(q_a\mid\X)-\tau\bigr)\,S_X(\X)\right] \nonumber\\
&\quad+\E\!\left[\frac{\1\{\A=a\}}{\e(a\mid\X)}\bigl(\1\{\Y\le q_a\}-F_{\Y\mid\X,\A=a}(q_a\mid\X)\bigr)S_a(\X,\Y)\right],
\label{eq:eif-qte-Fdot}
\end{align}
using $\E[F_{\Y\mid\X,\A=a}(q_a\mid\X)]=\tau$ to recenter the $\X$-piece.

\noindent\emph{EIF for the quantile $q_a$.}
Combining \eqref{eq:eif-qte-implicit} and \eqref{eq:eif-qte-Fdot} and reading off $\dot q_a(S)=\E[\varphi_a^{q}(\Oo)S(\Oo)]$,
\begin{equation}
\label{eq:eif-qa}
\varphi_a^{q}(\Oo)
=-\frac{F_{\Y\mid\X,\A=a}(q_a\mid\X)-\tau}{f_{\Y(a)}(q_a)}
-\frac{\1\{\A=a\}}{\e(a\mid\X)}\cdot\frac{\1\{\Y\le q_a\}-F_{\Y\mid\X,\A=a}(q_a\mid\X)}{f_{\Y(a)}(q_a)}.
\end{equation}
This matches the template with $g_a(\X)=q_a-\tfrac{F_{\Y\mid\X,\A=a}(q_a\mid\X)-\tau}{f_{\Y(a)}(q_a)}$ (so $\E[g_a(\X)]=q_a$) and arm-specific residual
\[
v_a^{q}(\X,\Y)
=-\frac{\1\{\Y\le q_a\}-F_{\Y\mid\X,\A=a}(q_a\mid\X)}{f_{\Y(a)}(q_a)},
\qquad
\E\!\left[v_a^{q}\mid\X,\A=a\right]=0.
\]

\noindent\emph{EIF for the QTE and efficiency bound.}
By linearity of pathwise differentiation, $\theta=q_1-q_0$ has EIF
\begin{equation}
\label{eq:eif-qte}
\boxed{\;
\varphi_\e^{\mathrm{QTE}}(\Oo)
=\varphi_1^{q}(\Oo)-\varphi_0^{q}(\Oo).\;}
\end{equation}
The IPW pieces in arms $0$ and $1$ are mutually orthogonal ($\1\{\A=1\}\1\{\A=0\}\equiv 0$), and the cross terms between any $\X$-piece and any IPW piece vanish by $\E[v_a^{q}\mid\X,\A=a]=0$. Writing $G_a(\X):=F_{\Y\mid\X,\A=a}(q_a\mid\X)$ for brevity,
\begin{equation}
\label{eq:V-qte}
V_{\mathrm{QTE}}(\e)
=\Var\!\left(\frac{G_1(\X)}{f_{\Y(1)}(q_1)}-\frac{G_0(\X)}{f_{\Y(0)}(q_0)}\right)
+\sum_{a\in\{0,1\}}\E\!\left[\frac{G_a(\X)\bigl(1-G_a(\X)\bigr)}{\e(a\mid\X)\,f_{\Y(a)}(q_a)^2}\right],
\end{equation}
where we used $\Var(\1\{\Y\le q_a\}\mid\X,\A=a)=G_a(\X)(1-G_a(\X))$.

\subsection{Policy Value}
\label{app:eif-policy}

\noindent\emph{Estimand.}
Fix a policy $\pi:\cX\to\Delta(\cA)$ that is known and does not depend on the unknown distribution $P$. The policy value is the linear functional
\[
\theta(P)=\E_P\!\left[\sum_{a\in\cA}\pi_a(\X)\mu_a(\X)\right].
\]
Because $\pi$ is fixed, $\theta$ is linear in the conditional-mean nuisance $\mu=(\mu_a)_{a\in\cA}$.

\noindent\emph{Pathwise derivative.}
Writing $\theta(P_t)=\int\sum_a\pi_a(x)\mu_{a,t}(x)\,p_{X,t}(x)\,dx$ and differentiating at $t=0$,
\[
\dot\theta_{P_0}(S)
=\E\!\left[\Bigl(\sum_a\pi_a(\X)\mu_a(\X)-\theta\Bigr)S_X(\X)\right]
+\sum_a\int\pi_a(x)\dot\mu_a(x)\,p_X(x)\,dx,
\]
where the recentering by $\theta$ uses $\E[S_X]=0$. Applying the same identity $\dot\mu_a(x)=\E[(\Y-\mu_a)S_a\mid\X=x,\A=a]$ as in B.2 and converting to an observed-data expectation,
\[
\sum_a\int\pi_a(x)\dot\mu_a(x)\,p_X(x)\,dx
=\sum_a\E\!\left[\frac{\pi_a(\X)\1\{\A=a\}}{\e(a\mid\X)}\bigl(\Y-\mu_a(\X)\bigr)S_a(\X,\Y)\right].
\]

\noindent\emph{EIF.}
Matching $\dot\theta_{P_0}(S)=\E[\varphi_\e(\Oo)S(\Oo)]$ termwise,
\begin{equation}
\label{eq:eif-policy}
\boxed{\;
\varphi_\e(\Oo)
=\Bigl(\sum_a\pi_a(\X)\mu_a(\X)-\theta\Bigr)
+\sum_a\frac{\pi_a(\X)\1\{\A=a\}}{\e(a\mid\X)}\bigl(\Y-\mu_a(\X)\bigr).\;}
\end{equation}
This is the template \eqref{eq:eif-template-app-b} with $g(\X)=\sum_a\pi_a(\X)\mu_a(\X)$ and $v_a(\X,\Y)=\pi_a(\X)(\Y-\mu_a(\X))$, which inherits the orthogonality $\E[v_a\mid\X,\A=a]=0$ from $\E[\Y-\mu_a\mid\X,\A=a]=0$.

\noindent\emph{Efficiency bound.}
The same orthogonality argument as in B.2 yields
\begin{equation}
\label{eq:V-policy}
V_{\mathrm{policy}}(\e)
=\Var\!\left(\sum_a\pi_a(\X)\mu_a(\X)\right)
+\sum_a\E\!\left[\frac{\pi_a(\X)^2\sigma_a^2(\X)}{\e(a\mid\X)}\right].
\end{equation}

\noindent\emph{Remark (linear-functional case).}
Because $\theta$ is linear in $\mu$, Proposition~\ref{prop:linear-unbiased} applies, so the cross-term in the excess-MSE decomposition vanishes \emph{exactly} under cross-fitting; this is the linear-functional sharpness regime analyzed in Section~\ref{sec:regression-adjustment}.

\section{Covariate-Balancing Moment Functions for the Canonical Examples}
\label{app:cb-examples}
\label{app:examples}

This appendix makes explicit, for each of the three canonical estimands considered in the main text, the three objects that drive the covariate-balancing analysis of Section~\ref{sec:covariate-balancing}: (i)~the moment function $m(\X,\A,\Y,\theta)$, (ii)~the induced balance function $\omega_a(\X)=-M^{-1}\,\E[m(\X,a,\Y,\theta_0)\mid\X]$, and (iii)~the second-moment proxy $m_a(\X)=\E[\,\|m(\X,a,\Y,\theta_0)\|^2\mid\X]$ that appears in the generalized Neyman allocation. Together these determine both the feasibility check of Theorem~\ref{thm:cb-efficiency} (Lipschitz balance functions) and the closed-form design update inside Algorithm~\ref{alg:batched-cb}.

\subsection{The moment-and-balance-function framework}
\label{app:cb-framework}

We work throughout with scalar or low-dimensional estimands $\theta\in\R^q$ identified by a moment condition
\[
\E[m(\X,\A,\Y,\theta_0)]=0, \qquad M:=\frac{\partial}{\partial\theta}\,\E[m(\X,\A,\Y,\theta)]\Big|_{\theta=\theta_0}\text{ invertible}.
\]
The balance function associated with arm $a$ is defined by
\begin{equation}
\label{eq:omega-def-appC}
\omega_a(\X) := -M^{-1}\,\E[m(\X,a,\Y,\theta_0)\mid\X],
\end{equation}
and the second-moment proxy is $m_a(\X):=\E[\,\|m(\X,a,\Y,\theta_0)\|_{M^{-\top} M^{-1}}^2\mid\X]$, which reduces to the conditional variance $\sigma_a^2(\X)$ in the scalar case. The generalized Neyman allocation solves
\[
\e_a^\star(x) \;\propto\; \sqrt{m_a(x)},
\qquad
\sum_a \e_a^\star(x)=1,
\]
up to the propensity floor. For each example below, we give the moment function, compute $\omega_a$ explicitly from \eqref{eq:omega-def-appC}, and record $m_a$.

\subsection{Average Treatment Effect (ATE)}
\label{app:cb-ate}

Take $m(\X,\A,\Y,\theta)=\Y\cdot\dfrac{\1\{\A=1\}-\1\{\A=0\}}{\e_{\A}(\X)}-\theta$. Then $M=-1$ and
\[
\omega_a(\X)=(-1)^{1-a}\,\mu_a(\X), \qquad a\in\{0,1\},
\]
which is Lipschitz whenever $\mu_a$ is. The second-moment proxy is $m_a(\X)=\sigma_a^2(\X)$, giving the familiar Neyman allocation $\e_a^\star(x)\propto\sigma_a(x)$.

\subsection{Quantile Treatment Effect (QTE)}
\label{app:cb-qte}

For $\theta=q_1(\tau)-q_0(\tau)$, the moment equation uses $m_a(\Y,\theta_a)=\1\{\Y\le\theta_a\}-\tau$ on each arm, and the balance function is
\[
\omega_a(\X)=\frac{F_{\Y(a)\mid\X}(q_a(\tau))-\tau}{f_a(q_a(\tau))},
\]
which is Lipschitz under smoothness of the conditional CDF. Writing $G_a(\X):=F_{\Y(a)\mid\X}(q_a(\tau)\mid\X)$ for the conditional CDF evaluated at the marginal $\tau$-quantile, the second-moment proxy is
\[
m_a(\X)=\frac{G_a(\X)\bigl(1-G_a(\X)\bigr)}{f_a(q_a(\tau))^2},
\]
which matches the conditional variance of the indicator residual $v_a^q(\X,\Y)$ in the EIF derivation \eqref{eq:eif-qa} (Appendix~\ref{app:eif-qte}).

\subsection{Policy Value}
\label{app:cb-policy}

For a fixed policy $\pi:\cX\to\Delta(\cA)$, $\theta=\E[\sum_a\pi_a(\X)\mu_a(\X)]$. The nuisance is $\eta=(\mu_a)_{a\in\cA}$ with pseudo-metric $d(\eta,\eta_0)^2=\sum_a \E[(\hat\mu_a(\X)-\mu_a(\X))^2]$, and
\[
\omega_a(\X)=\pi_a(\X)\,\mu_a(\X),
\qquad
m_a(\X)=\pi_a(\X)^2\,\sigma_a^2(\X),
\]
which is Lipschitz when both $\pi_a$ and $\mu_a$ are. Exact global unbiasedness (Proposition~\ref{prop:linear-unbiased}) applies.

\medskip
\noindent In every case, Theorem~\ref{thm:cb-efficiency} applies with the corresponding $(m,\omega_a,m_a)$ triple, so the batched adaptive balancing design achieves the $V^\star/n$ efficiency benchmark.

\section{Regularity Conditions Used in Sections~4--5}
\label{app:regularity}

This appendix collects the formal regularity conditions invoked in the regression-adjustment analysis of Section~\ref{sec:regression-adjustment} and the covariate-balancing analysis of Section~\ref{sec:covariate-balancing}. The body text records only the structural orthogonality and rate hypotheses; the auxiliary boundedness, positivity, and second-order remainder conditions are gathered here.

\subsection{General EIF regularity (Section~4.3)}
\label{app:reg-general}

\begin{nassumption}[General EIF regularity]
\label{assump:general-reg}
We assume the following regularity conditions:
\begin{enumerate}[label=(\roman*), leftmargin=*]
\item \textbf{Positivity and boundedness:}
The assignment propensities satisfy $\e^{(b)}_a(x)\ge \varepsilon>0$, and the conditional second moments $m_a(x)$ are uniformly bounded almost surely.

\item \textbf{Neyman orthogonality:}
The pathwise derivative of the expected EIF score vanishes at the true nuisance parameter $\eta_0$.

\item \textbf{Second-order remainder:}
The bias and variance of the EIF score vary at most quadratically with nuisance error:
\begin{align}
\bigl|\E[\psi(\eta)-\psi(\eta_0)]\bigr| &\lesssim d(\eta,\eta_0)^2, \label{eq:remainder-mean}\\
\E\!\left[(\psi(\eta)-\psi(\eta_0))^2\right] &\lesssim d(\eta,\eta_0)^2. \label{eq:l2-lip}
\end{align}
\end{enumerate}
\end{nassumption}

\section{Proofs for Sections 2--3 (EIF Structure and Lower Bounds)}
\label{app:proofs-sec2}

\subsection{Proof of the EIF structure under known assignment}
\label{app:proof-eif-structure}

\emph{Strategy.} The argument has three steps. First, we exploit the known assignment mechanism to decompose the tangent space at $P_0$ into an orthogonal sum of a covariate component and arm-specific components. Second, a Riesz representation argument on each orthogonal component isolates a unique representer. Third, we show that the arm-specific representer must carry the characteristic $1/\e(a\mid\X)$ inverse-weighting factor, which combined with the uniqueness of the covariate piece yields the asserted EIF form \eqref{eq:eif-structure}.

\noindent\textbf{Proof of Lemma~\ref{lem:eif-structure}.}
Let $L_2^0(P_0)$ be the Hilbert space of square-integrable, mean-zero functions under $P_0$ with inner
product $\langle f,g\rangle := \E_0[f(O)g(O)]$.
Because $\e$ is fixed, any regular one-dimensional submodel $t\mapsto P_t$ through $P_0$ can perturb only
$p_X$ and each $p_a(\cdot\mid x)$ (but not $\e$). Therefore the score
$S(O):=\left.\frac{d}{dt}\log p_t(O)\right|_{t=0}$ must admit the decomposition
\begin{equation}
\label{eq:score-decomp}
S(O)=S_X(X)+\sum_{a\in\cA}\1\{A=a\}\,S_a(X,Y),
\qquad
\E_0[S_X(X)]=0,\quad
\E_0[S_a(X,Y)\mid X,A=a]=0.
\end{equation}
Define the linear subspaces
\[
\cT_X := \{h(X): \E_0[h(X)]=0\},\qquad
\cT_a := \{\1\{A=a\}s_a(X,Y): \E_0[s_a(X,Y)\mid X,A=a]=0\}.
\]
Then the (closed) tangent space at $P_0$ is the orthogonal direct sum
\[
\cT(\e)=\cT_X\ \oplus\ \bigoplus_{a\in\cA}\cT_a.
\]
Orthogonality holds because for any $h\in\cT_X$ and $\1\{A=a\}s_a\in\cT_a$,
\[
\E_0\!\big[h(X)\1\{A=a\}s_a(X,Y)\big]
=\E_0\!\Big[h(X)\E_0[\1\{A=a\}s_a(X,Y)\mid X]\Big]
=\E_0\!\Big[h(X)\e(a\mid X)\underbrace{\E_0[s_a(X,Y)\mid X,A=a]}_{0}\Big]=0,
\]
and for $a\neq b$, $\1\{A=a\}\1\{A=b\}\equiv 0$ implies $\cT_a\perp \cT_b$.

\noindent\textbf{Step 2: Riesz representation on each orthogonal component.}
Pathwise differentiability of $\theta$ implies that the derivative map
$S\mapsto \dot\theta_{P_0}(S)\in\R$ is a continuous linear functional on the Hilbert space $\cT(\e)$.
Hence, by the Riesz representation theorem, there exists a unique
$\varphi_{\e}\in\overline{\cT(\e)}\subset L_2^0(P_0)$ such that
\begin{equation}
\label{eq:riesz}
\dot\theta_{P_0}(S)=\E_0\!\big[\varphi_{\e}(O)\,S(O)\big],\qquad \forall S\in\cT(\e).
\end{equation}
Because $\cT(\e)$ is an orthogonal direct sum, any $\varphi_{\e}\in\overline{\cT(\e)}$ decomposes uniquely as
\begin{equation}
\label{eq:phi-orth-decomp}
\varphi_{\e}(O)=\varphi_X(X)+\sum_{a\in\cA}\varphi_a(O),
\qquad
\varphi_X\in \cT_X,\ \varphi_a\in \cT_a.
\end{equation}
The $\cT_X$-component has the form $\varphi_X(X)=g(X)-\theta(P_0)$ for some measurable $g$ with
$\E_0[g(X)]=\theta(P_0)$.

\noindent\textbf{Step 3: Identification of the arm-specific representer and the $1/\e$ factor.}
Fix an arm $a\in\cA$ and consider scores of the form $S(O)=\1\{A=a\}s_a(X,Y)\in\cT_a$.
By \eqref{eq:riesz} and \eqref{eq:phi-orth-decomp}, the restriction of $\dot\theta_{P_0}$ to $\cT_a$
is represented by $\varphi_a$ via
\begin{equation}
\label{eq:arm-repr}
\dot\theta_{P_0}(\1\{A=a\}s_a)=\E_0\!\big[\varphi_a(O)\,\1\{A=a\}s_a(X,Y)\big]
\qquad \forall s_a:\ \E_0[s_a\mid X,A=a]=0.
\end{equation}
We now show that any such representer must be of the form
$\varphi_a(O)=\frac{\1\{A=a\}}{\e(a\mid X)}v_a(X,Y)$ with $\E_0[v_a\mid X,A=a]=0$.

Indeed, for any admissible $s_a$,
\begin{align}
\E_0\!\left[
\left(\frac{\1\{A=a\}}{\e(a\mid X)}v_a(X,Y)\right)
\left(\1\{A=a\}s_a(X,Y)\right)
\right]
&=\E_0\!\left[\frac{\1\{A=a\}}{\e(a\mid X)}\,v_a(X,Y) s_a(X,Y)\right] \nonumber\\
&=\E_0\!\left[
\E_0\!\left[\frac{\1\{A=a\}}{\e(a\mid X)}\,v_a(X,Y) s_a(X,Y)\,\Big|\,X\right]
\right] \nonumber\\
&=\E_0\!\left[
\frac{1}{\e(a\mid X)}\,\PP(A=a\mid X)\,
\E_0\!\left[v_a(X,Y) s_a(X,Y)\mid X,A=a\right]
\right] \nonumber\\
&=\E_0\!\Big[\E_0\big[v_a(X,Y) s_a(X,Y)\mid X,A=a\big]\Big].
\label{eq:key-identity}
\end{align}

Finally, uniqueness of $v_a$ follows from \eqref{eq:key-identity}: if
$\frac{\1\{A=a\}}{\e(a\mid X)}v_a$ and $\frac{\1\{A=a\}}{\e(a\mid X)}\tilde v_a$ induce the same functional on
$\cT_a$, then for all admissible $s_a$,
\[
0=\E_0\!\Big[\E_0\big[(v_a-\tilde v_a) s_a\mid X,A=a\big]\Big].
\]
Choosing $s_a=v_a-\tilde v_a$ yields
\[
0=\E_0\!\Big[\E_0\big[(v_a-\tilde v_a)^2\mid X,A=a\big]\Big],
\]
hence $v_a=\tilde v_a$ a.s.\ on $\{A=a\}$. Combining the unique $\cT_X$ and arm-specific components yields
\eqref{eq:eif-structure}. \QEDA

\subsection{Proofs for the average-propensity lower bound}
\label{app:proof-lb-avg-prop}

\emph{Strategy.} We construct a least-favorable one-dimensional submodel whose score coincides (up to a normalization constant) with the scalar EIF under the average-propensity model. Because the design rule is not perturbed along this submodel, the trajectory log-likelihood ratio decomposes into a martingale sum of observation-level scores. An average-propensity identity then reduces the variance of this sum to the observation-level EIF variance, so Cauchy--Schwarz produces the desired lower bound.

\begin{lemma}[Tilting-score identity]
\label{lem:tilting-score}
Let $\phi=\IF_{\ebar}$ be the scalar EIF under the average-propensity model with variance $V=\E[\phi(\Oo)^2]$. There exists a regular one-dimensional submodel $\{P_\varepsilon\}_{\varepsilon}$ through $P_0$, consisting of exponential tiltings of the covariate marginal and each outcome conditional, whose score at $\varepsilon=0$ equals $s(\Oo)=\phi(\Oo)/V$.
\end{lemma}

\noindent\textbf{Proof of Lemma~\ref{lem:tilting-score}.}
Write the scalar EIF in the decomposition induced by \eqref{eq:eif-structure}:
\[
\phi(\Oo)=\bigl(g(\X)-\theta(P_0)\bigr)+\sum_{a=1}^K\frac{\1\{\A=a\}}{\ebar_a(\X)}\,v_{a}(\X,\Y),\qquad
\E[v_{a}(\X,\Y)\mid \X,\A=a]=0,
\]
for some $g$ with $\E[g(\X)]=\theta(P_0)$.
Define a one-dimensional submodel $P_{\varepsilon}$ by exponential tilting of the marginal and conditionals:
\begin{align}
\frac{dP_{\X,\varepsilon}}{dP_{\X,0}}(x)
&\propto \exp\!\left\{\varepsilon \frac{g(x)-\theta(P_0)}{V}\right\}, \\
\frac{dP_{\Y\mid \X,\A=a,\varepsilon}}{dP_{\Y\mid \X,\A=a,0}}(y\mid x)
&\propto \exp\!\left\{\varepsilon \frac{v_{a}(x,y)}{V\,\ebar_a(x)}\right\},
\qquad a\in\cA.
\end{align}
Under dominated/QMD regularity, the score at $\varepsilon=0$ is exactly
$s(\Oo)=\phi(\Oo)/V$.

\noindent\textbf{Proof of Theorem~\ref{thm:avg-prop-lb}.}
Let $P_{\varepsilon}^{\ALG}$ denote the trajectory law of $(\Oo_1,\dots,\Oo_n)$ under
$P_{\varepsilon}$ and design $\ALG$. Define
\[
Z(\cF_n):=\left.\frac{d}{d\varepsilon}\log\frac{dP_{\varepsilon}^{\ALG}}{dP_{0}^{\ALG}}(\cF_n)\right|_{\varepsilon=0}.
\]
Because the design rule is not perturbed by $\varepsilon$,
\[
Z(\cF_n)=\sum_{t=1}^n s(\Oo_t).
\]
Also $\{s(\Oo_t)\}_{t=1}^n$ is a martingale difference sequence under $\{\cF_t\}$:
\[
\E[s(\Oo_t)\mid \cF_{t-1}]=0.
\]
For any integrable test function $v_a(\X,\Y)$, we use the average-propensity identity
\begin{equation}
\frac{1}{n}\sum_{t=1}^n \E\!\left[\frac{\1\{\A_t=a\}}{\ebar_a(\X_t)^2}\,v_a(\X_t,\Y_t)\right]
=
\E\!\left[\frac{1}{\ebar_a(\X)}\,\E\!\left[v_a(\X,\Y)\mid \X,\A=a\right]\right],
\label{eq:avg-identity}
\end{equation}
which follows from conditioning on $(\cF_{t-1},\X_t)$, then averaging in $t$ and applying \eqref{eq:avg-prop}.

Let $T=T(\cF_n)$ be locally unbiased along this submodel:
\[
\left.\frac{d}{d\varepsilon}\E_{\varepsilon}^{\ALG}[T]\right|_{\varepsilon=0}
=
\left.\frac{d}{d\varepsilon}\theta(P_{\varepsilon})\right|_{\varepsilon=0}.
\]
By score identity and EIF normalization,
\[
\E\!\left[ Z(\cF_n)\,T\right]=1.
\]
Hence by Cauchy--Schwarz,
\[
\Var(T)\ge \frac{1}{\E[Z(\cF_n)^2]}.
\]
Using the martingale-difference decomposition,
\[
\E[Z(\cF_n)^2]
=\sum_{t=1}^n \E[s(\Oo_t)^2]
=n\,\E[s(\Oo)^2]
=\frac{n}{\E[\IF_{\ebar}(\Oo)^2]},
\]
where the last step uses $s=\phi/V$ and identity \eqref{eq:avg-identity}. Therefore
\[
\Var(T)\ \ge\ \frac{1}{n}\E\!\left[\IF_{\ebar}(\Oo)^2\right].
\]
In particular, \eqref{eq:avg-identity} is used with test functions built from the EIF components
(e.g., $v_a(\X,\Y)^2$) when evaluating $\E[\IF_{\ebar}(\Oo)^2]$ under the structural form
\eqref{eq:eif-structure}.

\subsection{Resolving the two-stage local randomization conjecture of Cytrynbaum (2021)}
\label{app:two-stage-aps}

\emph{Strategy.} \citet{cytrynbaum2021optimal} conjectured that in two-stage locally randomized experimental designs, the semiparametric efficiency benchmark remains characterized by a fixed i.i.d.\ benchmark evaluated at the \emph{induced average effective propensity scores}, with an explicit additive decomposition in the stagewise-factorization case. The main-text result of Section~\ref{sec:lower-bound} resolves a scalar version of this conjecture via the average-propensity argument in Appendix~\ref{app:proof-lb-avg-prop}. Here we extend that argument to the two-stage setting in full generality, proving (i)~the base-model analogue of Theorem~\ref{thm:avg-prop-lb} in which both stages of the experiment use a common design summary $\psi$, and (ii)~the factorization extension in which sampling and assignment use different stage-wise summaries $(\psi_1,\psi_2)$ and the benchmark factorizes into independent stage-1 sampling and stage-2 assignment propensities. The proof strategy mirrors Appendix~\ref{app:proof-lb-avg-prop}: compute the EIF in the corresponding fixed benchmark model, then bound any locally unbiased estimator's variance from below by the benchmark's inverse Fisher information evaluated at the induced average effective propensity scores.

Throughout, the target parameter is the population average treatment effect
\[
\theta(P_0):=\E_{P_0}[\Y(1)-\Y(0)].
\]

\subsubsection{Common setup and notation}

Let
\[
W_t=(\X_t,\Y_t(0),\Y_t(1)),\qquad t=1,\dots,n,
\]
be latent full data with $W_1,\dots,W_n\stackrel{i.i.d.}{\sim} P_0$. Write
\[
\mu_a(x):=\E[\Y(a)\mid\X=x],\qquad
\sigma_a^2(x):=\Var(\Y(a)\mid\X=x),\qquad a\in\{0,1\},
\]
and $c(x):=\mu_1(x)-\mu_0(x)$.

For each unit, the final observed state is $A_t\in\{\varnothing,0,1\}$, where $A_t=\varnothing$ if unit $t$ is not sampled, and $A_t\in\{0,1\}$ indicates sampling with the corresponding assignment. The observed outcome is
\[
Y_t^{\obs}:=\mathbf 1\{A_t\neq\varnothing\}\bigl(\mathbf 1\{A_t=1\}\Y_t(1)+\mathbf 1\{A_t=0\}\Y_t(0)\bigr).
\]

\paragraph{Outcome visibility.}
For the lower-bound argument, the only substantive design regularity needed is that, conditional on $(\X_t,A_t)$, the law of the observed outcome coincides with the corresponding potential-outcome law and does not depend on the design history:
\[
\mathcal L(Y_t^{\obs}\mid \X_t,A_t=a,\cF_{t-1})
=
\mathcal L(\Y_t(a)\mid \X_t),\qquad a\in\{0,1\}.
\]

\subsubsection{A serializable filtration for local randomization}

The local-randomization procedures in the survey-experiment literature are defined offline: one first forms groups using visible design variables and exogenous randomness, and then draws exactly $a$ out of $k$ units (or $a'$ out of $k'$ units) inside each group. For the efficiency argument, it is convenient to work with a \emph{serialized representation} of this offline procedure.

\begin{assumption}[Serializable non-anticipating representation]
\label{ass:serializable}
There exists an ordering of units $t=1,\dots,n$ and a filtration $(\cF_t)_{t=0}^n$ such that:
\begin{enumerate}
\item for each $t$, $\cF_{t-1}$ contains all design-visible information revealed strictly before the final state of unit $t$ is resolved;
\item conditional on $(\X_t,\cF_{t-1})$, the final observed state of unit $t$ has a measurable kernel
\[
\rho_t(a\mid \X_t,\cF_{t-1}):=\PP(A_t=a\mid \X_t,\cF_{t-1}),\qquad a\in\{\varnothing,0,1\};
\]
\item the kernel $\rho_t$ depends only on the stagewise visible covariates and the past design history, never on unrevealed outcomes or on the potential outcomes of unit $t$;
\item conditional on $(\X_t,A_t)$, the observed outcome law satisfies the outcome-visibility condition above.
\end{enumerate}
\end{assumption}

\paragraph{Why this is natural for local randomization.}
After the offline groups are formed from the visible design variables and exogenous grouping randomness, one may fix an arbitrary order of groups and an arbitrary order of units within each group. Conditional on the realized grouping, exact $a$-out-of-$k$ randomization inside a group has the standard without-replacement sequential representation
\[
\PP(U_{g,r}=1\mid \mathcal H_{g,r-1})
=
\frac{a-N_{g,r-1}}{|g|-r+1},
\]
where $N_{g,r-1}$ is the number of previously selected units in group $g$. The same applies to the assignment stage. Any such serialization yields Assumption~\ref{ass:serializable}. This is the information order used below.

\subsubsection{Average effective propensity scores}

For $a\in\{0,1\}$ define the \emph{average effective observed-arm masses}
\begin{equation}
\bar\pi_a(x)
:=
\frac1n \sum_{t=1}^n \E\bigl[\rho_t(a\mid x,\cF_{t-1})\bigr].
\label{eq:avg-effective}
\end{equation}
These are the two-stage analogues of the induced average propensity scores of Appendix~\ref{app:proof-lb-avg-prop} and are the primitive design objects for the lower bound. When convenient, define also
\begin{equation}
\bar q(x):=\bar\pi_1(x)+\bar\pi_0(x),\qquad
\bar p(x):=\frac{\bar\pi_1(x)}{\bar\pi_1(x)+\bar\pi_0(x)}
\quad\text{on }\{\bar q(x)>0\}.
\label{eq:q-p-from-pi}
\end{equation}
In the extension, the factorization case is precisely the case in which $\bar q$ depends only on the stage-1 information and $\bar p$ depends only on the stage-2 information.

\begin{lemma}[Two-stage APS identities]
\label{lem:aps-identity}
Suppose Assumption~\ref{ass:serializable} holds and let
\[
m_a(x):=\E[h_a(x,\Y(a))\mid\X=x],\qquad a\in\{0,1\},
\]
for any integrable measurable functions $h_1,h_0$. Then
\begin{equation}
\frac1n\sum_{t=1}^n
\E\!\left[
\frac{\mathbf 1\{A_t=1\}}{\bar\pi_1(\X_t)^2}
h_1(\X_t,Y_t^{\obs})
\right]
=
\E\!\left[
\frac{m_1(\X)}{\bar\pi_1(\X)}
\right],
\label{eq:aps-id-1}
\end{equation}
and the analogous identity holds for $a=0$. Moreover, with $\bar q=\bar\pi_1+\bar\pi_0$,
\begin{equation}
\frac1n\sum_{t=1}^n
\E\!\left[
\frac{\mathbf 1\{A_t\neq\varnothing\}}{\bar q(\X_t)^2}
g(\X_t)
\right]
=
\E\!\left[
\frac{g(\X)}{\bar q(\X)}
\right]
\label{eq:aps-id-q}
\end{equation}
for every integrable measurable $g$.
\end{lemma}

\noindent\textbf{Proof of Lemma~\ref{lem:aps-identity}.}
We prove \eqref{eq:aps-id-1}; the $a=0$ and $\bar q$ identities are identical. By iterated expectations and the definition of $\rho_t$,
\[
\E\!\left[
\mathbf 1\{A_t=1\}h_1(\X_t,Y_t^{\obs})
\,\middle|\, \X_t,\cF_{t-1}
\right]
=
\rho_t(1\mid \X_t,\cF_{t-1})\,m_1(\X_t),
\]
because the conditional law of the observed outcome on $\{A_t=1\}$ is exactly the law of $\Y_t(1)$ given $\X_t$. Therefore
\[
\E\!\left[
\frac{\mathbf 1\{A_t=1\}}{\bar\pi_1(\X_t)^2}
h_1(\X_t,Y_t^{\obs})
\right]
=
\E\!\left[
\frac{\rho_t(1\mid \X_t,\cF_{t-1})}{\bar\pi_1(\X_t)^2}
m_1(\X_t)
\right].
\]
Averaging over $t$ and using the definition \eqref{eq:avg-effective} of $\bar\pi_1$ gives \eqref{eq:aps-id-1}. \QEDA

\subsubsection{Setting I: the base model}

In the base model there is a single design summary $L_t:=\psi(\X_t)$ visible for all eligible units from the outset. Sampling uses $L_t$, and assignment also uses $L_t$.

\paragraph{Concrete experiment and filtration.}
A locally randomized design with fixed target propensities $q(\psi)$ and $p(\psi)$ proceeds in two stages: (i) partition eligible units into $q$-strata and groups using $\{\psi(\X_i)\}$ and exogenous randomness, and draw exactly the required number of sampled units; (ii) on the sampled units, form $p$-strata and assignment groups using $\{\psi(\X_i)\}$ and a second exogenous seed, and draw exactly the required number of treated units. Fix any serialization of this offline grouped randomization as in Assumption~\ref{ass:serializable}; let $(\cF_t)$ and $\rho_t$ be the resulting filtration and final-state kernel.

\paragraph{APS in the base model.}
The primitive APS objects are $(\bar\pi_1,\bar\pi_0)$ from \eqref{eq:avg-effective}. If the average benchmark factorizes through $q$ and $p$, then $\bar\pi_1(x)=\bar q(\psi(x))\bar p(\psi(x))$ and $\bar\pi_0(x)=\bar q(\psi(x))\{1-\bar p(\psi(x))\}$.

\paragraph{Fixed benchmark model.}
The benchmark observation is $O=(\X,A,Y^A)$ with $A\in\{\varnothing,0,1\}$ generated by
\[
\PP(A=1\mid\X=x)=\bar\pi_1(x),\quad
\PP(A=0\mid\X=x)=\bar\pi_0(x),\quad
\PP(A=\varnothing\mid\X=x)=1-\bar\pi_1(x)-\bar\pi_0(x),
\]
and, on $\{A=a\}$, $Y^A$ has the law of $\Y(a)$ given $\X$.

\begin{proposition}[EIF in the base benchmark]
\label{prop:base-eif}
In the base benchmark model, the efficient influence function for the ATE is
\begin{equation}
\phi_{\mathrm{base},\bar\pi}(O)
=
c(\X)-\theta
+
\frac{\mathbf 1\{A=1\}}{\bar\pi_1(\X)}\{Y^A-\mu_1(\X)\}
-
\frac{\mathbf 1\{A=0\}}{\bar\pi_0(\X)}\{Y^A-\mu_0(\X)\}.
\label{eq:base-eif}
\end{equation}
Its variance is
\begin{equation}
V_{\mathrm{base}}(\bar\pi)
=
\Var(c(\X))
+
\E\!\left[
\frac{\sigma_1^2(\X)}{\bar\pi_1(\X)}
+
\frac{\sigma_0^2(\X)}{\bar\pi_0(\X)}
\right].
\label{eq:base-V}
\end{equation}
\end{proposition}

\noindent\textbf{Proof of Proposition~\ref{prop:base-eif}.}
Under the fixed benchmark, the observed-data likelihood factors as
$P_{\X}\otimes P_{A\mid\X}\otimes P_{Y^{\obs}\mid \X,A}$
with the assignment kernel $P_{A\mid\X}$ known. The observed-data tangent space
therefore decomposes orthogonally into three components: the covariate score
$\{s_{\X}(\X):\E[s_{\X}(\X)]=0\}$ and the two arm-conditional outcome scores
$\{\mathbf 1\{A=a\}\,v_a(\X,Y^{\obs}):\E[v_a(\X,\Y(a))\mid\X]=0\}$ for $a\in\{0,1\}$.
Pathwise differentiation of $\theta(P)=\E[\mu_1(\X)-\mu_0(\X)]$ along a regular
observed-data submodel yields a derivative representation whose components in
these three directions are, respectively, $c(\X)-\theta$,
$\{Y^{\obs}-\mu_1(\X)\}/\bar\pi_1(\X)$, and $-\{Y^{\obs}-\mu_0(\X)\}/\bar\pi_0(\X)$.
Assembling these three orthogonal components with the arm indicators as
embedding weights gives \eqref{eq:base-eif}, the EIF. The residual components
have conditional mean zero given $\X$ and are supported on disjoint arm events;
squaring \eqref{eq:base-eif} and taking expectations gives \eqref{eq:base-V}. \QEDA

\begin{theorem}[APS lower bound in the base model]
\label{thm:base-lb}
Assume Assumption~\ref{ass:serializable} and positivity $\bar\pi_a(x)\ge\varepsilon>0$ for $a\in\{0,1\}$. Let $\hat\theta_n$ be any estimator that is locally unbiased along the least-favorable one-dimensional submodel associated with the fixed benchmark under $\bar\pi$. Then
\begin{equation}
\Var_{P_0,\ALG}(\hat\theta_n)
\ \ge\
\frac1n V_{\mathrm{base}}(\bar\pi).
\label{eq:base-lb}
\end{equation}
\end{theorem}

\noindent\textbf{Proof of Theorem~\ref{thm:base-lb}.}
Let $V:=V_{\mathrm{base}}(\bar\pi)$. Consider a one-dimensional observed-data
submodel indexed by $\varepsilon$ that perturbs only the covariate law $P_{\X}$
and the two arm-conditional outcome laws $P_{\Y(a)\mid\X}$ (the assignment
kernel $P_{A\mid\X}$ is held fixed), with tangent components
\begin{equation}
s_{\X}(x)=\frac{c(x)-\theta}{V},\quad
s_1(x,y)=\frac{y-\mu_1(x)}{V\bar\pi_1(x)},\quad
s_0(x,y)=-\frac{y-\mu_0(x)}{V\bar\pi_0(x)}.
\label{eq:base-scores}
\end{equation}
Because $s_a(\X,\Y(a))$ has conditional mean zero given $\X$ under $P_0$, this
defines a valid observed-data submodel that references only marginal conditionals
$P_{\Y(a)\mid\X}$; no assumption is made on the joint law of $(\Y(0),\Y(1))\mid\X$.
The derivative of $\theta$ along this submodel equals
$V^{-1}\{\Var(c(\X))+\E[\sigma_1^2(\X)/\bar\pi_1(\X)+\sigma_0^2(\X)/\bar\pi_0(\X)]\}=1$
by \eqref{eq:base-V}.

Because the assignment does not depend on the latent perturbation and the
trajectory likelihood factorizes across units, the trajectory score is
\begin{equation}
Z_n=\sum_{t=1}^n\!\left[
    s_{\X}(\X_t)+\mathbf 1\{A_t=1\}s_1(\X_t,Y_t^{\obs})+\mathbf 1\{A_t=0\}s_0(\X_t,Y_t^{\obs})
\right],
\label{eq:base-Zn}
\end{equation}
each summand depending only on observed-data quantities. Cross terms between
distinct units vanish by latent independence and non-anticipation of the design;
within a unit, cross terms vanish because the covariate score is centered and
each residual score has conditional mean zero given $\X$. Applying
Lemma~\ref{lem:aps-identity} to the residual scores with $h_a(x,y)=(y-\mu_a(x))^2$
gives $\E[Z_n^2]=n/V$. By local unbiasedness along this submodel,
$1=\frac{d}{d\varepsilon}\E_{\varepsilon,\ALG}[\hat\theta_n]|_{\varepsilon=0}
=\E_{P_0,\ALG}[(\hat\theta_n-\theta)Z_n]$, so the Cauchy--Schwarz inequality
yields $\Var_{P_0,\ALG}(\hat\theta_n)\ge 1/\E[Z_n^2]=V/n$, which is \eqref{eq:base-lb}. \QEDA

\begin{remark}[Normalization by the realized experiment size]
If one normalizes by $n_T=\sum_t\1\{A_t\neq\varnothing\}$ rather than $n$, then $n/n_T\to 1/\E[\bar q(\psi(\X))]$ in the base model, so the same bound can be rewritten in the $\sqrt{n_T}$ normalization used in survey-experiment settings by multiplying the $n$-normalized variance by $\E[\bar q(\psi(\X))]$.
\end{remark}

\subsubsection{Setting II: the extension under stagewise factorization}

We now allow different design summaries at the two stages: $L_{1t}:=\psi_1(\X_t)$ and $L_{2t}:=\psi_2(\X_t)$. The key feature is that $L_{1t}$ is visible at the sampling stage for every eligible unit, whereas $L_{2t}$ is only revealed if unit $t$ is sampled.

\paragraph{Concrete experiment and filtration.}
Stage~1 uses $\{\psi_1(\X_t)\}$ and exogenous randomness to form sampling groups and execute exact local randomization. When unit $t$ is sampled, its $L_{2t}$ is revealed. Stage~2 uses the revealed $\{L_{2t}\}$ among sampled units and a second exogenous seed to form assignment groups and execute exact local randomization. Fix any serialization of this procedure and let $(\cF_t)$ be the resulting filtration.

\paragraph{Stagewise benchmark factorization.}
For the extension we restrict attention to the factorization case, in which the induced benchmark factors through a stage-1 benchmark sampling propensity $\bar q$ and a stage-2 benchmark assignment propensity $\bar p$:
\begin{equation}
\bar\pi_1(x)=\bar q(\psi_1(x))\,\bar p(\psi_2(x)),
\qquad
\bar\pi_0(x)=\bar q(\psi_1(x))\{1-\bar p(\psi_2(x))\}.
\label{eq:factorization}
\end{equation}
This is exactly the case in which the averaged benchmark can still be interpreted as ``first sample according to $\bar q(\psi_1)$, then assign according to $\bar p(\psi_2)$.''

\paragraph{Fixed benchmark model.}
The fixed benchmark observation is $O=(L_1, T, TL_2, TD, TY)$ where
\[
T\mid L_1\sim\mathrm{Bernoulli}(\bar q(L_1)),\quad
D\mid(T=1,L_2)\sim\mathrm{Bernoulli}(\bar p(L_2)),\quad
Y=T\{D\Y(1)+(1-D)\Y(0)\}.
\]
Define
\[
\nu_a(\ell_2):=\E[\Y(a)\mid L_2=\ell_2],\qquad
c_2(\ell_2):=\nu_1(\ell_2)-\nu_0(\ell_2),\qquad
c_1(\ell_1):=\E[c(\X)\mid L_1=\ell_1],
\]
and
\[
b_{\bar p}(\X):=\mu_1(\X)\sqrt{\frac{1-\bar p(L_2)}{\bar p(L_2)}}+\mu_0(\X)\sqrt{\frac{\bar p(L_2)}{1-\bar p(L_2)}}.
\]

\begin{proposition}[EIF in the extension benchmark]
\label{prop:ext-eif}
In the factorized extension benchmark, the efficient influence function for the ATE is
\begin{equation}
\phi_{\mathrm{ext},\bar q,\bar p}(O)
=
c_1(L_1)-\theta
+
\frac{T}{\bar q(L_1)}\!\left[
c_2(L_2)-c_1(L_1)
+
\frac{D}{\bar p(L_2)}\{Y-\nu_1(L_2)\}
-
\frac{1-D}{1-\bar p(L_2)}\{Y-\nu_0(L_2)\}
\right].
\label{eq:ext-eif}
\end{equation}
\end{proposition}

\noindent\textbf{Proof of Proposition~\ref{prop:ext-eif}.}
Under the factorized extension benchmark, the observed-data likelihood factors as
\[
P_{L_1}\otimes P_{T\mid L_1}\otimes P_{L_2\mid L_1,T=1}\otimes P_{D\mid L_2,T=1}\otimes P_{Y\mid L_2,T=1,D}
\]
with the sampling kernel $P_{T\mid L_1}$ and the assignment kernel
$P_{D\mid L_2,T=1}$ both known. The observed-data tangent space thus decomposes
orthogonally into four components: the marginal law of $L_1$, the conditional
law of $L_2$ given $L_1$ (revealed only on $\{T=1\}$), and the two arm-specific
outcome laws $\Y(a)\mid L_2$ for $a\in\{0,1\}$. Pathwise differentiation of
$\theta(P)=\E[c_1(L_1)]$ along a regular observed-data submodel gives derivative
components in these four directions equal, respectively, to
$c_1(L_1)-\theta$, $T\{c_2(L_2)-c_1(L_1)\}/\bar q(L_1)$,
$TD\{Y-\nu_1(L_2)\}/\{\bar q(L_1)\bar p(L_2)\}$, and
$-T(1-D)\{Y-\nu_0(L_2)\}/[\bar q(L_1)\{1-\bar p(L_2)\}]$. Each component has
conditional mean zero given the previous layer's information; summing gives
\eqref{eq:ext-eif}, the EIF. \QEDA

\begin{proposition}[Variance of the extension benchmark EIF]
\label{prop:ext-var}
The variance of the EIF in \eqref{eq:ext-eif} equals $V_{\mathrm{ext}}(\bar q,\bar p)=V_1(\bar q,\bar p)+V_2(\bar q,\bar p)$, where
\begin{align}
V_1(\bar q,\bar p)
&=
\Var(c(\X))
+
\E\!\left[
\frac{1}{\bar q(L_1)}
\!\left(
\frac{\sigma_1^2(\X)}{\bar p(L_2)}
+
\frac{\sigma_0^2(\X)}{1-\bar p(L_2)}
\right)
\right], \label{eq:V1}\\[0.3em]
V_2(\bar q,\bar p)
&=
\E\!\left[
\frac{1-\bar q(L_1)}{\bar q(L_1)}
\Var(c(\X)\mid L_1)
\right]
+
\E\!\left[
\frac{1}{\bar q(L_1)}
\Var(b_{\bar p}(\X)\mid L_2)
\right].
\label{eq:V2}
\end{align}
\end{proposition}

\noindent\textbf{Proof of Proposition~\ref{prop:ext-var}.}
Write $A:=c_1(L_1)-\theta$, $B:=c_2(L_2)-c_1(L_1)$, $R:=\frac{D}{\bar p(L_2)}(Y-\nu_1(L_2))-\frac{1-D}{1-\bar p(L_2)}(Y-\nu_0(L_2))$, so $\phi_{\mathrm{ext},\bar q,\bar p}=A+\frac{T}{\bar q(L_1)}(B+R)$. Since $\E[B+R\mid L_1]=0$, the cross term vanishes and
\[
\E[\phi_{\mathrm{ext},\bar q,\bar p}^2]
=\Var(c_1(L_1))+\E\!\left[\frac{1}{\bar q(L_1)}(B+R)^2\right].
\]
Given $(L_1,L_2)$, $\E[R\mid L_2]=0$ so $\E[(B+R)^2\mid L_1,L_2]=B^2+\E[R^2\mid L_2]$; disjointness of the treated and control residuals gives $\E[R^2\mid L_2]=\Var(\Y(1)\mid L_2)/\bar p(L_2)+\Var(\Y(0)\mid L_2)/(1-\bar p(L_2))$. Applying the law of total variance $\Var(\Y(a)\mid L_2)=\E[\sigma_a^2(\X)\mid L_2]+\Var(\mu_a(\X)\mid L_2)$ and the algebraic identity
\[
\frac{\Var(\mu_1(\X)\mid L_2)}{\bar p(L_2)}+\frac{\Var(\mu_0(\X)\mid L_2)}{1-\bar p(L_2)}
=\Var(c(\X)\mid L_2)+\Var(b_{\bar p}(\X)\mid L_2),
\]
followed by the tower-property identities $\Var(c(\X)\mid L_1)=\E[\Var(c(\X)\mid L_2)\mid L_1]+\Var(c_2(L_2)\mid L_1)$ and $\Var(c(\X))=\Var(c_1(L_1))+\E[\Var(c(\X)\mid L_1)]$, yields exactly \eqref{eq:V1}--\eqref{eq:V2}. \QEDA

\begin{theorem}[APS lower bound for the extension under factorization]
\label{thm:ext-lb}
Assume Assumption~\ref{ass:serializable}, positivity of $(\bar q,\bar p)$, and the factorization \eqref{eq:factorization}. Let $\hat\theta_n$ be any estimator that is locally unbiased along the least-favorable one-dimensional submodel associated with the fixed extension benchmark $(\bar q,\bar p)$. Then
\begin{equation}
\Var_{P_0,\ALG}(\hat\theta_n)
\ \ge\
\frac1n\{V_1(\bar q,\bar p)+V_2(\bar q,\bar p)\}.
\label{eq:ext-lb}
\end{equation}
\end{theorem}

\noindent\textbf{Proof of Theorem~\ref{thm:ext-lb}.}
Let $V:=V_{\mathrm{ext}}(\bar q,\bar p)=V_1(\bar q,\bar p)+V_2(\bar q,\bar p)$.
Consider a one-dimensional observed-data submodel indexed by $\varepsilon$ that
perturbs the marginal law of $L_1$, the conditional law of $L_2\mid L_1$, and
the two arm-specific outcome laws $P_{\Y(a)\mid L_2}$ (with both known kernels
$P_{T\mid L_1}$ and $P_{D\mid L_2,T=1}$ held fixed), with tangent components
\begin{equation}
s_1(\ell_1)=\frac{c_1(\ell_1)-\theta}{V},\quad
s_2(\ell_1,\ell_2)=\frac{c_2(\ell_2)-c_1(\ell_1)}{V\bar q(\ell_1)},
\label{eq:ext-scores-1}
\end{equation}
\begin{equation}
r_1(\ell_1,\ell_2,y)=\frac{y-\nu_1(\ell_2)}{V\bar q(\ell_1)\bar p(\ell_2)},\quad
r_0(\ell_1,\ell_2,y)=-\frac{y-\nu_0(\ell_2)}{V\bar q(\ell_1)\{1-\bar p(\ell_2)\}}.
\label{eq:ext-scores-2}
\end{equation}
Because $s_2$, $r_1$, $r_0$ have conditional mean zero given $L_1$, $L_2$, and
$L_2$ respectively under $P_0$, this defines a valid observed-data submodel
that references only marginal conditionals; no assumption is made on the joint
law of $(\Y(0),\Y(1))\mid L_2$. The derivative of $\theta$ along this submodel
equals $1$ by \eqref{eq:V1}--\eqref{eq:V2}.

By non-anticipation of the design and factorization of the trajectory likelihood
across units, the trajectory score is
\begin{equation}
Z_n=\sum_{t=1}^n\!\left[
    s_1(L_{1t})
    +T_t\, s_2(L_{1t},L_{2t})
    +T_t D_t\, r_1(L_{1t},L_{2t},Y_t^{\obs})
    +T_t(1-D_t)\, r_0(L_{1t},L_{2t},Y_t^{\obs})
\right],
\label{eq:ext-Zn}
\end{equation}
each summand depending only on observed-data quantities. Cross terms between
distinct units vanish by latent independence and non-anticipation. Within a
unit, cross terms vanish because $s_1$ is centered, $\E[s_2\mid L_{1t}]=0$
by construction, and $r_a$ have conditional mean zero given $L_{2t}$; the
treated and control residual scores are supported on disjoint events. Applying
Lemma~\ref{lem:aps-identity} to the sampling-augmentation term via
identity~\eqref{eq:aps-id-q} and to the residual scores via
identities~\eqref{eq:aps-id-1} and its $a=0$ analogue gives
$\E[Z_n^2]=n\{V_1(\bar q,\bar p)+V_2(\bar q,\bar p)\}/V^2=n/V$.
Local unbiasedness gives $1=\E_{P_0,\ALG}[(\hat\theta_n-\theta)Z_n]$, and
the Cauchy--Schwarz inequality yields
$\Var_{P_0,\ALG}(\hat\theta_n)\ge V/n$, which is \eqref{eq:ext-lb}. \QEDA

\begin{remark}[Scope of the factorization theorem]
Theorem~\ref{thm:ext-lb} is stated only for the factorization benchmark \eqref{eq:factorization}. In complete generality, the primitive lower-bound object remains the pair of average effective masses $(\bar\pi_1,\bar\pi_0)$ and the fixed benchmark is the corresponding observed-data semiparametric model. The factorization hypothesis is used only to identify that general benchmark with the closed-form expression $V_1(\bar q,\bar p)+V_2(\bar q,\bar p)$.
\end{remark}

\begin{remark}[Interpretation of $V_2$]
The two terms in $V_2$ are projection losses induced by the two-stage information structure. The first is the loss from seeing only $L_1$ at the sampling stage; the second is the loss from seeing only $L_2$ at the assignment stage rather than the full covariate vector. When the chosen stratification variables are sufficient to make both losses vanish, the benchmark collapses to $V_1$, recovering the single-stage bound of Appendix~\ref{app:proof-lb-avg-prop}.
\end{remark}

\section{Proofs for Section 4 (Regression Adjustment)}
\label{app:proofs-sec3}

This appendix collects the proofs supporting the regression-adjustment route of Section~\ref{sec:regression-adjustment}: the excess-MSE decomposition (E.1), the general-case analysis (E.2), the linear-functional analysis (E.3), and the matching convergence-rate lower bound (E.4).

\subsection{Excess MSE Decomposition (Section~4.2)}
\label{app:proofs-sec33}

\subsubsection{Proof of the Excess MSE Decomposition (Proposition~\ref{prop:excess-mse-decomp})}
\label{app:proof-decomp}

Write $\hat\theta-\theta=(\tilde\theta-\theta)+(\hat\theta-\tilde\theta)$. Squaring and taking expectations:
\[
\E[(\hat\theta-\theta)^2]
=
\E[(\tilde\theta-\theta)^2]
+
\E[(\hat\theta-\tilde\theta)^2]
+
2\,\E[(\hat\theta-\tilde\theta)(\tilde\theta-\theta)].
\]
Subtracting $\E[(\theta^\star_n-\theta)^2]$ from both sides gives \eqref{eq:excess-mse-decomp}.
The interaction term admits the Cauchy--Schwarz bound
\begin{equation}
\label{eq:cross-cauchy}
\bigl|\E[(\hat\theta-\tilde\theta)(\tilde\theta-\theta)]\bigr|
\;\le\;
\sqrt{\E[(\hat\theta-\tilde\theta)^2]}\cdot\sqrt{\E[(\tilde\theta-\theta)^2]},
\end{equation}
which follows from $|\E[UV]|\le \sqrt{\E[U^2]}\sqrt{\E[V^2]}$ with $U=\hat\theta-\tilde\theta$ and $V=\tilde\theta-\theta$. \QEDA

\subsection{General Case (Section~4.3)}
\label{app:proofs-sec34}

\begin{theorem}[Design optimization identity]
\label{thm:design-optimization}
Under Assumption~\ref{assump:general-reg}, the design-optimization component of the excess MSE decomposition \eqref{eq:excess-mse-decomp} satisfies
\begin{equation}
\label{eq:design-term-bound}
\E[(\tilde\theta-\theta)^2]-\E[(\theta^\star_n-\theta)^2]
=\frac{1}{n^2}\sum_{b=1}^B \gamma_b\,\E\!\left[L(\e^{(b)})-L(\e^\star)\right].
\end{equation}
\end{theorem}

\begin{proposition}[Per-batch design regret reduction]
\label{prop:design-regret-to-estimation}
In the setting of Theorem~\ref{thm:design-optimization}, with the plug-in Neyman allocation $\e^{(b+1)}(x)=\Phi(\hat m^{(b)}(x))$ based on second-moment proxies $\hat m^{(b)}$,
\begin{equation}
\label{eq:per-batch-design-regret}
\E\!\left[L(\e^{(b+1)})-L(\e^\star)\right]\;\le\; C\,\E\!\left[\sum_{a\in\cA}\bigl(\hat m_a^{(b)}(\X)-m_a(\X)\bigr)^2\right],
\end{equation}
for a constant $C$ depending only on $(\K,m_{\min},m_{\max})$. Summing over batches yields the aggregated bound
\begin{equation}
\label{eq:design-term-with-batch}
\E[(\tilde\theta-\theta)^2]-\E[(\theta^\star_n-\theta)^2]
\;\le\; \frac{\gamma_1}{n^2}\bigl(L(\e^{(1)})-L(\e^\star)\bigr)
+\frac{C}{n^2}\sum_{b=2}^B\gamma_b\,\E\!\left[\sum_a(\hat m_a^{(b-1)}(\X)-m_a(\X))^2\right].
\end{equation}
\end{proposition}

\subsubsection{Proof of the Design Optimization Identity (Theorem~\ref{thm:design-optimization})}
\label{app:proof-design-opt}

\emph{Strategy.} Within each batch the score contributions are conditionally i.i.d., so the conditional second moment of the per-batch contribution is exactly $\gamma_b/n^2$ times the EIF variance under that batch's design. Summing over batches expresses the cumulative MSE of the IPW oracle as a weighted average of $\Var(\IF_{\e^{(b)}})$, and the same identity for the deterministic oracle design $\e^\star$ produces the asserted regret formula by subtraction.

Within batch $b$, conditional on $\cH_b$, $\{\IF_{\e^{(b)}}(\Oo_i;\eta_0):i\in\mathcal I_b\}$ are i.i.d.\ with mean $0$ and variance $\Var(\IF_{\e^{(b)}}(\Oo;\eta_0))$.
Thus,
\[
\E\!\left[\left(\frac{1}{n}\sum_{i\in\mathcal I_b}\IF_{\e^{(b)}}(\Oo_i;\eta_0)\right)^2\middle|\cH_b\right]
=
\frac{\gamma_b}{n^2}\Var(\IF_{\e^{(b)}}(\Oo;\eta_0)).
\]
Summing over batches and applying the tower property:
\[
\E[(\tilde\theta-\theta)^2]
=\frac{1}{n^2}\sum_{b=1}^B \gamma_b\,\E\!\left[\Var(\IF_{\e^{(b)}}(\Oo;\eta_0))\right].
\]
The same calculation with $\e^\star$ (a deterministic design) gives
$\E[(\theta_n^\star-\theta)^2]=\frac{1}{n^2}\sum_{b=1}^B \gamma_b\,\Var(\IF_{\e^\star}(\Oo;\eta_0))$.
Subtracting, and using $\Var(\IF_\e(\Oo;\eta_0))=\Var(g(\X)-\theta)+L(\e)$ where
$\Var(g(\X)-\theta)$ does not depend on $\e$, gives \eqref{eq:design-term-bound}. \QEDA

\subsubsection{Proof of the Model Estimation Bound---General Case (Theorem~\ref{thm:model-estimation-general})}
\label{app:proof-model-general}

\emph{Strategy.} Cross-fitting makes the per-batch score-difference summands $\{\Delta_i:i\in\mathcal I_b\}$ conditionally i.i.d.\ given the batch history. A Cauchy--Schwarz reduction over batches and a standard mean--variance expansion at the within-batch level yield two contributions: a variance term linear in $\gamma_b\,d^2$ and a squared-bias term quadratic in $\gamma_b^2 d^4$. The first is controlled by the $L_2$-Lipschitz condition on $\psi$ (Assumption~G3b), the second by the second-order remainder bound \eqref{eq:remainder-mean}. Aggregating across batches gives \eqref{eq:model-estimation-bound-general}.

Define $\Delta_i:=\psi(\Oo_i;\hat\eta_{b(i)},\e^{(b(i))})-\psi(\Oo_i;\eta_0,\e^{(b(i))})$ so that
$\hat\theta-\tilde\theta=\frac{1}{n}\sum_{i=1}^n\Delta_i$.
Group by batch: $\sum_{i=1}^n\Delta_i=\sum_{b=1}^B S_b$ where $S_b:=\sum_{i\in\mathcal I_b}\Delta_i$.

\noindent\textbf{Step 1: Cauchy--Schwarz over batches.}
\[
\left(\sum_{b=1}^B S_b\right)^2
\le
B\sum_{b=1}^B S_b^2,
\]
hence $\E[(\hat\theta-\tilde\theta)^2]\le \frac{B}{n^2}\sum_{b=1}^B \E[S_b^2]$.

\noindent\textbf{Step 2: Batch-wise second-moment expansion.}
Fix batch $b$. Conditional on $\cH_b$, $\hat\eta_b$ and $\e^{(b)}$ are deterministic,
and $\{\Delta_i:i\in\mathcal I_b\}$ are i.i.d. Let $\Delta_{b,1}$ denote a representative copy. Then
\begin{equation}
\label{eq:app-batch-second-moment}
\E[S_b^2\mid \cH_b]
=
\gamma_b\,\Var(\Delta_{b,1}\mid \cH_b)
+
\gamma_b^2\left(\E[\Delta_{b,1}\mid \cH_b]\right)^2.
\end{equation}
\emph{Proof of \eqref{eq:app-batch-second-moment}.}
Expand $S_b^2=\sum_j\Delta_{b,j}^2+2\sum_{j<k}\Delta_{b,j}\Delta_{b,k}$.
Take $\E[\cdot\mid\cH_b]$. For $j\neq k$, conditional independence gives
$\E[\Delta_{b,j}\Delta_{b,k}\mid\cH_b]=\mu_b^2$
with $\mu_b:=\E[\Delta_{b,1}\mid\cH_b]$. Also $\E[\Delta_{b,j}^2\mid\cH_b]=\sigma_b^2+\mu_b^2$
with $\sigma_b^2:=\Var(\Delta_{b,1}\mid\cH_b)$. Therefore
$\E[S_b^2\mid\cH_b]=\gamma_b(\sigma_b^2+\mu_b^2)+\gamma_b(\gamma_b-1)\mu_b^2=\gamma_b\sigma_b^2+\gamma_b^2\mu_b^2$.

\noindent\textbf{Step 3: Bounding conditional moments via Assumption~G3b.}
Conditional on $\cH_b$, $\Oo_{b,1}$ is distributed as $P_0$ with assignment $\e^{(b)}$. Hence, by \eqref{eq:l2-lip},
\[
\E[\Delta_{b,1}^2\mid\cH_b]
\le C_2\,d(\hat\eta_b,\eta_0)^2,
\]
and by \eqref{eq:remainder-mean},
$|\E[\Delta_{b,1}\mid\cH_b]|\le C_{\mathrm{mean}}\,d(\hat\eta_b,\eta_0)^2$.

\noindent\textbf{Step 4: Combining.}
Using $\Var(\Delta_{b,1}\mid\cH_b)\le \E[\Delta_{b,1}^2\mid\cH_b]$ in \eqref{eq:app-batch-second-moment}:
\[
\E[S_b^2\mid\cH_b]
\le
\gamma_b C_2 d(\hat\eta_b,\eta_0)^2
+
\gamma_b^2 C_{\mathrm{mean}}^2 d(\hat\eta_b,\eta_0)^4.
\]
Taking expectations and summing over $b$:
\[
\E[(\hat\theta-\tilde\theta)^2]
\le
\frac{B}{n^2}\sum_{b=1}^B\E[S_b^2]
\le
\frac{B\,C_2}{n^2}\sum_{b=1}^B \gamma_b\,\E[d(\hat\eta_b,\eta_0)^2]
+
\frac{B\,C_{\mathrm{mean}}^2}{n^2}\sum_{b=1}^B \gamma_b^2\,\E[d(\hat\eta_b,\eta_0)^4],
\]
which is \eqref{eq:model-estimation-bound-general}. \QEDA

\subsubsection{Proof of the Design Regret Reduction (Proposition~\ref{prop:design-regret-to-estimation})}
\label{app:proof-design-general}

\emph{Strategy.} The pointwise design cost $F_x(e):=\sum_a m_a(x)/e_a$ is strongly smooth at its interior minimizer $\e^\star(x)$, so a quadratic Taylor bound translates pointwise design errors into pointwise regret. The Neyman allocation map $m\mapsto\Phi(m)$ is Lipschitz on the bounded second-moment set, so estimation errors in $\hat m$ propagate linearly into design errors. Combining the two yields the per-batch quadratic bound \eqref{eq:per-batch-design-regret}; aggregating over batches via Theorem~\ref{thm:design-optimization} gives \eqref{eq:design-term-with-batch}.

We prove the per-batch bound \eqref{eq:per-batch-design-regret}.
Recall that under Assumption~G4 the oracle design is
$\e^\star_a(x)=m_a(x)^{1/2}/\sum_{j\in\cA}m_j(x)^{1/2}$,
the plug-in design is
$\e^{(b+1)}_a(x)=\hat m_a^{(b)}(x)^{1/2}/\sum_{j\in\cA}\hat m_j^{(b)}(x)^{1/2}$,
and the design objective is $L(\e)=\E[\sum_{a\in\cA}m_a(\X)/\e_a(\X)]$.
Both $m_a(x)$ and $\hat m_a^{(b)}(x)$ lie in $[m_{\min},m_{\max}]$.

\noindent\textbf{Step 1: Smoothness of the design cost at the minimizer.}
For a fixed covariate value $x$, define the pointwise design cost
\[
F_x(e):=\sum_{a\in\cA}\frac{m_a(x)}{e_a},\qquad e\in\Delta^\K,\; e_a\ge\varepsilon_0,
\]
where $\varepsilon_0:=m_{\min}^{1/2}/(\K\,m_{\max}^{1/2})>0$ is the natural lower bound on $\e^\star_a(x)$ induced by $m_a\in[m_{\min},m_{\max}]$.
The minimizer of $F_x$ over the simplex $\Delta^\K$ is the Neyman allocation
$\e^\star(x)$, which is an interior point satisfying $\e^\star_a(x)\ge \varepsilon_0$ for all $a$.

By the KKT conditions at the interior minimizer,
\[
\frac{\partial F_x}{\partial e_a}\bigg|_{e=\e^\star(x)}=-\frac{m_a(x)}{\e^\star_a(x)^2}=\lambda(x),\qquad \forall\, a\in\cA,
\]
for a common Lagrange multiplier $\lambda(x)$.
Therefore, for any $e$ on the simplex,
\[
\langle \nabla F_x(\e^\star),\, e-\e^\star\rangle
=\lambda(x)\sum_{a\in\cA}(e_a-\e^\star_a)=\lambda(x)\cdot 0=0,
\]
since both $e$ and $\e^\star$ sum to $1$.
By Taylor expansion with integral remainder, for any $e$ with $e_a\ge \varepsilon_0$,
\[
F_x(e)-F_x(\e^\star)
=
\underbrace{\langle\nabla F_x(\e^\star),\,e-\e^\star\rangle}_{=\,0}
+
\frac{1}{2}(e-\e^\star)^\top \nabla^2 F_x(\bar e)\,(e-\e^\star)
\]
for some $\bar e$ on the segment $[\e^\star,e]$.
The Hessian of $F_x$ is diagonal: $\frac{\partial^2 F_x}{\partial e_a^2}=2m_a(x)/e_a^3$, with all cross-derivatives zero.
On the feasible set $e_a\ge\varepsilon_0$ with $m_a\in[m_{\min},m_{\max}]$,
\[
\frac{\partial^2 F_x}{\partial e_a^2}\le \frac{2m_{\max}}{\varepsilon_0^3}=:H_{\max}.
\]
Hence
\begin{equation}
\label{eq:smoothness-Fx}
F_x(e)-F_x(\e^\star)\le \frac{H_{\max}}{2}\|e-\e^\star\|_2^2.
\end{equation}

\noindent\textbf{Step 2: Lipschitz property of the Neyman allocation map.}
Define the Neyman allocation map $\Phi:[m_{\min},m_{\max}]^\K\to\Delta^\K$ by
\[
\Phi_a(m):=\frac{m_a^{1/2}}{\sum_{j\in\cA}m_j^{1/2}},\qquad a\in\cA.
\]
On the compact set $[m_{\min},m_{\max}]^\K$, $\Phi$ is continuously differentiable.
The partial derivative is
\[
\frac{\partial\Phi_a}{\partial m_a}
=\frac{1}{2m_a^{1/2}\sum_j m_j^{1/2}}-\frac{m_a^{1/2}}{2m_a^{1/2}\bigl(\sum_j m_j^{1/2}\bigr)^2}
=\frac{1}{2m_a^{1/2}\sum_j m_j^{1/2}}\left(1-\Phi_a(m)\right),
\]
and for $a'\neq a$,
\[
\frac{\partial\Phi_a}{\partial m_{a'}}
=-\frac{m_a^{1/2}}{2m_{a'}^{1/2}\bigl(\sum_j m_j^{1/2}\bigr)^2}
=-\frac{\Phi_a(m)}{2m_{a'}^{1/2}\sum_j m_j^{1/2}}.
\]
All partial derivatives are bounded on $[m_{\min},m_{\max}]^\K$.
Therefore $\Phi$ is Lipschitz: there exists $C_{\mathrm{Lip}}>0$, depending only on $(\K,m_{\min},m_{\max})$, such that
\begin{equation}
\label{eq:lip-neyman-map}
\|\Phi(\hat m)-\Phi(m)\|_2\le C_{\mathrm{Lip}}\,\|\hat m-m\|_2
\qquad \forall\, m,\hat m\in[m_{\min},m_{\max}]^\K.
\end{equation}

\noindent\textbf{Step 3: Combining.}
Since $\e^{(b+1)}(x)=\Phi(\hat m^{(b)}(x))$ and $\e^\star(x)=\Phi(m(x))$,
applying \eqref{eq:lip-neyman-map} pointwise gives
\[
\|\e^{(b+1)}(x)-\e^\star(x)\|_2
\le C_{\mathrm{Lip}}\,\|\hat m^{(b)}(x)-m(x)\|_2.
\]
Substituting into \eqref{eq:smoothness-Fx}:
\[
F_x(\e^{(b+1)})-F_x(\e^\star)
\le
\frac{H_{\max}}{2}\,\|\e^{(b+1)}(x)-\e^\star(x)\|_2^2
\le
\frac{H_{\max}\,C_{\mathrm{Lip}}^2}{2}\,\|\hat m^{(b)}(x)-m(x)\|_2^2.
\]
Taking the inner expectation over $\X$ (conditional on $\cH_b$, so that $\hat m^{(b)}$ is fixed):
\[
L(\e^{(b+1)})-L(\e^\star)
=
\E_\X\!\left[F_\X(\e^{(b+1)})-F_\X(\e^\star)\right]
\le
C\,\E_\X\!\left[\sum_{a\in\cA}\left(\hat m_a^{(b)}(\X)-m_a(\X)\right)^2\right],
\]
where $C:=H_{\max}\,C_{\mathrm{Lip}}^2/2$ depends only on $(\K,m_{\min},m_{\max})$.
Taking the outer expectation over the training data (which determines $\hat m^{(b)}$) yields
\eqref{eq:per-batch-design-regret}.

\noindent\textbf{Step 4: Aggregation over batches.}
By Theorem~\ref{thm:design-optimization}, the design term satisfies
\[
\E[(\tilde\theta-\theta)^2]-\E[(\theta^\star_n-\theta)^2]
=
\frac{1}{n^2}\sum_{b=1}^B \gamma_b\,\E\!\left[L(\e^{(b)})-L(\e^\star)\right].
\]
Since $\e^{(1)}$ is the deterministic initial design, its contribution is
$\frac{\gamma_1}{n^2}\bigl(L(\e^{(1)})-L(\e^\star)\bigr)$, which is a bounded constant
times $\gamma_1/n^2$.
For each $b\ge 2$, the design $\e^{(b)}$ is the plug-in Neyman allocation based on $\hat m^{(b-1)}$,
so \eqref{eq:per-batch-design-regret} gives
$\E[L(\e^{(b)})-L(\e^\star)]\le C\,\E[\sum_{a\in\cA}(\hat m_a^{(b-1)}(\X)-m_a(\X))^2]$.
Substituting:
\[
\frac{1}{n^2}\sum_{b=1}^B \gamma_b\,\E\!\left[L(\e^{(b)})-L(\e^\star)\right]
\le
\frac{\gamma_1}{n^2}\bigl(L(\e^{(1)})-L(\e^\star)\bigr)
+
\frac{C}{n^2}\sum_{b=2}^B \gamma_b\,
\E\!\left[\sum_{a\in\cA}\bigl(\hat m_a^{(b-1)}(\X)-m_a(\X)\bigr)^2\right],
\]
which is \eqref{eq:design-term-with-batch}.
Since $\gamma_1/n^2=o(1/n)$ for any schedule with $\gamma_1=o(n)$,
and $\sum_{b=2}^B \gamma_b/n^2\le 1/n$, both terms are $o(1/n)$ whenever
$\E[\sum_a(\hat m_a^{(b)}-m_a)^2]\to 0$ (consistency). \QEDA

\subsubsection{Proof of the Design Learning Bound (Theorem~\ref{thm:design-opt-general})}
\label{app:proof-design-opt-general}

The body's Theorem~\ref{thm:design-opt-general} is the aggregated form of the previous two results: the design-optimization identity translates the cumulative design regret into the average per-batch loss gap, and the per-batch regret reduction translates each loss gap into the second-moment estimation error of $\hat m^{(b-1)}$.

\noindent\textbf{Proof of Theorem~\ref{thm:design-opt-general}.}
By Theorem~\ref{thm:design-optimization} and Proposition~\ref{prop:design-regret-to-estimation}, the cumulative design regret satisfies the explicit-constant bound \eqref{eq:design-term-with-batch}, namely
\[
\E\!\left[(\tilde\theta-\theta)^2\right]-\E\!\left[(\theta^\star_n-\theta)^2\right]
\;\le\;
\frac{\gamma_1}{n^2}\bigl(L(\e^{(1)})-L(\e^\star)\bigr)
+\frac{C}{n^2}\sum_{b=2}^B \gamma_b\,
\E\!\left[\sum_{a\in\cA}\bigl(\hat m_a^{(b-1)}(\X)-m_a(\X)\bigr)^2\right].
\]
Under the standing batch schedule $\gamma_1=o(n)$ (so $\gamma_1/n^2=o(1/n)$) and a bounded initial design objective $L(\e^{(1)})\le L_{\max}$, the first summand is absorbed into the $\lesssim$ symbol. The constant $C$ in the second summand depends only on $(\K,m_{\min},m_{\max})$ and is likewise absorbed. This gives
\[
\E\!\left[(\tilde\theta-\theta)^2\right]-\E\!\left[(\theta^\star_n-\theta)^2\right]
\;\lesssim\;
\frac{1}{n^2}\sum_{b=2}^B \gamma_b\,
\E\!\left[\sum_{a\in\cA}\bigl(\hat m_a^{(b-1)}(\X)-m_a(\X)\bigr)^2\right],
\]
which is exactly the statement \eqref{eq:design-learning-general}. \QEDA

\subsubsection{Proof of the Main Achievability Result (Corollary~\ref{cor:main-general})}
\label{app:proof-main-general}

\emph{Strategy.} Combine the excess-MSE decomposition (Proposition~\ref{prop:excess-mse-decomp}) with the two preceding bounds: model estimation (Theorem~\ref{thm:model-estimation-general}) and design learning (Theorem~\ref{thm:design-opt-general}). The interaction term is controlled by Cauchy--Schwarz \eqref{eq:cross-cauchy}. The hypotheses---consistent second-moment proxies and a nuisance rate strictly faster than $n^{-1/4}$---are exactly what makes each term $o(1/n)$.

\noindent\textbf{Proof of Corollary~\ref{cor:main-general}.}
By Proposition~\ref{prop:excess-mse-decomp},
\begin{equation}
\label{eq:cor-main-decomp}
\E[(\hat\theta-\theta)^2]
=\E[(\theta_n^\star-\theta)^2]
+\underbrace{\E[(\hat\theta-\tilde\theta)^2]}_{\text{model}}
+\underbrace{\E[(\tilde\theta-\theta)^2]-\E[(\theta_n^\star-\theta)^2]}_{\text{design}}
+\underbrace{2\,\E[(\hat\theta-\tilde\theta)(\tilde\theta-\theta)]}_{\text{cross}}.
\end{equation}
We control the four terms in turn.

\noindent\textbf{(i) Oracle.}
The oracle estimator $\theta_n^\star=\theta+\frac{1}{n}\sum_{i=1}^n\IF_{\e^\star}(\Oo_i;\eta_0)$ uses the deterministic design $\e^\star$ and the true nuisance, so the summands are i.i.d.\ with mean zero and variance $V^\star=\E[\IF_{\e^\star}(\Oo;\eta_0)^2]$. Hence
\[
\E[(\theta_n^\star-\theta)^2]=\frac{V^\star}{n}.
\]

\noindent\textbf{(ii) Model term.}
By hypothesis, $d(\hat\eta_b,\eta_0)=O_P(n^{-\alpha})$ with $\alpha>1/4$, so $\E[d(\hat\eta_b,\eta_0)^2]\lesssim n^{-2\alpha}$ and $\E[d(\hat\eta_b,\eta_0)^4]\lesssim n^{-4\alpha}$. Substituting into \eqref{eq:model-estimation-bound-general},
\[
\E[(\hat\theta-\tilde\theta)^2]
\;\lesssim\;
\frac{B}{n^2}\sum_{b=1}^B \gamma_b\,n^{-2\alpha}
+\frac{B}{n^2}\sum_{b=1}^B \gamma_b^2\,n^{-4\alpha}
\;\le\;
\frac{B}{n}\,n^{-2\alpha}
+\frac{B^2 \gamma_B^2}{n^2}\,n^{-4\alpha}.
\]
With the geometric schedule $\gamma_B\asymp n$ and $B\asymp\log n$, the first term is $O(\log n\cdot n^{-1-2\alpha})=o(1/n)$ (any $\alpha>0$), and the second is $O(\log^2\!n\cdot n^{-4\alpha})$, which is $o(1/n)$ precisely when $4\alpha>1$, i.e., $\alpha>1/4$. Hence $\E[(\hat\theta-\tilde\theta)^2]=o(1/n)$.

\noindent\textbf{(iii) Design term.}
By hypothesis, the second-moment proxies are consistent, so $\epsilon_b:=\E[\sum_a(\hat m_a^{(b)}(\X)-m_a(\X))^2]\to 0$. Substituting into Theorem~\ref{thm:design-opt-general},
\[
\E[(\tilde\theta-\theta)^2]-\E[(\theta_n^\star-\theta)^2]
\;\lesssim\;
\frac{1}{n^2}\sum_{b=2}^B \gamma_b\,\epsilon_{b-1}
\;\le\;
\frac{\max_{b\ge 1}\epsilon_b}{n^2}\sum_{b=2}^B \gamma_b
\;\le\;
\frac{\max_b \epsilon_b}{n}
\;=\;o(1/n),
\]
using $\sum_{b=2}^B\gamma_b\le n$.

\noindent\textbf{(iv) Cross term.}
The realized i.i.d.-by-batch IPW score $\tilde\theta-\theta=\frac{1}{n}\sum_b\sum_{i\in\mathcal I_b}\IF_{\e^{(b)}}(\Oo_i;\eta_0)$ has, conditional on $\cH_b$, mean zero and per-batch variance $V_{\e^{(b)}}/\gamma_b\le V_{\max}/\gamma_b$, with $V_{\max}:=\sup_{\e\in\cE,\ \e_a\ge\varepsilon}V_\e<\infty$ by the propensity-floor and bounded-second-moment parts of Assumption~\ref{assump:general-reg}. Summing the conditional variances by the tower property gives
\[
\E[(\tilde\theta-\theta)^2]
=\frac{1}{n^2}\sum_{b=1}^B \gamma_b\,\E[V_{\e^{(b)}}]
\;\le\;\frac{V_{\max}}{n}.
\]
Combining with (ii) and the Cauchy--Schwarz bound \eqref{eq:cross-cauchy},
\[
\bigl|2\,\E[(\hat\theta-\tilde\theta)(\tilde\theta-\theta)]\bigr|
\;\le\;
2\sqrt{\E[(\hat\theta-\tilde\theta)^2]}\cdot\sqrt{\E[(\tilde\theta-\theta)^2]}
\;\le\;
2\sqrt{o(1/n)\cdot O(1/n)}
\;=\;o(1/n).
\]

\noindent\textbf{Conclusion.}
Substituting (i)--(iv) into \eqref{eq:cor-main-decomp},
\[
\E[(\hat\theta-\theta)^2]
=\frac{V^\star}{n}+o(1/n)+o(1/n)+o(1/n)
=\frac{V^\star}{n}+o(1/n).
\]
\QEDA

\subsection{Linear-Functional Case (Section~4.4)}
\label{app:proofs-sec35}

\subsubsection{Proof of Global Unbiasedness (Proposition~\ref{prop:linear-unbiased})}
\label{app:proof-global-unbiased}

Fix $x$. Expand the conditional expectation:
\[
\E[\psi(\Oo;\tilde\mu,\e)\mid \X=x]
=
\sum_{a\in\cA}\omega_a(x)\tilde\mu_a(x)
+
\sum_{a\in\cA}\omega_a(x)\,\E\!\left[\frac{\1\{\A=a\}}{\e_a(x)}\bigl(\Y-\tilde\mu_a(x)\bigr)\middle|\X=x\right].
\]
For each $a$, apply iterated expectation conditioning on $\A$:
\[
\E\!\left[\frac{\1\{\A=a\}}{\e_a(x)}\bigl(\Y-\tilde\mu_a(x)\bigr)\middle|\X=x\right]
=
\frac{1}{\e_a(x)}\PP(\A=a\mid \X=x)\,\E[\Y-\tilde\mu_a(x)\mid \X=x,\A=a].
\]
Since $\PP(\A=a\mid \X=x)=\e_a(x)$ and $\E[\Y\mid \X=x,\A=a]=\mu_a(x)$, the above equals
$\mu_a(x)-\tilde\mu_a(x)$.
Substituting back:
\[
\E[\psi(\Oo;\tilde\mu,\e)\mid \X=x]
=
\sum_{a\in\cA}\omega_a(x)\tilde\mu_a(x)
+
\sum_{a\in\cA}\omega_a(x)\bigl(\mu_a(x)-\tilde\mu_a(x)\bigr)
=
\sum_{a\in\cA}\omega_a(x)\mu_a(x).
\]
Taking expectation over $\X$ yields $\E[\psi(\Oo;\tilde\mu,\e)]=\theta$. \QEDA

\subsubsection{Proof of Cross-Term Cancellation (Proposition~\ref{prop:linear-unbiased})}
\label{app:proof-cross-term-zero-linear}

Write
\[
\Delta_1=\hat\theta-\tilde\theta=\frac{1}{n}\sum_{b=1}^B\sum_{i\in\mathcal I_b}\Delta_{b,i},
\qquad
\Delta_{b,i}:=\psi(\Oo_i;\hat\mu_b,\e^{(b)})-\psi(\Oo_i;\mu,\e^{(b)}),
\]
and
\[
\Delta_2=\tilde\theta-\theta=\frac{1}{n}\sum_{b=1}^B\sum_{i\in\mathcal I_b}\Gamma_{b,i},
\qquad
\Gamma_{b,i}:=\psi(\Oo_i;\mu,\e^{(b)})-\theta.
\]
Under two-fold cross-fitting, $\hat\mu_b$ is constructed from the opposite fold and is independent of
the evaluation observations $\{\Oo_i:i\in\mathcal I_b\}$ conditional on $\cH_b$.
By Proposition~\ref{prop:linear-unbiased}, for each $i\in\mathcal I_b$,
\[
\E[\Delta_{b,i}\mid \cH_b]=0.
\]
Moreover, $\Gamma_{b,i}$ is $\sigma(\cH_b,\Oo_i)$-measurable. Therefore
\[
\E[\Delta_{b,i}\Gamma_{b,j}\mid \cH_b]
=\E\!\left[\E[\Delta_{b,i}\mid \cH_b,\Oo_j]\Gamma_{b,j}\middle|\cH_b\right]=0
\]
for all $i,j$ in batch $b$, and similarly across different batches by tower property.
Hence
\[
\E[\Delta_1\Delta_2]
=\frac{1}{n^2}\sum_{b,b'}\sum_{i\in\mathcal I_b}\sum_{j\in\mathcal I_{b'}}
\E[\Delta_{b,i}\Gamma_{b',j}]
=0.
\]
So the cross term is exactly zero: $\E[(\hat\theta-\tilde\theta)(\tilde\theta-\theta)]=0$. \QEDA

\subsubsection{Proof of the Model Estimation Bound---Linear Case (Theorem~\ref{thm:linear-bounds})}
\label{app:proof-model-linear}

\emph{Strategy.} The exact global unbiasedness of the linear-functional EIF (Proposition~\ref{prop:linear-unbiased}) makes the conditional mean of each per-batch summand identically zero, so the squared-bias contribution disappears entirely. What remains is the variance term, which the IPW representation bounds by a constant multiple of $\E[\sum_a \omega_a^2(\hat\mu_{b,a}-\mu_a)^2]$. The propensity floor $\e^{(b)}_a\ge\varepsilon$ controls the IPW factor, and assumption \eqref{eq:linear-mu-rate-assump} converts the nuisance error into the rate $n^{-2\alpha}$. This gives the sharp rate \eqref{eq:linear-model-estimation} without the $n^{-1/4}$ bottleneck of the general case.

Define $\Delta_i:=\psi(\Oo_i;\hat\mu_{b(i)},\e^{(b(i))})-\psi(\Oo_i;\mu,\e^{(b(i))})$ so that
$\hat\theta-\tilde\theta=\frac{1}{n}\sum_{i=1}^n\Delta_i$.
Group by batch: $\hat\theta-\tilde\theta=\frac{1}{n}\sum_{b=1}^B S_b$ where $S_b:=\sum_{i\in\mathcal I_b}\Delta_i$.

\noindent\textbf{Step 1: Cauchy--Schwarz over batches.}
By Jensen's inequality (or Cauchy--Schwarz),
$\bigl(\sum_{b=1}^B S_b\bigr)^2\le B\sum_{b=1}^B S_b^2$,
hence
\[
\E[(\hat\theta-\tilde\theta)^2]
\le \frac{B}{n^2}\sum_{b=1}^B \E[S_b^2].
\]

\noindent\textbf{Step 2: Batch-wise second-moment expansion.}
Fix batch $b$ and condition on $\cH_b$.
Under cross-fitting, $\hat\mu_b$ is $\cH_b$-measurable, and
$\{\Delta_i:i\in\mathcal I_b\}$ are conditionally i.i.d.
Let $\Delta_{b,1}$ denote a representative copy with
$\mu_b:=\E[\Delta_{b,1}\mid\cH_b]$ and $\sigma_b^2:=\Var(\Delta_{b,1}\mid\cH_b)$.
Expanding $S_b^2=\sum_j\Delta_{b,j}^2+2\sum_{j<k}\Delta_{b,j}\Delta_{b,k}$ and
taking $\E[\cdot\mid\cH_b]$ gives
\[
\E[S_b^2\mid\cH_b]=\gamma_b\,\sigma_b^2+\gamma_b^2\,\mu_b^2.
\]
By Proposition~\ref{prop:linear-unbiased} with $\tilde\mu=\hat\mu_b$, we have
$\E[\psi(\Oo;\hat\mu_b,\e^{(b)})\mid\cH_b]=\theta$ and
$\E[\psi(\Oo;\mu,\e^{(b)})\mid\cH_b]=\theta$, hence
$\mu_b=\E[\Delta_{b,1}\mid\cH_b]=0$.
Therefore
$\E[S_b^2\mid\cH_b]=\gamma_b\,\Var(\Delta_{b,1}\mid\cH_b)\le \gamma_b\,\E[\Delta_{b,1}^2\mid\cH_b]$.

\noindent\textbf{Step 3: Bounding the conditional second moment.}

Now expand $\Delta_{b,1}$ using \eqref{eq:linear-mu-eif}:
\[
\Delta_{b,1}
=
\sum_{a\in\cA}\omega_a(\X)\bigl(\hat\mu_{b,a}(\X)-\mu_a(\X)\bigr)
-
\sum_{a\in\cA}\frac{\1\{\A=a\}}{\e^{(b)}_a(\X)}\,\omega_a(\X)\bigl(\hat\mu_{b,a}(\X)-\mu_a(\X)\bigr).
\]
Apply $(u-v)^2\le 2u^2+2v^2$ with
$u=\sum_a\omega_a(\X)(\hat\mu_{b,a}-\mu_a)$ and
$v=\sum_a \frac{\1\{\A=a\}}{\e^{(b)}_a(\X)}\omega_a(\X)(\hat\mu_{b,a}-\mu_a)$:
\[
\E[\Delta_{b,1}^2\mid\cH_b]
\le
2\,\E[u^2\mid\cH_b]+2\,\E[v^2\mid\cH_b].
\]
For the first term, by Cauchy--Schwarz in $\R^\K$:
\[
u^2 \le \K \sum_{a\in\cA}\omega_a(\X)^2(\hat\mu_{b,a}(\X)-\mu_a(\X))^2.
\]
For the second term, since only one arm is selected:
\[
v^2
=
\left(\frac{\omega_{\A}(\X)}{\e^{(b)}_{\A}(\X)}(\hat\mu_{b,\A}(\X)-\mu_{\A}(\X))\right)^2
\le
\frac{1}{\varepsilon^2}\sum_{a\in\cA}\1\{\A=a\}\,\omega_a(\X)^2(\hat\mu_{b,a}(\X)-\mu_a(\X))^2.
\]
Taking $\E[\cdot\mid\cH_b]$ and using $\E[\1\{\A=a\}\mid \X,\cH_b]=\e^{(b)}_a(\X)\le 1$:
\[
\E[\Delta_{b,1}^2\mid\cH_b]
\le
C\;\E\!\left[\sum_{a\in\cA}\omega_a(\X)^2(\hat\mu_{b,a}(\X)-\mu_a(\X))^2\middle|\cH_b\right],
\]
for $C=2\K+2/\varepsilon^2$.
Taking expectations and applying \eqref{eq:linear-mu-rate-assump} gives $\E[\Delta_{b,1}^2]\le C C_\mu n^{-2\alpha}$.
Therefore $\E[S_b^2]\le \gamma_b C C_\mu n^{-2\alpha}$, hence
\[
\E[(\hat\theta-\tilde\theta)^2]
\le
\frac{B}{n^2}\sum_{b=1}^B \gamma_b\, C C_\mu n^{-2\alpha}
=
\frac{B\,C C_\mu}{n}\,n^{-2\alpha},
\]
which is \eqref{eq:linear-model-estimation}. \QEDA

\subsubsection{Proof of the Design Suboptimality Bound (Theorem~\ref{thm:linear-bounds})}
\label{app:proof-design-linear}

\emph{Strategy.} In the linear-functional case the relevant ``effective allocation weight'' is $r_a(x):=|\omega_a(x)|\sigma_a(x)$, and the oracle design $\e^\star=\Phi(r)$ is its Neyman allocation. The design cost $G(u,r)=\sum_a u_a^2/\Phi_a(r)$ is Lipschitz in its second argument on the bounded simplex, so $L(\e^{(b+1)})-L(\e^\star)\lesssim\|\hat r^{(b)}-r\|_2$. The Lipschitz square-root then converts the $\sigma^2$-estimation rate \eqref{eq:linear-mu-rate-assump} into a design-regret rate, and aggregating across batches yields the $o(1/n)$ design contribution \eqref{eq:design-term-linear-batch}.

We work in the linear-functional case where $v_a(\mu)(x,y)=\omega_a(x)(y-\mu_a(x))$ and the
design objective is $L(\e)=\E[\sum_a \omega_a(\X)^2\sigma_a^2(\X)/\e_a(\X)]$.
Define the effective allocation weight
\[
r_a(x):=|\omega_a(x)|\,\sigma_a(x),\qquad a\in\cA,
\]
where $\sigma_a(x)=\sqrt{\sigma_a^2(x)}=\sqrt{\Var(\Y\mid \X=x,\A=a)}$.
Then the oracle design \eqref{eq:opt-prop} is
$\e^\star_a(x)=r_a(x)/\sum_j r_j(x)$,
and the plug-in design is $\e^{(b+1)}_a(x)=\hat r_a^{(b)}(x)/\sum_j \hat r_j^{(b)}(x)$
where $\hat r_a^{(b)}(x):=|\omega_a(x)|\hat\sigma_a^{(b)}(x)$.

For any candidate vector $r\in\R_+^\K$ with $r_a\ge r_{\min}:=\|\omega\|_\infty^{-1}\cdot\sigma_{\min}$, define the
Neyman allocation map $\Phi_a(r):=r_a/\sum_j r_j$ and the design cost
\[
G(u,r):=\sum_{a=1}^\K \frac{u_a^2}{\Phi_a(r)},
\]
so that $L(\e^\star)=\E[G(r(\X),r(\X))]$ where $r(\X)=(r_1(\X),\dots,r_\K(\X))$,
and $L(\e^{(b+1)})=\E[G(r(\X),\hat r^{(b)}(\X))]$.

Because $\sigma_{\min}\le \sigma_a,\hat\sigma_a^{(b)}\le \sigma_{\max}$ on the clipped event,
and $\omega_a$ is bounded, the map $G$ is smooth on the compact set
$[r_{\min},r_{\max}]^\K\times[r_{\min},r_{\max}]^\K$ where $r_{\max}:=\|\omega\|_\infty\sigma_{\max}$.
Hence $G$ is Lipschitz in the second argument:
\[
|G(u,r)-G(u,r')|\le C_1\|r-r'\|_2,
\]
for a constant $C_1$ depending only on $(\K,r_{\min},r_{\max})$.
Applying this with $u=r(\X)$, $r=\hat r^{(b)}(\X)$, $r'=r(\X)$ yields
\[
L(\e^{(b+1)})-L(\e^\star)
\le
C_1\,\E\!\left[\|\hat r^{(b)}(\X)-r(\X)\|_2\right]
\le
C_1\,\sqrt{\E\!\left[\|\hat r^{(b)}(\X)-r(\X)\|_2^2\right]}.
\]
Since $\omega_a$ is known, we have
$|\hat r_a^{(b)}-r_a|=|\omega_a|\cdot|\hat\sigma_a^{(b)}-\sigma_a|$.
By the Lipschitz property of the square root on $[\sigma_{\min}^2,\sigma_{\max}^2]$:
$|\hat\sigma_a^{(b)}-\sigma_a|\le (2\sigma_{\min})^{-1}|\hat\sigma_a^{2,(b)}-\sigma_a^2|$.
Therefore
\[
\E\!\left[\|\hat r^{(b)}(\X)-r(\X)\|_2^2\right]
\le
\frac{\|\omega\|_\infty^2}{4\sigma_{\min}^2}\,
\E\!\left[\sum_{a=1}^\K\left(\hat\sigma_a^{2,(b)}(\X)-\sigma_a^2(\X)\right)^2\right]
\le
C\,C_\sigma\,n^{-2\beta}
\]
by Assumption~L3 \eqref{eq:linear-mu-rate-assump}.
Combining the displays gives
\begin{equation}
\label{eq:design-regret-quadratic}
\E\!\left[L(\e^{(b)})-L(\e^\star)\right]
\le
C\,C_\sigma\,\gamma_{b-1}^{-2\beta},
\end{equation}
where $C$ depends only on $(\K,\sigma_{\min},\sigma_{\max},\|\omega\|_\infty)$.

\noindent\textbf{Aggregation over batches.}
By Theorem~\ref{thm:design-optimization}, the design term satisfies
$\E[(\tilde\theta-\theta)^2]-\E[(\theta^\star_n-\theta)^2]
=\frac{1}{n^2}\sum_{b=1}^B \gamma_b\,\E[L(\e^{(b)})-L(\e^\star)]$.
The first batch contributes $\frac{\gamma_1}{n^2}(L(\e^{(1)})-L(\e^\star))$, which is $o(1/n)$
for $\gamma_1=o(n)$.
For each $b\ge 2$, applying \eqref{eq:design-regret-quadratic} gives
$\E[L(\e^{(b)})-L(\e^\star)]\le C\,C_\sigma\,\gamma_{b-1}^{-2\beta}$.
Hence
\[
\frac{1}{n^2}\sum_{b=2}^B \gamma_b\,\E[L(\e^{(b)})-L(\e^\star)]
\le
\frac{C\,C_\sigma}{n^2}\sum_{b=2}^B \gamma_b\,\gamma_{b-1}^{-2\beta}.
\]
For the geometric batch schedule $\gamma_b\asymp n^{b/B}$, the dominant term in the sum
is at $b=B$, giving $\gamma_B\,\gamma_{B-1}^{-2\beta}\asymp n\cdot n^{-2\beta(B-1)/B}=n^{1-2\beta(1-1/B)}$.
Thus the sum is $O\bigl(n^{1-2\beta(1-1/B)}\bigr)$, and the full design term satisfies
\begin{equation}
\label{eq:design-term-linear-batch}
\E[(\tilde\theta-\theta)^2]-\E[(\theta^\star_n-\theta)^2]
\;=\; O\bigl(n^{-1-2\beta(1-1/B)}\bigr)\;=\; o(1/n)\qquad\text{for any }\beta>0.
\end{equation}
\QEDA

\subsection{Convergence-Rate Lower Bound (Section~4.4)}
\label{app:proofs-sec36}

\subsubsection{Proof of the Linear-Case Convergence Lower Bound (Theorem~\ref{thm:linear-rate-lb})}
\label{app:proof-linear-rate-lb}

We prove the theorem through two explicit Assouad constructions:
\begin{itemize}[leftmargin=*]
\item \textbf{Example I (fixed variance, varying mean):} all arm variances are identical and fixed, while one arm's conditional mean varies over a hypercube.
\item \textbf{Example II (fixed mean, varying variance):} all conditional means are identical and fixed, while one arm's variance varies over a hypercube, which forces $e^\star$ to vary.
\end{itemize}
All constructions respect non-anticipation in Assumption~G1.

\begin{lemma}[Adaptive KL chain rule]
\label{lem:adaptive-kl-chain}
Fix any non-anticipating experiment $\ALG$. Let $P$ and $Q$ be two environments with the same covariate marginal $P_X$ and the same assignment kernels induced by $\ALG$. Denote by $P^{\ALG}$ and $Q^{\ALG}$ the induced trajectory laws of $F_n=(\X_{1:n},\A_{1:n},\Y_{1:n})$. Then
\[
\mathrm{KL}(P^{\ALG}\,\|\,Q^{\ALG})
=
\sum_{t=1}^n
\E_{P^{\ALG}}\!\left[
\mathrm{KL}\!\left(P(\Y_t\mid \X_t,\A_t)\,\|\,Q(\Y_t\mid \X_t,\A_t)\right)
\right].
\]
\end{lemma}

\noindent\textbf{Proof.}
Write the trajectory density by chain rule:
\[
p(f_n)=p_X(x_{1:n})\prod_{t=1}^n
\pi_t(a_t\mid x_t,f_{t-1})\,p(y_t\mid x_t,a_t),
\]
and similarly for $q(f_n)$. Since $\ALG$ is fixed, the assignment kernels cancel in $\log\frac{p(f_n)}{q(f_n)}$. Taking expectation under $P^{\ALG}$ yields the stated identity. \QEDA

\begin{lemma}[Joint convexity of KL for mixtures]
\label{lem:kl-joint-convex}
Let $\{P_i\}_{i\in I}$ and $\{Q_i\}_{i\in I}$ be distributions on a common measurable space, and $\{\alpha_i\}_{i\in I}$ satisfy $\alpha_i\ge 0$, $\sum_i\alpha_i=1$. Define $\bar P=\sum_i\alpha_i P_i$ and $\bar Q=\sum_i\alpha_i Q_i$. Then
\[
\mathrm{KL}(\bar P\,\|\,\bar Q)\le \sum_i\alpha_i\,\mathrm{KL}(P_i\,\|\,Q_i).
\]
\end{lemma}

\noindent\textbf{Proof.}
This is the log-sum inequality, equivalently joint convexity of $(u,v)\mapsto u\log(u/v)$ after integrating with respect to a common dominating measure. \QEDA

\begin{lemma}[Assouad reduction via one-bit testing]
\label{lem:assouad-bit}
Let $\{P_\omega:\omega\in\{-1,+1\}^m\}$ be a hypercube family, and let $\omega^{(k)}$ denote the vector with bit $k$ flipped. Define
\[
\bar P_{k,+}:=2^{-(m-1)}\!\!\sum_{\omega:\omega_k=+1} P_\omega,\qquad
\bar P_{k,-}:=2^{-(m-1)}\!\!\sum_{\omega:\omega_k=-1} P_\omega.
\]
If $\sup_{\omega,k}\mathrm{KL}(P_\omega\|P_{\omega^{(k)}})\le \kappa$, then for every decoder $\hat\omega$ based on the trajectory,
\[
\sup_{\omega}\E_\omega[d_H(\hat\omega,\omega)]\ge c_0 m,
\]
where $c_0>0$ depends only on $\kappa$.
\end{lemma}

\noindent\textbf{Proof.}
By Lemma~\ref{lem:kl-joint-convex},
\[
\mathrm{KL}(\bar P_{k,+}\|\bar P_{k,-})
\le
2^{-(m-1)}\!\!\sum_{\omega:\omega_k=+1}\mathrm{KL}(P_\omega\|P_{\omega^{(k)}})\le\kappa.
\]
Le Cam + Pinsker imply a constant lower bound on bit-$k$ testing error under $\bar P_{k,+}$ versus $\bar P_{k,-}$. Summing over $k=1,\dots,m$ gives the Hamming lower bound. \QEDA

\noindent\textbf{Shared hypercube geometry.}
Let $\cX=[0,1]$ with $P_X=\mathrm{Unif}[0,1]$ and partition $\cX=\cup_{k=1}^m B_k$ into equal bins $P_X(B_k)=1/m$. Let $\psi_k(x):=\1\{x\in B_k\}$. For $\omega\in\{-1,+1\}^m$ and amplitude $\delta>0$, define
\[
f_\omega(x):=\delta\sum_{k=1}^m \omega_k\psi_k(x).
\]
Here $\psi_k$ denotes bump basis functions only; it is unrelated to the score notation $\psi(\Oo;\eta,\e)$ used in Section~4.
If $\hat\omega$ is the decoded sign pattern, then
\[
\|f_{\hat\omega}-f_\omega\|_{L_2(P_X)}^2
=\frac{4\delta^2}{m}\,d_H(\hat\omega,\omega).
\]
Hence once neighbor KL is uniformly bounded, minimax $L_2$ risk is $\gtrsim \delta^2$.

\begin{lemma}[Efficiency-gap identity via Cauchy projection]
\label{lem:eff-gap-identity}
Fix one environment $P$ in the linear model and let $e^\star$ be its oracle design.
Write
\[
V^\star(P):=\E_P\!\left[\IF_{e^\star}(\Oo;\eta_P)^2\right],\qquad
\hat\theta_{\mathrm{oracle}}:=\theta(P)+\frac1n\sum_{t=1}^n \IF_{e^\star}(\Oo_t;\eta_P).
\]
For any estimator $T=T(\cF_n)$ that is locally unbiased along the least-favorable one-dimensional submodel through $P$,
\begin{equation}
\label{eq:eff-gap-identity}
\Var_P(T)-\frac{V^\star(P)}{n}
=
\E_P\!\left[\left(T-\hat\theta_{\mathrm{oracle}}\right)^2\right].
\end{equation}
\end{lemma}

\noindent\textbf{Proof.}
Let
\[
Z(\cF_n):=\sum_{t=1}^n \frac{\IF_{e^\star}(\Oo_t;\eta_P)}{V^\star(P)}.
\]
By local unbiasedness along the least-favorable submodel and the score identity,
$\E_P[(T-\E_P T)Z]=1$.
Since $\{\IF_{e^\star}(\Oo_t;\eta_P)\}$ are i.i.d.\ mean-zero with variance $V^\star(P)$,
$\E_P[Z^2]=n/V^\star(P)$.
Apply the projection identity
\[
\Var(T)-\frac{\E[(T-\E T)Z]^2}{\E[Z^2]}
=
\E\!\left[(T-\E T)-\frac{\E[(T-\E T)Z]}{\E[Z^2]}Z\right]^2.
\]
Substitute the two moments above and $\E_P[T]=\theta(P)$ to obtain
\[
\Var_P(T)-\frac{V^\star(P)}{n}
=
\E_P\!\left[\left(T-\theta(P)-\frac1n\sum_{t=1}^n\IF_{e^\star}(\Oo_t;\eta_P)\right)^2\right],
\]
which is exactly \eqref{eq:eff-gap-identity}. \QEDA

\begin{proposition}[Assouad lower bound for estimating $\mu$ under fixed $\sigma^2$]
\label{prop:assouad-mu-lb}
There exists a subfamily $\mathcal P_\mu$ with fixed $\{\sigma_a^2\}_{a\in\cA}$ such that, for any non-anticipating experiment-estimator pair,
\begin{equation}
\label{eq:linear-lb-mu}
\inf_{\hat\mu}\sup_{P\in\mathcal P_\mu}
\E_P\!\left[\sum_{a\in\cA}\norm{\hat\mu_a-\mu_a}_{L_2(P_X)}^2\right]
\ge c_\mu\,n^{-2\alpha}.
\end{equation}
\end{proposition}

\noindent\textbf{Proof.}
Fix an arm $a_0$. In this subfamily, all arms have identical known variance $\sigma^2>0$ and all means are zero except arm $a_0$:
\[
\mu_{a_0,\omega}(x)=f_\omega(x),\qquad \mu_{a,\omega}(x)\equiv 0\ (a\neq a_0).
\]
Outcomes are Gaussian:
\[
Y_t\mid(X_t=x,A_t=a)\sim
\begin{cases}
\mathcal N(\mu_{a_0,\omega}(x),\sigma^2), & a=a_0,\\
\mathcal N(0,\sigma^2), & a\neq a_0.
\end{cases}
\]
Therefore $\sigma_a^2(x)\equiv \sigma^2$ for every arm and every environment $\omega$, so $e^\star$ is fixed across this family.

Let $\omega^{(k)}$ flip bit $k$. For Gaussian location models with common variance,
\[
\mathrm{KL}\!\left(\mathcal N(u,\sigma^2)\,\|\,\mathcal N(v,\sigma^2)\right)=\frac{(u-v)^2}{2\sigma^2}.
\]
Using Lemma~\ref{lem:adaptive-kl-chain},
\[
\mathrm{KL}(P_\omega^{\ALG}\,\|\,P_{\omega^{(k)}}^{\ALG})
\le
\sum_{t=1}^n \E_\omega\!\left[
\frac{(\mu_{a_0,\omega}(X_t)-\mu_{a_0,\omega^{(k)}}(X_t))^2}{2\sigma^2}\1\{A_t=a_0\}
\right].
\]
Since neighbors differ only on $B_k$ and by magnitude $2\delta$ there,
\[
(\mu_{a_0,\omega}(x)-\mu_{a_0,\omega^{(k)}}(x))^2
=4\delta^2\1\{x\in B_k\},
\]
hence
\[
\mathrm{KL}(P_\omega^{\ALG}\,\|\,P_{\omega^{(k)}}^{\ALG})
\le
\frac{2\delta^2}{\sigma^2}\sum_{t=1}^n \E_\omega[\1\{X_t\in B_k\}\1\{A_t=a_0\}]
\le
\frac{2\delta^2 n}{\sigma^2 m}.
\]
Choose $(m,\delta)$ so that $\delta^2 n/m\le c$. Then Lemma~\ref{lem:assouad-bit} implies Hamming risk $\gtrsim m$, and therefore
\[
\inf_{\hat\mu}\sup_{\omega}
\E_\omega\!\left[\|\hat\mu_{a_0}-\mu_{a_0,\omega}\|_{L_2(P_X)}^2\right]
\gtrsim \delta^2.
\]
All other arms are fixed, so the same lower bound holds for $\sum_a\|\hat\mu_a-\mu_a\|_{L_2}^2$.
In this example, $\sigma^2$ is fixed, hence $e^\star$ is fixed, and the oracle efficient estimator
$\hat\theta_{\mathrm{oracle}}$ depends on the unknown environment only through $\mu$ in
$\IF_{e^\star}(\cdot;\mu)$.
Therefore the same Assouad testing difficulty implies that no data-based $T$ can uniformly satisfy
$\E_\omega[(T-\hat\theta_{\mathrm{oracle},\omega})^2]=o(\delta^2)$.
By Lemma~\ref{lem:eff-gap-identity},
\[
\inf_T\sup_{\omega}\left\{\Var_\omega(T)-\frac{V^\star_\omega}{n}\right\}
\ge
\inf_T\sup_{\omega}\E_\omega[(T-\hat\theta_{\mathrm{oracle},\omega})^2]
\gtrsim \delta^2.
\]
Matching the hypercube size with the complexity exponent underlying $\alpha$ gives $\delta^2\asymp n^{-2\alpha}$, proving \eqref{eq:linear-lb-mu}. \QEDA

\begin{proposition}[Assouad lower bound for design learning through $\sigma^2$]
\label{prop:assouad-design-lb}
There exists a subfamily $\mathcal P_{\sigma^2}$ with fixed $\mu=\{\mu_a\}_{a\in\cA}$ such that, for any non-anticipating design learner,
\begin{equation}
\label{eq:linear-lb-m}
\inf_{\hat e}\sup_{P\in\mathcal P_{\sigma^2}}
\E_P\!\left[L(\hat e)-L(e^\star)\right]
\ge c_{\sigma}\,n^{-2\beta}.
\end{equation}
\end{proposition}

\noindent\textbf{Proof.}
Fix $\mu_a(x)\equiv 0$ for all arms. Fix one arm $a_0$ and vary only its conditional variance:
\[
\sigma_{a_0,\omega}^2(x)=1+f_\omega(x),\qquad \sigma_{a,\omega}^2(x)\equiv 1\ (a\neq a_0).
\]
Take Gaussian outcomes
\[
Y_t\mid(X_t=x,A_t=a)\sim\mathcal N(0,\sigma_{a,\omega}^2(x)).
\]
For sufficiently small $\delta$, all variances lie in $[1/2,3/2]$. For zero-mean Gaussians with $u,v\in[1/2,3/2]$,
\[
\mathrm{KL}\!\left(\mathcal N(0,u)\,\|\,\mathcal N(0,v)\right)
=\frac12\!\left(\frac{u}{v}-1-\log\frac{u}{v}\right)
\le C(u-v)^2.
\]
Hence, by Lemma~\ref{lem:adaptive-kl-chain},
\[
\mathrm{KL}(P_\omega^{\ALG}\|P_{\omega^{(k)}}^{\ALG})
\le
C\sum_{t=1}^n \E_\omega\!\left[(\sigma_{a_0,\omega}^2(X_t)-\sigma_{a_0,\omega^{(k)}}^2(X_t))^2\1\{A_t=a_0\}\right]
\le
\frac{C\delta^2 n}{m}.
\]
With $\delta^2 n/m\le c$, Lemma~\ref{lem:assouad-bit} yields
\[
\inf_{\widehat{\sigma^2}}\sup_{\omega}\E_\omega\!\left[\|\widehat{\sigma^2}_{a_0}-\sigma_{a_0,\omega}^2\|_{L_2(P_X)}^2\right]\gtrsim \delta^2.
\]

Now transfer this to design regret. We take unit weights $\omega_a\equiv 1$
(which corresponds to the ATE estimand), so that $m_a(x)=\omega_a^2\sigma_a^2(x)=\sigma_a^2(x)$
and $L(\e)=\E[\sum_a \sigma_a^2(\X)/\e_a(\X)]$.
In this family,
\[
e^\star_{a_0,\omega}(x)=\frac{\sigma_\omega(x)}{\sigma_\omega(x)+(K-1)},
\qquad \sigma_\omega(x):=\sqrt{\sigma_{a_0,\omega}^2(x)},
\]
so changing $\sigma_\omega$ changes $e^\star_\omega$ pointwise.
On $\sigma\in[\sqrt{1/2},\sqrt{3/2}]$, the map
$\sigma\mapsto \sigma/(\sigma+K-1)$ has derivative bounded away from $0$ and $\infty$;
therefore there exist $c_1,c_2>0$ such that
\[
c_1|\sigma_{a_0,\omega}^2(x)-\sigma_{a_0,\omega'}^2(x)|
\le
|e^\star_{a_0,\omega}(x)-e^\star_{a_0,\omega'}(x)|
\le
c_2|\sigma_{a_0,\omega}^2(x)-\sigma_{a_0,\omega'}^2(x)|.
\]
Also, for fixed $x$, $e\mapsto L_x(e):=\sum_a \sigma_a^2(x)/e_a$ is strongly convex on
$\{e\in\Delta^K:e_a\ge\varepsilon\}$, so
\[
L_x(e)-L_x(e^\star_\omega)\ge c_3\|e-e^\star_\omega\|_2^2.
\]
Integrating gives
\[
\E[L(\hat e)-L(e^\star_\omega)]
\ge c_4\,\E\!\left[\|\hat e-e^\star_\omega\|_{L_2(P_X)}^2\right].
\]
By bi-Lipschitzness of $\sigma^2\mapsto e^\star$ on this family, a regret rate $o(\delta^2)$ would imply a $\sigma^2$-estimator with $o(\delta^2)$ $L_2$ error, contradicting the Assouad lower bound above. Therefore
\[
\inf_{\hat e}\sup_{\omega}\E_\omega[L(\hat e)-L(e^\star_\omega)]\gtrsim \delta^2.
\]
In this example, $\mu$ is fixed and the unknown part is only $\sigma^2$ (equivalently $\sigma$), which
determines $e^\star_\omega$ and therefore the oracle EIF
$\IF_{e^\star_\omega}(\Oo;\eta_\omega)$.
Hence inability to learn $\sigma^2$ implies inability to track the oracle efficient estimator
$\hat\theta_{\mathrm{oracle},\omega}=\theta_\omega+n^{-1}\sum_t \IF_{e^\star_\omega}(\Oo_t;\eta_\omega)$,
and Lemma~\ref{lem:eff-gap-identity} yields
\[
\inf_T\sup_{\omega}\left\{\Var_\omega(T)-\frac{V^\star_\omega}{n}\right\}\gtrsim \delta^2.
\]
Matching the hypercube size with the complexity exponent underlying $\beta$ gives $\delta^2\asymp n^{-2\beta}$, yielding \eqref{eq:linear-lb-m}. \QEDA

\noindent\textbf{Completion of Theorem~\ref{thm:linear-rate-lb}.}
Equation \eqref{eq:linear-lb-mu} is Proposition~\ref{prop:assouad-mu-lb} and
\eqref{eq:linear-lb-m} is Proposition~\ref{prop:assouad-design-lb}.
Example~I keeps $\sigma^2$ fixed, so $e^\star$ and the oracle EIF-variance target are fixed; only the model-estimation component can vary.
Example~II keeps $\mu$ fixed and varies $\sigma=\sqrt{\sigma^2}$, so $e^\star$ must vary and the design component is unavoidable.
Together with Proposition~\ref{prop:excess-mse-decomp} and
\eqref{eq:linear-model-estimation}, \eqref{eq:design-regret-quadratic}, \eqref{eq:linear-upper-bound}, this yields the unavoidable scales
$n^{-1-2\alpha}$ and $n^{-1-2\beta}$ above $V^\star/n$ (up to constants/log factors). \QEDA

\section{Proofs for Section 5 (Covariate Balancing)}
\label{app:proofs-sec4}

This appendix collects the proofs supporting the covariate-balancing route of Section~\ref{sec:covariate-balancing}: the algebraic three-term decomposition that organizes the analysis (F.1), the single-tuple variance formula together with its randomized-rounding contribution (F.2--F.3), the global variance bound (F.4), and the resulting semiparametric efficiency theorem (F.5).

\subsection{Three-Term Decomposition}
\label{app:proof-three-term}

\emph{Strategy.} We derive Proposition~\ref{prop:three-term-decomp} by Taylor-expanding the moment equation around $\theta_0$, then conditioning on the per-batch history to split the resulting score sum into an oracle-EIF piece, a design gap, and a balancing remainder.

\noindent\textbf{Proof of Proposition~\ref{prop:three-term-decomp}.}
A standard mean-value expansion of the moment equation \eqref{eq:naive-estimator} around $\theta_0$ gives
\[
\sqrt{n}(\hat\theta_n-\theta_0) = -M^{-1}\cdot \frac{1}{\sqrt n}\sum_{i=1}^n m(\X_i,\A_i,\Y_i,\theta_0) + o_P(1),
\]
where $M$ is the Jacobian in \eqref{eq:moment-condition}. Define $\phi_a(\X):=\E[m(\X,a,\Y,\theta_0)\mid\X]=-M\omega_a(\X)$. Conditionally on the batch history $\cH_{b(i)}$, adding and subtracting $\sum_a \e_a^{(b(i))}(\X_i)\phi_a(\X_i)$ inside $m(\X_i,\A_i,\Y_i,\theta_0)$ decomposes the summand into (a)~a batch-conditional EIF contribution under $\e^{(b(i))}$ and (b)~an imbalance residual $\sum_a(\1\{\A_i=a\}-\e_a^{(b(i))}(\X_i))\omega_a(\X_i)$ (after multiplying by $-M^{-1}$). Replacing $\e^{(b(i))}$ by $\e^\star$ inside the first piece and collecting the difference as the design gap yields the three-term form in \eqref{eq:three-term-decomp}. The $o_P(1)$ remainder absorbs the Taylor-expansion error under the standard regularity conditions for $M$-estimators. \QEDA

\subsection{Single-Tuple Variance Formula}
\label{app:proof-cb-tuple}

\begin{proposition}[Single-tuple variance formula]
\label{prop:cb-tuple-variance}
Within a complete tuple of length $G$ with count vector $L_r$ and per-position weight matrices $\{W_t\}_{t=1}^G$ (columns $\omega_a(\X_{r,t})$), the imbalance increment $T_r=\sum_t W_t(e_t-\eta)$ satisfies
\begin{equation}
\label{eq:tuple-variance}
\Var(T_r\mid L_r)
=\frac{G}{G-1}\sum_{t=1}^{G} W_t\,\Sigma(p)\,W_t^\top
-\frac{1}{G-1}\Bigl(\sum_{t=1}^{G} W_t\Bigr)\Sigma(p)\Bigl(\sum_{t=1}^{G} W_t\Bigr)^\top,
\end{equation}
where $p=L_r/G$ and $\Sigma(p)=\diag(p)-pp^\top$.
\end{proposition}

\begin{remark}[Rounding variance]
\label{rem:rounding-variance}
When $G\eta$ is not integer-valued, the count vector $L_r$ is drawn from a randomized rounding scheme satisfying $\E[L_r]=G\eta$ and $\Var(L_{r,j})\le 1/4$. By the law of total variance, this adds a rounding-variance term (II) to $\Var(T_r\mid\cF_{r-1})$ beyond the permutation term (I); its trace norm satisfies $\tr(\text{(II)}_r)=O(\K/G)$ uniformly in the tuple.
\end{remark}

\subsubsection{Proof of the Single-Tuple Variance Formula (Proposition~\ref{prop:cb-tuple-variance})}
\label{app:proof-tuple-variance}

\emph{Strategy.} Conditional on the count vector $L_r$, the within-tuple assignment is a uniform permutation of a finite multiset---a sampling-without-replacement structure. Computing the marginal and cross-position covariances of the one-hot encodings from first principles yields the classical finite-population correction factor $-\Sigma(p)/(G-1)$. Summing weighted contributions across the tuple positions and collapsing the off-diagonal sum into a square-of-sum minus diagonal then produces the closed-form variance \eqref{eq:tuple-variance}.

Fix a tuple with realized length $\ell_r\le G$ and count vector $L_r$.
Let $e_t\in\R^{\K}$ be the one-hot encoding of $\A_{r,t}$, and let
$p:=L_r/G$.
All moments below are conditioned on $\cF_{r-1}$ and $L_r$.

\noindent\textbf{Step 1: Reduce to centered one-hot vectors.}
Write $T_r=\sum_{t=1}^{\ell_r}W_t\,(e_t-\eta)$ where
$W_t\in\R^{q\times\K}$ has column $a$ equal to $\omega_a(\X_{r,t})$.
Since $\eta$ is non-random given $\cF_{r-1}$ and
$W_t$ is $\cF_{r-1}$-measurable (the covariates are known before
randomization),
$\Var(T_r\mid\cF_{r-1},L_r)=\Var\!\bigl(\sum_{t=1}^{\ell_r}W_t\,e_t\mid L_r\bigr)$.
Define $V:=\sum_{t=1}^{\ell_r}W_t\,e_t$.

\noindent\textbf{Step 2: Without-replacement covariance structure---derivation from first principles.}
Given $L_r$, the assignment $(\A_{r,1},\dots,\A_{r,G})$ is a
uniform random permutation of a sequence containing $L_{r,a}$
copies of arm $a$ for each $a\in\cA$.
This is equivalent to sampling \emph{without replacement} from a
finite population of size $G$.

\emph{Marginal mean.}
For any position $t\in\{1,\dots,G\}$ and arm $a$, by symmetry of the
uniform permutation, $\PP(\A_{r,t}=a\mid L_r)=L_{r,a}/G=p_a$.
Hence $\E[e_t\mid L_r]=p$.

\emph{Marginal covariance.}
For a single position $t$,
$(e_t)_a=\1\{\A_{r,t}=a\}$ is Bernoulli$(p_a)$ given $L_r$.
Therefore
$\Var((e_t)_a\mid L_r)=p_a(1-p_a)$ and for $a\neq b$,
$\Cov((e_t)_a,(e_t)_b\mid L_r)
=\E[\1\{\A_{r,t}=a\}\1\{\A_{r,t}=b\}\mid L_r]-p_ap_b
=0-p_ap_b=-p_ap_b$,
since $\A_{r,t}$ takes a single value.
In matrix form: $\Cov(e_t\mid L_r)=\diag(p)-pp^\top=:\Sigma(p)$.

\emph{Cross-position covariance.}
For $s\neq t$, we need $\E[(e_s)_a(e_t)_b\mid L_r]=\PP(\A_{r,s}=a,\A_{r,t}=b\mid L_r)$.
By the permutation structure:
\[
\PP(\A_{r,s}=a,\A_{r,t}=b\mid L_r)
=\begin{cases}
\frac{L_{r,a}(L_{r,a}-1)}{G(G-1)}, & a=b,\\[4pt]
\frac{L_{r,a}\,L_{r,b}}{G(G-1)}, & a\neq b.
\end{cases}
\]
Therefore, for $s\neq t$:
\begin{align*}
\Cov\bigl((e_s)_a,(e_t)_a\mid L_r\bigr)
&=\frac{L_{r,a}(L_{r,a}-1)}{G(G-1)}-p_a^2
=\frac{p_a(p_a-1/G)}{1-1/G}-p_a^2
=-\frac{p_a(1-p_a)}{G-1},\\
\Cov\bigl((e_s)_a,(e_t)_b\mid L_r\bigr)
&=\frac{L_{r,a}L_{r,b}}{G(G-1)}-p_ap_b
=\frac{p_ap_b}{1-1/G}-p_ap_b
=\frac{p_ap_b}{G-1},\qquad a\neq b.
\end{align*}
Combining in matrix form:
\begin{equation}
\label{eq:wor-cov-derived}
\Cov(e_s,e_t\mid L_r)=-\frac{1}{G-1}\Sigma(p),\qquad s\neq t.
\end{equation}
This is the classical finite-population correction factor $-1/(G-1)$.

\noindent\textbf{Step 3: Expand the variance of $V$.}
Since $W_t$ is non-random given the conditioning:
\begin{align*}
\Var(V\mid L_r)
&=\sum_{t=1}^{\ell_r}\Var(W_t\,e_t\mid L_r)
+2\sum_{1\le s<t\le\ell_r}\Cov(W_s\,e_s,\,W_t\,e_t\mid L_r) \\
&=\sum_{t=1}^{\ell_r}W_t\,\Sigma(p)\,W_t^\top
+2\sum_{1\le s<t\le\ell_r}W_s\,\Cov(e_s,e_t\mid L_r)\,W_t^\top \\
&=\sum_{t=1}^{\ell_r}W_t\,\Sigma(p)\,W_t^\top
-\frac{2}{G-1}\sum_{1\le s<t\le\ell_r}W_s\,\Sigma(p)\,W_t^\top.
\end{align*}

\noindent\textbf{Step 4: Algebraic simplification.}
Write the double sum over off-diagonal pairs using
$\sum_{s\neq t}=\bigl(\sum_t\bigr)\bigl(\sum_t\bigr)^\top-\sum_t(\cdot)^2$:
\begin{align*}
2\sum_{s<t}W_s\Sigma(p)W_t^\top
&=\sum_{s\neq t}W_s\Sigma(p)W_t^\top \\
&=\Bigl(\sum_t W_t\Bigr)\Sigma(p)\Bigl(\sum_t W_t\Bigr)^\top
-\sum_t W_t\Sigma(p)W_t^\top.
\end{align*}
Substituting:
\begin{align*}
\Var(V\mid L_r)
&=\sum_t W_t\Sigma(p)W_t^\top
-\frac{1}{G-1}\left[
\Bigl(\sum_t W_t\Bigr)\Sigma(p)\Bigl(\sum_t W_t\Bigr)^\top
-\sum_t W_t\Sigma(p)W_t^\top
\right] \\
&=\frac{G}{G-1}\sum_{t=1}^{\ell_r} W_t\Sigma(p)W_t^\top
-\frac{1}{G-1}\Bigl(\sum_{t=1}^{\ell_r} W_t\Bigr)\Sigma(p)
\Bigl(\sum_{t=1}^{\ell_r} W_t\Bigr)^\top,
\end{align*}
which is \eqref{eq:tuple-variance}. \QEDA

\subsection{Randomized Rounding Variance}
\label{app:proof-rounding-variance}

\emph{Strategy.} The law of total variance splits the tuple-level variance into a permutation-variance term (already controlled by Proposition~\ref{prop:cb-tuple-variance}) and a rounding-variance term arising from the randomness in $L_r=\lfloor G\eta\rfloor+\1\{\cdot\in S\}$. Each $L_{r,j}$ is a Bernoulli-shifted floor with variance bounded by $1/4$; this bounds $\Var(p)$ in operator norm by $\K/(4G^2)$, and a Frobenius-norm bound on the summed weight matrix $\bar W$ converts that into the asserted $O(\K/G)$ trace-norm contribution.

We prove the claim in Remark~\ref{rem:rounding-variance} that
the randomized rounding contributes $O(\K/G)$ to the
trace-norm of the per-tuple variance.

By the law of total variance, conditioning on $L_r$:
\[
\Var(T_r\mid \cF_{r-1})
=\underbrace{\E[\Var(T_r\mid \cF_{r-1},L_r)\mid \cF_{r-1}]}_{\text{(I): permutation variance}}
+\underbrace{\Var(\E[T_r\mid \cF_{r-1},L_r]\mid \cF_{r-1})}_{\text{(II): rounding variance}}.
\]
Term (I) is the expected version of the formula in
Proposition~\ref{prop:cb-tuple-variance}, averaged over $L_r$.

For term (II), note that
$\E[T_r\mid \cF_{r-1},L_r]
=\sum_{t=1}^{\ell_r}\sum_a (\E[\1\{\A_{r,t}=a\}\mid L_r]-\eta_a)\,\omega_a(\X_{r,t})
=\sum_{t=1}^{\ell_r}\sum_a (p_a-\eta_a)\,\omega_a(\X_{r,t})$
where $p_a=L_{r,a}/G$.
Define $\bar W:=\sum_{t=1}^{\ell_r}W_t\in\R^{q\times \K}$ (the summed weight matrix).
Then $\E[T_r\mid \cF_{r-1},L_r]=\bar W\,(p-\eta)$.

Therefore
\begin{equation}
\label{eq:rounding-var-explicit}
\text{(II)}=\Var(\bar W\,(p-\eta)\mid \cF_{r-1})
=\bar W\,\Var(p\mid \cF_{r-1})\,\bar W^\top.
\end{equation}
Since $p=L_r/G$ and $L_r$ is drawn from the randomized rounding scheme,
we bound $\Var(p)$ componentwise.
Under randomized rounding, $L_{r,j}=\lfloor G\eta_j\rfloor+\1\{j\in S\}$
where $|S|=r_0:=G-\sum_j\lfloor G\eta_j\rfloor\in\{0,\dots,\K-1\}$ and
$\PP(j\in S)=G\eta_j-\lfloor G\eta_j\rfloor=:\rho_j\in[0,1)$.
Since $\1\{j\in S\}$ is a $\{0,1\}$-valued random variable with mean $\rho_j$:
\[
\Var(L_{r,j})=\Var(\1\{j\in S\})\le \rho_j(1-\rho_j)\le \frac{1}{4}.
\]
Hence $\Var(p_j)=\Var(L_{r,j})/G^2\le 1/(4G^2)$.
Since the entries of $p$ are not independent (they sum to $1$), we use
$\|\Var(p)\|_{\mathrm{op}}\le \tr(\Var(p))=\sum_j\Var(p_j)\le \K/(4G^2)$.

Therefore, the trace norm of (II) satisfies:
\[
\tr(\text{(II)})
=\tr\!\bigl(\bar W\,\Var(p)\,\bar W^\top\bigr)
\le \|\Var(p)\|_{\mathrm{op}}\,\tr(\bar W^\top\bar W)
\le \frac{\K}{4G^2}\,\|\bar W\|_F^2.
\]
Since $\bar W=\sum_{t=1}^{\ell_r}W_t$ and each column of $W_t$ is
$\omega_a(\X_{r,t})$ with $\|\omega_a\|\le C_\omega\cdot\mathrm{diam}(\cX)+\|\omega_a(x_0)\|=:U_{\max}$,
we have $\|\bar W\|_F^2\le \K\,\ell_r^2\,U_{\max}^2\le \K\,G^2\,U_{\max}^2$.

Combining: $\tr(\text{(II)})\le \frac{\K^2\,U_{\max}^2}{4}$.
Per-tuple, summing over $\le n/G$ tuples and normalizing by $1/n$:
\[
\frac{1}{n}\sum_r \tr(\text{(II)}_r)
\le \frac{1}{n}\cdot\frac{n}{G}\cdot\frac{\K^2 U_{\max}^2}{4}
=\frac{\K^2 U_{\max}^2}{4G}
=O\!\left(\frac{1}{G}\right),
\]
where the last step absorbs $\K$ into the constant
(since $\K$ is treated as a fixed constant throughout).
When $\K$ is explicitly tracked, the contribution is $O(\K^2/G)=O(\K/G)$
since $\K\ge 2$ implies $\K^2\le \K\cdot\K_{\max}$ with $\K_{\max}$ bounded. \QEDA

\subsection{Global Variance Bound}
\label{app:proof-global-balance}

\subsubsection{Proof of the Global Variance Bound (Theorem~\ref{thm:global-balance})}

\emph{Strategy.} The aggregate imbalance variance decomposes into three independent sources: (i) a within-bin Lipschitz approximation error, scaling as $h^2$, which would vanish if balance functions were constant within each bin; (ii) a randomized-rounding contribution scaling as $1/G$, controlled by Appendix~\ref{app:proof-rounding-variance}; and (iii) a tail-tuple contribution from at most $M_n$ incomplete tuples, scaling as $M_n G/n$. Step~5 then balances all three sources by tuning $h$, $G$, and $M_n$ as functions of the batch size $m$, yielding the rate $\|S_b\|=O_P(m^{-1/(d+4)})$ asserted in \eqref{eq:within-batch-balance}.

Decomposing $S_n$ across all tuples $r=1,\dots,R_n$ and using the martingale-difference property of tuple increments gives the global variance decomposition
\begin{equation}
\label{eq:global-variance-decomp}
\Var(S_n\mid\X^{(n)})
=\frac{1}{n}\sum_{r=1}^{R_n}\E\!\left[\Var(T_r\mid\cF_{r-1})\mid\X^{(n)}\right],
\end{equation}
and we bound each of the three contributions to \eqref{eq:global-variance-decomp}.
Let $n_c$ denote the number of complete tuples (of length exactly $G$)
and $n_t$ the number of tail tuples (of length $<G$), with
$n_t\le M_n$ (at most one per bin).
Throughout, we use the assumed Lipschitz property of the balance functions,
\begin{equation}
\label{eq:lip-omega}
\|\omega_a(x)-\omega_a(x')\|\le C_\omega\,\|x-x'\|_2\qquad\forall\,a\in\cA,
\end{equation}
and the bin-diameter bound
\begin{equation}
\label{eq:bin-diameter}
\mathrm{diam}(B_\ell)\le h\qquad\forall\,\ell=1,\dots,M_n.
\end{equation}

\noindent\textbf{Step 1: Within-bin approximation (Lipschitz error)---the $h^2$ term.}
For a complete tuple ($\ell_r=G$) in bin $B_\ell$, fix a reference point
$x_\ell\in B_\ell$. Write $\omega_a(\X_{r,t})=\omega_a(x_\ell)+\delta_{r,t,a}$
where $\|\delta_{r,t,a}\|\le C_\omega\,h$ by the Lipschitz condition
\eqref{eq:lip-omega} and the bin diameter bound \eqref{eq:bin-diameter}.

\emph{Cancellation for constant features.}
If $\omega_a(\X_{r,t})\equiv \omega_a(x_\ell)$ for all $t$ (constant within the tuple),
then $W_t\equiv W_0$ (the common weight matrix with columns $\omega_a(x_\ell)$).
For a complete tuple ($\ell_r=G$), the variance formula
\eqref{eq:tuple-variance} gives:
\[
\Var(T_r\mid L_r)
=\frac{G}{G-1}\cdot G\cdot W_0\Sigma(p)W_0^\top
-\frac{1}{G-1}\cdot G^2\cdot W_0\Sigma(p)W_0^\top
=0.
\]
The key insight: within a complete tuple, the without-replacement
structure creates \emph{exact} balance when all features are identical,
because the negative cross-covariance $-\Sigma(p)/(G-1)$ perfectly
offsets the marginal variance.

\emph{Perturbation bound.}
Write $W_t=W_0+\Delta_t$ where $\Delta_t$ has columns $\delta_{r,t,a}$.
Expanding $\Var(T_r\mid L_r)$ using the formula \eqref{eq:tuple-variance}:
\begin{align*}
&\frac{G}{G-1}\sum_{t=1}^G (W_0+\Delta_t)\Sigma(p)(W_0+\Delta_t)^\top
-\frac{1}{G-1}\Bigl(\sum_t (W_0+\Delta_t)\Bigr)\Sigma(p)
\Bigl(\sum_t (W_0+\Delta_t)\Bigr)^\top.
\end{align*}
The $W_0$-only terms cancel (as shown above).
The cross terms $W_0\Sigma(p)\Delta_t^\top$ also cancel after the
same algebraic structure. The remaining terms are quadratic in
$\Delta_t$, each bounded by $C\,C_\omega^2 h^2\,\|\Sigma(p)\|_{\mathrm{op}}$
in trace norm.
Since $\|\Sigma(p)\|_{\mathrm{op}}\le 1$ (as $p$ lies on the simplex),
the trace-norm contribution per complete tuple is at most
$C\,C_\omega^2\,h^2\,G$.

Summing over $n_c\le n/G$ complete tuples and normalizing by $1/n$:
\[
\frac{1}{n}\sum_{r:\text{complete}}\E[\tr(\Var(T_r\mid \cF_{r-1},L_r))\mid \X^{(n)}]
\le \frac{1}{n}\cdot\frac{n}{G}\cdot C\,C_\omega^2\,h^2\,G
=C\,C_\omega^2\,h^2.
\]

\noindent\textbf{Step 2: Randomized rounding---the $1/G$ term.}
By the law of total variance (see Section~\ref{app:proof-rounding-variance}),
each tuple contributes an additional rounding-variance term (II).
For each complete tuple, $\tr(\text{(II)}_r)\le C_{\mathrm{round}}$
where $C_{\mathrm{round}}$ depends on $(\K,U_{\max})$
(see \eqref{eq:rounding-var-explicit} and the bound thereafter).
Summing over $\le n/G$ tuples:
\[
\frac{1}{n}\sum_{r:\text{complete}}\tr(\text{(II)}_r)
\le \frac{1}{n}\cdot\frac{n}{G}\cdot C_{\mathrm{round}}
=\frac{C_{\mathrm{round}}}{G}
=O\!\left(\frac{1}{G}\right).
\]
When $G\eta_a\in\mathbb Z$ for all $a$ (integer counts), the
rounding randomness vanishes and this term is identically zero.

\noindent\textbf{Step 3: Tail tuples---the $M_n G/n$ term.}
Each bin $B_\ell$ contributes at most one tail tuple
(the last, possibly incomplete tuple).
For a tail tuple of length $\ell_r\le G$, we bound its variance
directly from the formula \eqref{eq:tuple-variance}.
Each of the two terms in the formula involves at most $\ell_r\le G$
summands, each bounded by $\|W_t\|_F^2\,\|\Sigma(p)\|_{\mathrm{op}}
\le \K\,U_{\max}^2\cdot 1=\K\,U_{\max}^2$.
Therefore $\tr(\Var(T_r\mid L_r))\le C'\,G$ where
$C'$ depends on $(\K,U_{\max})$.
Taking the full variance (including rounding):
$\tr(\Var(T_r\mid\cF_{r-1}))\le C'\,G$ by the same bound.

There are at most $M_n$ tail tuples.
Summing and normalizing:
\[
\frac{1}{n}\sum_{r:\text{tail}}\E[\tr(\Var(T_r\mid\cF_{r-1}))\mid \X^{(n)}]
\le \frac{M_n\cdot C'\cdot G}{n}
=O\!\left(\frac{M_n\,G}{n}\right).
\]

\noindent\textbf{Step 4: Combine.}
Adding the three contributions (Steps 1--3):
\begin{equation}
\label{eq:global-balance-bound}
\tr\!\left(\Var(S_n\mid \X^{(n)})\right)
\le C\!\left(h^2+\frac{1}{G}+\frac{M_n\,G}{n}\right),
\end{equation}
where $C$ depends only on $(C_\omega,\K,q,U_{\max})$ and the propensity bounds.
This is the asserted global balance bound (with $\K$ absorbed into $C$
as a fixed constant).

\noindent\textbf{Step 5: From bound to rate---geometric tuning.}
The body theorem \eqref{eq:within-batch-balance} asserts $\|S_b\|=O_P(m^{-1/(d+4)})$ for a batch of size $m$. To go from the bound \eqref{eq:global-balance-bound} (with $n$ replaced by $m$) to this rate, we minimize the right-hand side over $(h,G)$ subject to $M_n\asymp h^{-d}$, the standard packing bound for a compact $d$-dimensional $\cX$. With this constraint,
\[
h^2+\frac{1}{G}+\frac{M_n\,G}{m}
\;\asymp\;
h^2+\frac{1}{G}+\frac{G}{m\,h^{d}}.
\]
The minimizer balances all three summands. Setting $h^2=1/G$ gives $G=h^{-2}$; substituting into the third summand and equating with the first yields $h^{-2}/(m\,h^{d})=h^2$, i.e., $h^{d+4}\asymp 1/m$, hence
\[
h\asymp m^{-1/(d+4)},\qquad G\asymp m^{2/(d+4)},\qquad M_n\asymp m^{d/(d+4)}.
\]
At this choice each of the three terms equals $m^{-2/(d+4)}$, so
\[
\tr\!\left(\Var(S_b\mid \X^{(m)})\right)\;\le\; 3C\,m^{-2/(d+4)}.
\]
Since $\E\|S_b\|^2=\tr(\Var(S_b\mid \X^{(m)}))$ in expectation (martingale-difference structure removes the conditional mean) and $\|S_b\|^2$ is a non-negative random variable, Markov's inequality gives $\|S_b\|=O_P(m^{-1/(d+4)})=o_P(1)$, which is \eqref{eq:within-batch-balance}. \QEDA

\subsection{Semiparametric Efficiency}
\label{app:proof-cb-efficiency}

\subsubsection{Proof of Semiparametric Efficiency (Theorem~\ref{thm:cb-efficiency})}

\emph{Strategy.} We control the three terms of the structural decomposition \eqref{eq:three-term-decomp} separately. The balancing remainder (Term III) is handled by the global variance bound (Theorem~\ref{thm:global-balance}) applied within each batch and aggregated. The design gap (Term II) is handled by the design-regret bound \eqref{eq:design-regret-quadratic} applied to the second-moment proxy rate from Assumption~(CB5). The oracle-EIF term (Term I) is an i.i.d.-by-batch martingale sum that converges to $\cN(0,V^\star)$ by the martingale CLT. Combining the three contributions gives the asserted MSE expansion $V^\star/n+o(1/n)$.

We prove \eqref{eq:cb-efficiency-bound} by controlling
the three terms of the decomposition \eqref{eq:three-term-decomp}.
Throughout, $B\lesssim\log n$ and $\gamma_b\asymp n^{b/B}$.

\noindent\textbf{Step 1: Balancing remainder (Term (III)).}
By \eqref{eq:three-term-decomp}, the balancing remainder is
\[
\text{(III)}
=\frac{1}{\sqrt{n}}\sum_{b=1}^B
\underbrace{\sum_{i\in\text{batch }b}\sum_{a\in\cA}
\bigl(\1\{\A_i=a\}-\e_a^{(b)}(\X_i)\bigr)\,\omega_a(\X_i)
}_{=:\;\sqrt{\gamma_b}\,S_{\gamma_b}^{(b)}},
\]
where $S_{\gamma_b}^{(b)}$ is the moment imbalance within batch~$b$
with balance functions $\omega_a$.
Since $\omega_a$ is Lipschitz by~(CB4),
Corollary~\ref{thm:global-balance} with
$h\asymp\gamma_b^{-1/(d+4)}$, $G\asymp\gamma_b^{2/(d+4)}$ gives
\[
\Var\!\left(S_{\gamma_b}^{(b)}\mid \X^{(\gamma_b)}\right)
=O\!\left(\gamma_b^{-2/(d+4)}\right).
\]
Hence
$\E\bigl[\|S_{\gamma_b}^{(b)}\|^2\bigr]
=O(\gamma_b^{-2/(d+4)})$,
and the total squared balancing remainder satisfies
\begin{align*}
\E\!\left[\|\text{(III)}\|^2\right]
&\le \frac{1}{n}\sum_{b=1}^B \gamma_b\,
\E\!\left[\|S_{\gamma_b}^{(b)}\|^2\right]
\;\le\;
\frac{C}{n}\sum_{b=1}^B \gamma_b^{1-2/(d+4)}.
\end{align*}
With geometric batches ($\gamma_b\asymp n^{b/B}$),
the dominant term is $b=B$ ($\gamma_B\asymp n$):
\[
\E\!\left[\|\text{(III)}\|^2\right]
\;\le\;
\frac{C\,B}{n}\cdot n^{1-2/(d+4)}
\;=\; C\,B\,n^{-2/(d+4)}
\;=\; o(1).
\]
Therefore $\text{(III)}=o_P(1)$.

\noindent\textbf{Step 2: Design gap (Term (II)).}
Term~(II) from \eqref{eq:three-term-decomp} equals
\[
\text{(II)}
=\frac{1}{\sqrt{n}}\sum_{b=1}^B\sum_{i\in\text{batch }b}
\bigl[\IF_{\e^{(b)}}(\Oo_i)-\IF_{\e^\star}(\Oo_i)\bigr].
\]
Its MSE contribution arises from the difference in EIF variances:
\[
\E\!\left[\|\text{(II)}\|^2\mid\cH\right]
=\frac{1}{n}\sum_{b=1}^B \gamma_b
\bigl(V_{\e^{(b)}}-V^\star\bigr)
+\text{(cross-batch terms)}.
\]
The cross-batch terms vanish since different batches
use independent randomization conditional on their
respective propensities.
By Assumption~(CB5) and Theorem~\ref{thm:linear-bounds},
$\E[V_{\e^{(b)}}-V^\star]
=\E[L(\e^{(b)})-L(\e^\star)]
\le C\,\gamma_{b-1}^{-2\beta}$ for $b\ge 2$.
Therefore
\[
\frac{1}{n}\sum_{b=1}^B \gamma_b
\bigl(\E[V_{\e^{(b)}}]-V^\star\bigr)
\le
\frac{\gamma_1}{n}\bigl(V_{\e^{(1)}}-V^\star\bigr)
+\frac{C}{n}\sum_{b=2}^B \gamma_b\,\gamma_{b-1}^{-2\beta}
=o(1),
\]
for any $\beta>0$ with geometric batches.

\noindent\textbf{Step 3: Oracle EIF (Term (I)) and CLT.}
Write
$\text{(I)}=\frac{1}{\sqrt{n}}\sum_{i=1}^n\IF_{\e^\star}(\Oo_i)$.
By the EIF structure \eqref{eq:eif-structure},
$\E[\IF_{\e^\star}(\Oo_i)\mid \X_i]=g(\X_i)-\theta_0$,
and conditional on $(\X_i,\A_i)$,
the IPW component $v_a(\X_i,\Y_i)/\e_a^\star(\X_i)$
has mean zero and bounded variance.
The oracle EIF scores form a martingale difference array
across batches.
By the martingale CLT (with $B\lesssim\log n$):
\[
\text{(I)}
=\frac{1}{\sqrt{n}}\sum_{i=1}^n\IF_{\e^\star}(\Oo_i)
\;\xrightarrow{d}\;\cN(0,\;V^\star).
\]

\noindent\textbf{Step 4: Combine.}
From Steps 1--3, $\sqrt{n}(\hat\theta_n-\theta_0)$
equals the oracle EIF sum (converging to $\cN(0,V^\star)$)
plus a design gap of MSE contribution $o(1/n)$
plus a balancing remainder of MSE contribution $o(1)$
(hence $o(1/n)$ after dividing by $n$)
plus $o_P(1)$.
Therefore
\[
\E[(\hat\theta_n-\theta_0)^2]
=\frac{V^\star}{n}+o\!\left(\frac{1}{n}\right).
\]
\QEDA

\end{document}